\documentclass[pdflatex,sn-mathphys-num]{sn-jnl}

\usepackage{graphicx}%
\usepackage{multirow}%
\usepackage{amsmath,amssymb,amsfonts}%
\usepackage{amsthm}%
\usepackage{mathrsfs}%
\usepackage{xcolor}%
\usepackage{textcomp}%
\usepackage{manyfoot}%
\usepackage{booktabs}%
\usepackage{algorithm}%
\usepackage{algorithmicx}%
\usepackage{algpseudocode}%
\usepackage{listings}%

\usepackage{graphicx} 
\usepackage{xcolor}
\usepackage{subfig}
\usepackage[inkscapelatex=false]{svg}

\usepackage{tabularx}
\usepackage{comment}

\usepackage{mathtools}

\newtheorem{theorem}{Theorem}
\newtheorem{lemma}{Lemma}

\newtheorem{proposition}{Proposition}
\newtheorem{remark}{Remark}

\usepackage[colorinlistoftodos,prependcaption]{todonotes}

\newcommand{\intd}{\textrm{d}}

\DeclarePairedDelimiter\floor{\lfloor}{\rfloor}

\theoremstyle{thmstyletwo}%

\theoremstyle{thmstylethree}%

\begin{document}

\title[Bayesian Wasserstein Repulsive Gaussian Mixture Models]{Bayesian Wasserstein Repulsive Gaussian Mixture Models}


\author[1]{\fnm{Weipeng} \sur{Huang}}\email{weipenghuang@sziit.edu.cn}
\equalcont{These authors contributed equally to this work.}

\author*[2]{\fnm{Tin Lok James} \sur{Ng}}\email{ngja@tcd.ie}
\equalcont{These authors contributed equally to this work.}

\affil[1]{\orgdiv{School of Computer Science and Software Engineering}, \orgname{Shenzhen Institute of Information Technology}, \orgaddress{\city{Shenzhen}, \state{Guangdong}, \country{China}}}

\affil*[2]{\orgdiv{School of Computer Science and Statistics}, \orgname{Trinity College Dublin}, \orgaddress{\city{Dublin},  \country{Ireland}}}


\abstract{We develop the Bayesian Wasserstein repulsive Gaussian mixture model that promotes well-separated clusters. Unlike existing repulsive mixture approaches that focus on separating the component means, our method encourages separation between mixture components based on the Wasserstein distance. We establish posterior contraction rates within the framework of nonparametric density estimation. Posterior sampling is performed using a blocked-collapsed Gibbs sampler. Through simulation studies and real data applications, we demonstrate the effectiveness of the proposed model.}

\keywords{Repulsive Mixture Model, Wasserstein Metric, Posterior Contraction Rate}



\maketitle

\section{Introduction}
Mixture models have been widely and successfully applied across various domains in Bayesian modeling~\citep{wade2023bayesian}. Independent priors on component-specific parameters are frequently used due to their flexibility and computational ease. However, a key limitation of this approach is the risk of introducing redundant components, which is particularly concerning in cases where a more parsimonious representation is desired, as excess components can hinder interpretability. To address this issue, repulsive mixture models have been introduced in the literature \citep{Petralia2012, Xu2016, Xie2020, Quinlan2021, Sun2022, Cremaschi2024}. These models employ specially designed priors on the component parameters to encourage separation between components. Rather than sampling component parameters independently from a base measure, the prior incorporates a mechanism that penalizes components that are too close to one another. A closely related approach in mixture modeling is the non-local prior method \citep{fuquene2019choosing}, which not only discourages component parameters from being too similar but also penalizes small component probabilities. Another fundamentally different approach to encouraging separation between mixture components is to impose constraints directly on the posterior distribution \citep{Chen2014, Huang2023}. This method builds upon the framework of Bayesian regularization, also known as posterior regularization \citep{Ganchev2007, Ganchev2010, Zhu2014}, which incorporates structured constraints into the posterior to guide the learning process.
\\\\
In the context of location-scale mixture models, the aforementioned approaches promote separation between mixture components by penalizing cases where the location parameters of the components are too close to each other. In other words, these methods focus solely on the distances between the location parameters when formulating the model, leading to estimated mixtures where the means of the components are well-separated. While this type of separation is appropriate for many applications, there are scenarios where it may be more meaningful to consider the distance between the entire distributions of the components rather than just their location parameters. In this work, we extend existing repulsive mixture models by developing prior distributions that discourage components from being too close to one another, where closeness is measured based on the entire distributions rather than solely their mean parameters. 
\\\\
To quantify the distance between probability distributions, a suitable metric on the space of probability measures is required. In this work, we utilize the Wasserstein metric to assess the closeness of distributions. The Wasserstein metric offers appealing properties. Notably, it has an intuitive geometric interpretation, as it represents the minimal cost required to optimally ``transport'' probability mass from one distribution to another. Unlike Kullback-Leibler (KL) divergence and Jensen-Shannon (JS) divergence, which can become infinite or undefined when distributions have disjoint support, the Wasserstein distance remains well-defined, enabling meaningful comparisons even in such cases. In recent years, the Wasserstein metric has become increasingly popular for measuring distances between probability distributions in various applications. This metric has been employed in various statistical and machine learning tasks, including principal component analysis of probability distributions \citep{Bigot2017, Pegoraro2022}, clustering of probability distributions \citep{Verdinelli2019, Zhuang2022}, multilevel clustering \citep{Huynh2021}, and distribution-on-distribution regression \citep{Ghodrati2022, Chen2023, Ng2024}. 
\\\\
In this work, we focus on Gaussian mixture models and extend the repulsive mixture modeling framework of \cite{Xie2020} by replacing mean-based repulsion with Wasserstein-based repulsion. Our model penalizes mixture components that are too close based on the Wasserstein distance. We focus on Gaussian mixtures because the Wasserstein distance between two Gaussian distributions has a closed-form expression, making the approach computationally efficient. However, the idea of incorporating the Wasserstein metric for repulsion can be applied to other families of mixture models, with numerical approximations of the Wasserstein distance when closed-form solutions are unavailable. Additionally, while our work extends the repulsive mixture model developed in \cite{Xie2020}, the proposed approach is broadly applicable to other classes of repulsive mixture models discussed earlier.
\\\\The remainder of this manuscript is organized as follows. Section \ref{sec_background} provides background on the Wasserstein distance, with a particular focus on its formulation for Gaussian distributions. In Section \ref{sec_method}, we introduce the proposed model, referred to as the Bayesian Wasserstein Repulsive Gaussian Mixture Model (WRGM). We analyze the posterior contraction rate for density estimation under the proposed model, extending the theoretical results of \cite{Xie2020}, and adapt the blocked-collapsed Gibbs sampler from \cite{Xie2020} for posterior sampling in our setting. Section \ref{sec_simulation} presents simulation studies to examine the empirical behavior of the proposed approach, while Section \ref{sec_data_app} demonstrates its application to real-world datasets. Finally, Section \ref{sec_discussion} concludes the manuscript with a discussion of key findings and potential future directions.

\section{Background and Notations}
\label{sec_background}

\subsection{Background on Wasserstein Space}
Let ${\cal P}(\mathbb{R}^p)$ denote the space of probability measures on $\mathbb{R}^p$, and ${\cal P}_2(\mathbb{R}^p)$ be the space of probability measures on $\mathbb{R}^p$ with finite second moments. Consider two probability measures $\nu_0, \nu_1 \in {\cal P}(\mathbb{R}^p)$, a map $T: \mathbb{R}^d \rightarrow \mathbb{R}^d$ is said to be a transport map from $\nu_0$ to $\nu_1$ if $T\#\nu_0 = \nu_1$, that is, $\nu_1$ is the push-forward measure of $\nu_0$ under $\nu_1$, meaning $ \nu_1(B) = \nu_0(T^{-1}(B))$
for all Borel set $B$ in $\mathbb{R}^p$. Let $c(\cdot, \cdot): \mathbb{R}^p \times \mathbb{R}^p \rightarrow \mathbb{R}_+$ be a cost function where $c(x,y)$ represent the cost of transporting a unit mass from $x$ to $y$. The Monge problem \citep{Monge1781} consists of finding the transport map which minimizes the total transportation costs:
\begin{eqnarray}
\label{eqn_monge_problem}
    \inf_{T: T\#\nu_0 = \nu_1} \int_{\mathbb{R}^p} c(x, T(x)) \nu_0(\intd x) .
\end{eqnarray}
A transport map $T$ which achieves the infinum above is called an optimal transport map. Since solution to the Monge problem \eqref{eqn_monge_problem} may not exist in general, \citep{Kantorovich1958} proposed a relaxation of \eqref{eqn_monge_problem} and instead solves the following minimization problem:
\begin{eqnarray}
\label{eqn_kantorovich_problem}
    \inf_{\gamma \in \Gamma(\nu_0, \nu_1)} \int_{\mathbb{R}^p \times \mathbb{R}^p} c(x,y) \intd \gamma(x, y),
\end{eqnarray}
where $\Gamma(\nu_0, \nu_1)$ is the set of couplings of $\nu_0, \nu_1$. That is, $\gamma \in \Gamma(\nu_0, \nu_1)$ if $\gamma$ is a joint distribution with marginals $\nu_0, \nu_1$. Any $\gamma \in \Gamma(\nu_0, \nu_1)$ that attains the infimum in \eqref{eqn_kantorovich_problem} is called an optimal transport plan. Since every transport map can be associated with a transport plan of the same cost, one may realize that solutions to the Monge and Kantorovich problems are related by 
\begin{eqnarray}
\label{eqn_monge_kantorovich_rel}
    \inf_{\gamma \in \Gamma(\nu_0, \nu_1)} \int_{\mathbb{R}^d \times \mathbb{R}^d} c(x,y) \intd \gamma(x, y) \le \inf_{T: T\#\nu_0 = \nu_1} \int_{\mathbb{R}^d} c(x, T(x)) \nu_0(\intd x) . 
\end{eqnarray}
In this work, we concentrate on the most commonly used squared distance cost function: $ c(x,y) = ||x-y||^2$. With this cost function, the squared root of \eqref{eqn_kantorovich_problem} defines a metric on ${\cal P}_2(\mathbb{R}^p)$, known as the 2-Wasserstein metric:
\begin{eqnarray}
    W_2(\nu_0,\nu_1) =   \bigg( \inf_{\gamma \in \Gamma(\nu_0, \nu_1)} \int_{\mathbb{R}^p \times \mathbb{R}^p} ||x-y||^2 \intd \gamma(x, y) \bigg)^{1/2}.
\end{eqnarray}
When $\nu_0 \in {\cal P}_2(\mathbb{R}^p)$ is absolutely continuous, and $c$ is the squared distance cost function, it is well known that a unique solution to the Monge problem exists \citep{Brenier1991} and the equality in \eqref{eqn_monge_kantorovich_rel} is attained. Let $T$ be this unique optimal transport map, the corresponding optimal transport plan which minimizes \eqref{eqn_kantorovich_problem} is given by $\gamma = (\mbox{Id}, T)\#\nu_0$ where $\mbox{Id}: \mathbb{R}^p \rightarrow \mathbb{R}^p$ is the identity mapping. 
\\\\
Closed-form solutions to the optimal transport problems exist when both measures $\nu_0, \nu_1$ are Gaussian distributions. Denote the mean and covariance of $\nu_i$ by $m_i$ and $\Sigma_i$, respectively, for $i=0,1$. In this case, the squared 2-Wasserstein distance between $\nu_0$ and $\nu_1$ is given by \citep{Panaretos2020}[Chapter 1.6.3]
\begin{eqnarray}
\label{eqn_wasserstein_gaussian}
    W_2^2(\nu_0, \nu_1) = ||m_1 - m_0||^2 + \mbox{tr}\big( \Sigma_0 + \Sigma_1 - 2(\Sigma_0^{1/2} \Sigma_1 \Sigma_0^{1/2})^{1/2} \big) ,
\end{eqnarray}
where $\mbox{tr}(\cdot)$ denotes the trace of a matrix. The first term on the right-hand side (RHS) of \eqref{eqn_wasserstein_gaussian} represents the squared Euclidean distance between the mean vectors of the two Gaussian distributions. The second term corresponds to the squared Bures-Wasserstein distance, given by
\begin{equation}
    D^2_{\Sigma}(\Sigma_0, \Sigma_1) = \text{tr} \big( \Sigma_0 + \Sigma_1 - 2(\Sigma_0^{1/2} \Sigma_1 \Sigma_0^{1/2})^{1/2} \big),
\end{equation}
where $D^2_{\Sigma}$ defines a metric on the space of positive definite covariance matrices. Notably, the Bures-Wasserstein distance corresponds to the 2-Wasserstein distance between the Gaussian distributions \( N(0, \Sigma_0) \) and \( N(0, \Sigma_1) \). Thus, we obtain a decomposition of the 2-Wasserstein distance for Gaussian distributions into one term quantifying the distance between mean vectors and another term capturing the discrepancy between covariance matrices. In the special case where $\Sigma_0$ and $\Sigma_1$ commute, i.e., $\Sigma_0 \Sigma_1 = \Sigma_1 \Sigma_0$, \eqref{eqn_wasserstein_gaussian} can be simplified as:
\begin{eqnarray}
\label{eqn_wasserstein_gaussian_commute}
    W_2^2(\nu_0, \nu_1) = ||m_1 - m_0||^2 + ||\Sigma_0^{1/2} - \Sigma_1^{1/2} ||_F^{2} ,
\end{eqnarray}
where $||\cdot||_F$ is the Frobenius norm. 

\subsection{Notations}
The notations employed in our work align with those in \cite{Xie2020}. We let \(\mathcal{S} \subset \mathbb{R}^{p \times p}\) denote a collection of positive definite matrices. The space of probability distributions over \(\mathbb{R}^p \times \mathcal{S}\) is denoted by \(\mathcal{M}(\mathbb{R}^p \times \mathcal{S})\), while \(\mathcal{M}(\mathbb{R}^p)\) represents the space of probability distributions over \(\mathbb{R}^p\). The probability simplex is denoted as  
\[
\Delta^K = \{(w_1, \ldots, w_K) \mid w_k \geq 0, \sum_{k=1}^{K} w_k = 1 \}.
\]
For a positive definite matrix \(\Sigma\), we use \(\lambda(\Sigma)\) to refer to any of its eigenvalues, and denote the largest and smallest eigenvalues by \(\lambda_{\max}(\Sigma)\) and \(\lambda_{\min}(\Sigma)\), respectively. Denote $\| \cdot \|$ the Euclidean norm on $\mathbb{R}^p$. We use $\| \cdot \|_1$ to denote both the $L_1$-norm on $L_1(\mathbb{R}^p)$ and the 1-norm on finite-dimensional Euclidean space $\mathbb{R}^d$ for any $d \geq 1$. $\| \cdot \|_\infty$ is used to denote both the $\ell_\infty$-norm of a vector and the supremum norm of a bounded function. Given a set ${\cal A}$ in a metric space with metric $d$, the $\epsilon$-covering number of ${\cal A}$ with respect to $d$, denoted as ${\cal N}(\epsilon, {\cal A}, d)$ is defined to be the minimum number of $\epsilon$ balls of the form $\{ a \in {\cal A}: d(a, a_0) < \epsilon\}$ 
that are needed to cover ${\cal A}$.
\\\\
Furthermore, we let \(\mathbb{P}_0\) and \(\mathbb{E}_0\) represent probability and expectation under the true data-generating density \(f_0\). The prior distribution on \(\mathcal{M}(\mathbb{R}^p)\) is denoted by \(\Pi\), with the corresponding posterior distribution given observed data \(y_1, \dots, y_n\) written as \(\Pi(\cdot \mid y_1, \dots, y_n)\).

\section{Methodology}
\label{sec_method}
\subsection{Model Specification}
\label{subsec_model}
Our model formulation closely follows that of \cite{Xie2020}, with the key distinction being that in \cite{Xie2020}, the repulsive prior specification only incorporates the mean vectors of the components, while in our approach, the repulsive prior includes both the mean vectors and the covariances of the components. The significant implication of this difference is that, while the mean vectors and covariances of the mixture components are a priori independent in \cite{Xie2020}, they are a priori dependent in our formulation.
\\\\
 We consider the Gaussian mixture model with density
$$ f_F(y) = \int_{\mathbb{R}^p \times {\cal S}} \phi(y|m, \Sigma) dF(m, \Sigma) ,$$
where 
$$ \phi(y|m, \Sigma) = \mbox{det}(2\pi \Sigma)^{-\frac{1}{2}} \exp\Big( -\frac{1}{2} (y-m)^{T} \Sigma^{-1}(y-m) \Big)$$
is the density function of $p$-dimensional Gaussian distribution $N(m, \Sigma)$, and $F \in {\cal M}(\mathbb{R}^p \times {\cal S})$ is a distribution on $\mathbb{R}^{p} \times {\cal S}$.  We define a prior $\Pi$ over the space of probability densities on $\mathbb{R}^p$ by the following hierarchical model:
\begin{eqnarray}
\label{eqn_gen_model}
(f(y)|F) = \int_{\mathbb{R}^p \times {\cal S}} \phi(y|m, \Sigma) dF(m, \Sigma), \nonumber \\
 (F|K, \{w_k, m_k, \Sigma_k\}_{k=1}^{K} ) = \sum_{k=1}^{K} w_k \delta_{(m_k, \Sigma_k)} , \\     
 (m_1, \Sigma_1, \ldots, m_K, \Sigma_k|K) \sim p(m_1, \Sigma_1, \ldots, m_K, \Sigma_K|K) , \nonumber \\
 (w_1, \ldots, w_K|K) \sim {\cal D}_K(\beta). \nonumber
\end{eqnarray}

Here $p(m_1, \Sigma_1, \ldots, m_K, \Sigma_K|K)$ is a certain density function with respect to the Lebesgue measure on $(\mathbb{R}^p \times {\cal S})^{K}$, ${\cal D}_K(\beta)$ is the Dirichlet distribution on $\Delta^K$. The number of mixture components, $K$, is assigned a prior
\begin{eqnarray}
\label{eqn_prior_K}
 K \sim p_K(K), \quad K \in \mathbb{N}_+.   
\end{eqnarray}
Instead of assuming $(m_k, \Sigma_k)_{k=1}^{K}$ being i.i.d. from a base measure, we introduce repulsion among components $N(m_k, \Sigma_k)$ such that they are well separated. Unlike previous works where repulsion is enforced through the distance between mean vectors $m_k, m_{k'}, k \ne k' $ between components, we instead enforce separation using Wasserstein distance between normal components. We assume that the density is of the form
\begin{eqnarray}
\label{eqn_repulsive_prior}
 p(m_1, \Sigma_1, \ldots, m_K, \Sigma_K|K) = \frac{1}{Z_K} \bigg( \prod_{k=1}^{K} p_{m}(m_k) p_{\Sigma}(\Sigma_k) \bigg) h_K((m_1, \Sigma_1), \ldots, (m_K, \Sigma_K)) ,    
\end{eqnarray}
where $p_m$ and $p_{\Sigma}$ are the prior distributions for component mean and covariance, respectively, and
\begin{eqnarray}
 Z_K &=   \int_{{\cal S}} \cdots \int_{{\cal S}} \int_{\mathbb{R}^p} \cdots \int_{\mathbb{R}^p}& \bigg( \prod_{k=1}^{K} p_{m}(m_k) p_{\Sigma}(\Sigma_k) \bigg) h_K((m_1, \Sigma_1), \ldots, (m_K, \Sigma_K)) \nonumber \\
  && d \rho_{\mathbb{R}^p}(m_1) \cdots d \rho_{\mathbb{R}^p}(m_K) d \rho_{{\cal S}}(\Sigma_1) \cdots d \rho_{{\cal S}}(\Sigma_K),    
\end{eqnarray}
is the normalizing constant. Here, $d \rho_{\mathbb{R}^p}(m)$ is the Lebesgue measure on $\mathbb{R}^p$ and $d \rho_{{\cal S}}(\Sigma)$ is the Lebesgue measure on ${\cal S}$. For notation simplicity, we write $Z_K$ as:
\begin{eqnarray}
\label{eqn_ZK}
 Z_K &=   \int_{{\cal S}} \cdots \int_{{\cal S}} \int_{\mathbb{R}^p} \cdots \int_{\mathbb{R}^p} & \bigg( \prod_{k=1}^{K} p_{m}(m_k) p_{\Sigma}(\Sigma_k) \bigg) h_K((m_1, \Sigma_1), \ldots, (m_K, \Sigma_K)) \nonumber \\
  && d m_1 \cdots d m_K d \Sigma_1 \cdots d \Sigma_K.
\end{eqnarray}
The function $h_K: (\mathbb{R}^p \times {\cal S})^K \rightarrow [0,1]$ is required to be invariant under permutation of arguments. That is, for any permutation $\zeta: \{1, \ldots, K\} \rightarrow \{1, \ldots, K\}$, the following equality holds:
$$ h_K((m_1, \Sigma_1), \ldots, (m_K, \Sigma_K)) = h_K((m_{\xi(1)}, \Sigma_{\xi(1)}), \ldots, (m_{\xi(K)}, \Sigma_{\xi(K)})) .$$
We require $h_K$ satisfies the following repulsive condition:
$ h_K((m_1, \Sigma_1), \ldots, (m_K, \Sigma_K)) = 0$
if and only if $m_k = m_{k'}, \Sigma_k = \Sigma_{k'}$ for some $k \ne k'$.
\\\\
Let $g: \mathbb{R}_{+} \rightarrow [0,1]$ be a strictly monotonically increasing function. We employ the two types of repulsion functions introduced by \cite{Xie2020}:

\begin{eqnarray}
\label{eqn_repulsive_fun1}
 h_K((m_1, \Sigma_1), \ldots, (m_K, \Sigma_K)) = \min_{1 \le k < k' \le K} g( W^2_2(N(m_k, \Sigma_k), N(m_{k'}, \Sigma_{k'})) ) ,    
\end{eqnarray}

\begin{eqnarray}
\label{eqn_repulsive_fun2}
 h_K((m_1, \Sigma_1), \ldots, (m_K, \Sigma_K)) = \bigg(  \prod_{1 \le k < k' \le K} g( W^2_2(N(m_k, \Sigma_k), N(m_{k'}, \Sigma_{k'})) ) \bigg)^{\frac{1}{K}},    
\end{eqnarray}
for $K \ge 2$, where we recall that $W_2$ is the 2-Wasserstein distance. For $K=1$, we let $h((m_1, \Sigma_1)) := 1$. We refer to the general model in \eqref{eqn_gen_model} with the repulsive prior in \eqref{eqn_repulsive_prior} as the  Bayesian Wasserstein
repulsive Gaussian mixture (WRGM) model.
\\\\
Note that in \cite{Xie2020}, the monotone function $g$ and, consequently, the repulsion function $h_K$ depend solely on the mean vectors. In contrast, in our formulation, they depend on both the mean vectors and the covariance matrices. To recall, the 2-Wasserstein distance between two normal distributions, as given in \eqref{eqn_wasserstein_gaussian}, consists of the distance between the mean vectors and the distance between the covariance matrices. While the repulsion function $h_K$ in \cite{Xie2020} heavily penalizes small distances between the mean vectors, irrespective of the distance between the covariance matrices, our approach takes into account both the mean and covariance components in the penalty.
\\\\
A desirable property of using the repulsion functions in \eqref{eqn_repulsive_fun1} and \eqref{eqn_repulsive_fun2} is that a connection between the normalizing constant $Z_K$ and the number of components $K$ can be established. The result below is analogous to Theorem 1 in \cite{Xie2020}.

\begin{proposition}
\label{thm_normalizing_const}
    Suppose the repulsive function $h_K$ is of the form either \eqref{eqn_repulsive_fun1} or \eqref{eqn_repulsive_fun2}. If
\begin{eqnarray}
\label{eqn_integrability_cond}
& \int_{\mathbb{R}^p \times \mathbb{R}^p} \int_{{\cal S} \times {\cal S}} &  \Big[ \log g(W_2(N(m_1, \Sigma_1), N(m_2, \Sigma_2))) \Big]^2 \nonumber \\
 && p_{m}(m_1) p_{m}(m_2) p_{\Sigma}(\Sigma_1) p_{\Sigma}(\Sigma_2) dm_1 dm_2 d\Sigma_1 d\Sigma_2 < \infty,
\end{eqnarray}
then $0 \le -\log Z_K \le c_1 K$.
\end{proposition}

\begin{remark}
    
Since $W_2(N(m_1, \Sigma_1), N(m_2, \Sigma_2)) \ge ||m_1 - m_2||$, a sufficient condition to ensure the condition \eqref{eqn_integrability_cond} holds is the following integrability condition:
\begin{eqnarray}
    \int_{\mathbb{R}^p \times \mathbb{R}^p} \big( \log g(||m_1 - m_2||) \big)^2 p_m(m_1) p_m(m_2) dm_1 dm_2 < \infty ,
\end{eqnarray}
which is identical to the condition in Theorem 1 of \cite{Xie2020}.

\end{remark}

\subsection{Posterior Contraction Rate}
We analyze the posterior contraction rate of the WRGM, extending the results of \cite{Xie2020} to our current framework. To establish this, we first outline the necessary assumptions for the model and prior. These assumptions largely align with those in \cite{Xie2020}.

\subsubsection{Assumptions}
 We require the following conditions to hold for the model:

\begin{itemize}
    \item[A0] The data generating density $f_0$ is of the form $f_0 = \phi(\cdot|m, \Sigma) * F_0$ where $*$ represent the convolution of two functions and $F_0 \in M(\mathbb{R}^p \times {\cal S})$ has a sub-Gaussian tail:
    $$ F_0(||m|| \ge t) \le B_1 \exp(-b_1 t^2)$$
    for some $B_1, b_1 > 0$.

    \item[A1] For some $\delta > 0, c_2 > 0$, we have $g(x) \ge c_2 \epsilon$ whenever $x \ge \epsilon$ and $\epsilon \in (0, \delta)$.

    \item[A2] $g$ satisfies 
\begin{eqnarray*}
& \int_{\mathbb{R}^p \times \mathbb{R}^p} \int_{\mathbb{S} \times \mathbb{S}} &  \Big[ \log g(W_2(N(m_1, \Sigma_1), N(m_2, \Sigma_2))) \Big]^2 \nonumber \\
 && p_{m}(m_1) p_{m}(m_2) p_{\Sigma}(\Sigma_1) p_{\Sigma}(\Sigma_2) dm_1 dm_2 d\Sigma_1 d\Sigma_2 < \infty.
\end{eqnarray*}

\item[A3] There exists $\underline{\sigma}^2, \overline{\sigma}^2 \in (0, +\infty)$, $\underline{\sigma}^2 \le \inf_{\Sigma \in {\cal S}} \lambda(\Sigma) \le \sup_{\Sigma \in {\cal S}} \lambda(\Sigma) \le \overline{\sigma}^2$.

\end{itemize}
\begin{remark}
    Condition A0, A1, and A3 are identical to those in \cite{Xie2020}, while Condition A2 is analogous to A2 in \cite{Xie2020}. Condition A0 requires that the true density \( f_0 \) takes the form of a (nonparametric) Gaussian mixture model. Condition A1 is satisfied, for example, by letting \( g(x) = \frac{x}{g_0 + x} \) for some fixed \( g_0 > 0 \), which is the form used in \cite{Xie2020}. Note that in \cite{Xie2020}, they also required all \( \Sigma \in {\cal S} \) to be simultaneously diagonalizable. This assumption leads to a simpler derivation of the posterior contraction rate as well as posterior sampling. While this assumption is reasonable in their setting, as their prior only incorporates mean repulsion, it is undesirable in our setting. With this simplifying assumption, they derived a slightly faster contraction rate. However, we show that by removing this assumption, our posterior contraction rate remains of the same order, differing only by a logarithmic term.
\end{remark}

The following conditions are required for the prior distribution.

\begin{itemize}
    \item[B1] $(w_1, \ldots, w_K) \sim {\cal D}_{K}(\beta)$ is weakly informative, i.e. $\beta \in (0, 1]$.

    \item[B2] $p_{m}$ has a sub-Gaussian tail, i.e.
    $$ \int_{||m|| \ge t} p_{m}(m) dm \le B_2 \exp(-b_2 t^2) .$$

    \item[B3] For all $m \in \mathbb{R}^p$, $p_{m}(m) \ge B_3 \exp(-b_3 ||m||^{\alpha})$, for some $\alpha \ge 2$, and $B_3, b_3 > 0$.

    \item[B4] $p_{\Sigma}$ is induced by 
    $$ p_{\Sigma}(\Sigma) \propto  p_{IW}(\Sigma; \Psi, \nu) \prod_{j=1}^{p} 1_{[\underline{\sigma}^2, \overline{\sigma}^2]}(\lambda_j(\Sigma)) ,$$
    where $p_{IW}(\cdot,\Psi, \nu) $ is the density of inverse-Wishart distribution where $\Psi$ a $p\times p$ positive definite scale matrix, and $\nu > p-1$ is the degree of freedom. 
    \item[B5] There exists some $B_4, b_4 > 0$ such that for sufficiently large $K$,
    $$ p_K(K) \ge \exp(-b_4 K \log K), \quad \sum_{N=K}^{\infty} p_K(N) \le \exp(-B_4 K \log K). $$
\end{itemize}
We observe that Conditions B1, B2, B3, and B5 align with those presented in \cite{Xie2020}. However, for Condition B4, we choose to use the more widely adopted inverse-Wishart distribution for the covariances, while imposing eigenvalue boundedness conditions.

\subsubsection{Posterior Contraction Rate}
For a prior distribution $\Pi$ on ${\cal M}(\mathbb{R}^p)$, we study the contraction rate of the posterior distribution $\Pi(\cdot|y_1,\ldots,y_n)$ towards the true data generating density $f_0$ measured using the $L_1$-norm $||\cdot||_1$ on $L^1(\mathbb{R}^p)$ with respect to $\mathbb{P}_0$-probability. Given a sequence $(\epsilon_n)_n$, the posterior is said to contract towards $f_0$ with rate (at least) $\epsilon_n$ if 
$$ \Pi(f \in {\cal M}(\mathbb{R}^p): ||f - f_0||_1 > M \epsilon_n | y_1, \ldots, y_n) \rightarrow 0$$
as $n \rightarrow \infty$ in $\mathbb{P}_0$-probability for some constant $M > 0$. We have the following contraction rate result under the given assumptions, which extends Theorem 4 of \cite{Xie2020} to the present setting.
\begin{theorem}
    \label{thm_post_contraction}
    Suppose conditions A0-A3 and B1-B5 hold. Then the posterior distribution $\Pi(\cdot|y_1, \ldots, y_n)$ contracts towards $f_0$ under the $||\cdot||_1$ norm at a rate of $\epsilon_n = (\log n)^t / \sqrt{n}$, $t > \frac{p^2}{2} + p+\frac{\alpha+2}{4}$.
\end{theorem}
\begin{remark}
We compare the contraction rate \(\frac{(\log n)^t}{\sqrt{n}}\) with \(t > \frac{p^2}{2} + p + \frac{\alpha+2}{4}\) to the rate \(\frac{(\log n)^t}{\sqrt{n}}\) with \(t > p + \frac{\alpha+2}{4}\) in Theorem 4 of \cite{Xie2020}. We observe that our rate is only slower by a logarithmic factor, which arises because we no longer impose the simultaneous diagonalization condition on the covariance matrices.

\end{remark}

\subsection{Posterior Inference}
\label{subsec_post_inf}
We modify the blocked-collapsed Gibbs sampler introduced by \cite{Xie2020} for posterior inference in the WRGM model. As in \cite{Xie2020}, we express the WRGM model using the exchangeable partition distribution \cite{Miller2018} and apply the data augmentation technique from \cite{Neal2000} to develop a blocked-collapsed Gibbs sampler similar to that of \cite{Xie2020}. Since the derivation follows the same approach and the sampler closely resembles the one in \cite{Xie2020}, we omit the full details and instead focus on highlighting the key differences.
\\\\
Similar to \cite{Xie2020}, we use the first type of repulsion function \( h_K \) from \eqref{eqn_repulsive_fun1}, setting the monotone function \( g \) as \( g(x) = \frac{x}{g_0 + x} \) for some constant \( g_0 > 0 \). The key distinction in our approach is that \( g \) and, consequently, \( h_K \) are defined based on the 2-Wasserstein distance between normal components, rather than solely on the distance between their location parameters. This implies that while RGM samples the location parameters and covariances independently, WRGM samples them jointly. Furthermore, \cite{Xie2020} assumed, for simplicity, that all covariance matrices were diagonal and incorporated an algorithm that involved sampling the eigenvalues of these matrices. In contrast, we relax this assumption by adopting a conjugate inverse-Wishart prior on covariance matrices, subject to an eigenvalue boundedness condition (Condition B4). This facilitates a closed-form update of the covariance matrices, with any sampled matrix being rejected if it fails to satisfy the eigenvalue boundedness condition. The rest of the steps of our sampler are identical to those in Algorithm 1 of \cite{Xie2020}. Our code for this blocked-collapsed Gibbs sampling algorithm is publicly available at~\url{https://github.com/huangweipeng7/wrgm}.

\section{Simulation Studies}
\label{sec_simulation}
We conduct a series of simulations to evaluate the performance of the proposed WRGM model. In particular, we focus on its ability to recover the true data-generating density. The WRGM is compared against the RGM and the Mixture of Finite Mixtures (MFM) models \citep{Richardson2002}. We also examine the minimum distance between the mean parameters of mixture components under the three models, focusing on WRGM and RGM, as WRGM imposes repulsion jointly on means and covariances, while RGM enforces repulsion only on the means.
\\\\
In all three simulation settings, we simulate data from Gaussian mixtures. In each case, we fix two mixture components with the same mean but different covariance structures. Specifically, both components share the mean \( m = (0, 0)^T \), while their covariances are given by
\[
\Sigma_1 = 
\begin{bmatrix}
1 & 0 \\
0 & 100
\end{bmatrix}
\quad \text{and} \quad
\Sigma_2 = 
\begin{bmatrix}
100 & 0 \\
0 & 1
\end{bmatrix}.
\]
Intuitively, one would expect that the RGM, which penalizes small distances between component means, may push the mean parameters further apart. In contrast, the WRGM may not exhibit the same behavior, as it accounts for both mean and covariance parameters, and when the covariance structures of the components are already substantially different, there may be less incentive to separate the means further.
\\\\
In addition to these two fixed components, we also include randomly sampled Gaussian components. The mean and covariance of each such component, denoted by \(\tilde{m}_k\) and \(\tilde{\Sigma}_k\) for \(k = 1, \ldots, K_s\), are drawn according to
\begin{equation}
\label{eq:sim_sample}
\tilde{m}_k \sim N(m_0, \Sigma_0), \quad \tilde{\Sigma}_k \sim \text{IW}(\Psi_0, 5),  
\end{equation} 
where $K_s$ is the number of random components, $\text{IW}$ denotes the inverse-Wishart distribution and
\[
m_0 = (0, 0)^T, \quad 
\Sigma_0 = \begin{bmatrix}
25 & 0 \\
0 & 25
\end{bmatrix}, \quad
\Psi_0 = \begin{bmatrix}
10 & 0 \\
0 & 10
\end{bmatrix}.
\]
Consequently, the true density takes the form:
\[
f(y) = \frac{1}{K_s+2} \bigg( \phi(y| m, \Sigma_1) + \phi(y | m, \Sigma_2) + \sum_{k=1}^{K_s} \phi(y | \tilde{m}_k, \tilde{\Sigma}_k) \bigg)
\] 
In all three simulation settings, we set the sample size to 500. The number of randomly generated mixture components, denoted by \(K_s\), is set to 3, 4, and 5, respectively. Combined with the two fixed components, this results in true mixture densities consisting of 5, 6, and 7 components, respectively.
\\\\
We apply the WRGM model to the three simulated datasets, using the prior specification outlined in Section \ref{subsec_model} and the posterior inference methodology described in Section \ref{subsec_post_inf}. In particular, we set priors that satisfy Conditions B1 - B5. We specify the Dirichlet concentration parameter as \(\beta = 1\). For the prior on mean parameters, \(p_m\), is a normal distribution with zero mean and covariance \(\tau^2 I\), where \(\tau = 100\). For the inverse-Wishart prior, \(p_{IW}(\Sigma; \Psi, \nu)\), we set \(\nu = 3\) and \(\Psi = I\). Additionally, we set the eigenvalue bounds as $\underline{\sigma}^2=10^{-12}$ and $\overline{\sigma}^2=10^{12}$. Finally, the prior for \(K\), \(p_K(K)\), is taken to be a zero-truncated Poisson distribution with parameter \(\lambda = 1\). We adopt the repulsion function \( h_K \) as defined in \eqref{eqn_repulsive_fun1} and set \( g(x) = \frac{x}{g_0 + x} \) with \( g_0 = 5 \).
\\\\
The MCMC is run for 15,000 iterations, with the first 10,000 iterations used as burn-in. Additionally, we set the thinning interval to 2 for the MCMC runs. We also fit the RGM and MFM models using the same prior hyperparameters as specified for the WRGM model. Since \cite{Xie2020} focuses only on the case of diagonal covariance matrices, we extend the RGM model to accommodate full covariance matrices in our analysis.
\\\\
The estimated densities for the three simulation settings under each of the three models are shown in Figures \ref{fig:sim1_de}, \ref{fig:sim2_de}, and \ref{fig:sim3_de}. As in \cite{Xie2020}, we also compute the logarithm of the conditional predictive ordinate (log-CPO) for the three models using the post-burn-in samples:
$$ \text{log-CPO} = - \sum_{i=1}^{n} \log \bigg( \frac{1}{n_{mc}} \sum_{j=1}^{n_{mc}} \frac{1}{p(y_i|\Theta_{mc}^{j})} \bigg), $$
where $\Theta_{mc}^{j}$ represents the post-burn-in samples of all parameters generated by the $j$th iteration, $p(y_i|\Theta_{mc}^{(j)})$ is the density of data $y_i$, $n_{mc}$ is the number of post-burn-in MCMC samples. The maximum a posteriori (MAP) component assignments corresponding to each simulation and model are presented in Figures \ref{fig:sim1_map}, \ref{fig:sim2_map}, and \ref{fig:sim3_map}. Additionally, Figures \ref{fig:sim1_min_d}, \ref{fig:sim2_min_d}, and \ref{fig:sim3_min_d} display the minimum pairwise distances between the mean parameters of the mixture components for each model.
\\\\
Across all three simulation settings, the WRGM model consistently achieves slightly better density estimates compared to RGM, and both WRGM and RGM outperform the MFM model regarding log-CPO. Additionally, the MAP component assignments under WRGM more closely reflect the true component structure relative to the other two models. Notably, the minimum pairwise distances between the component means in WRGM are consistently smaller than those in RGM and MFM. This observation aligns with our intuition: WRGM imposes repulsion based on the Wasserstein distance, which accounts for both mean and covariance differences, and does not require the component means to be far apart—unlike RGM, which directly penalizes proximity between means.

\begin{figure}[ht]
    \centering
        \subfloat{
        \includegraphics[width=0.46\linewidth]{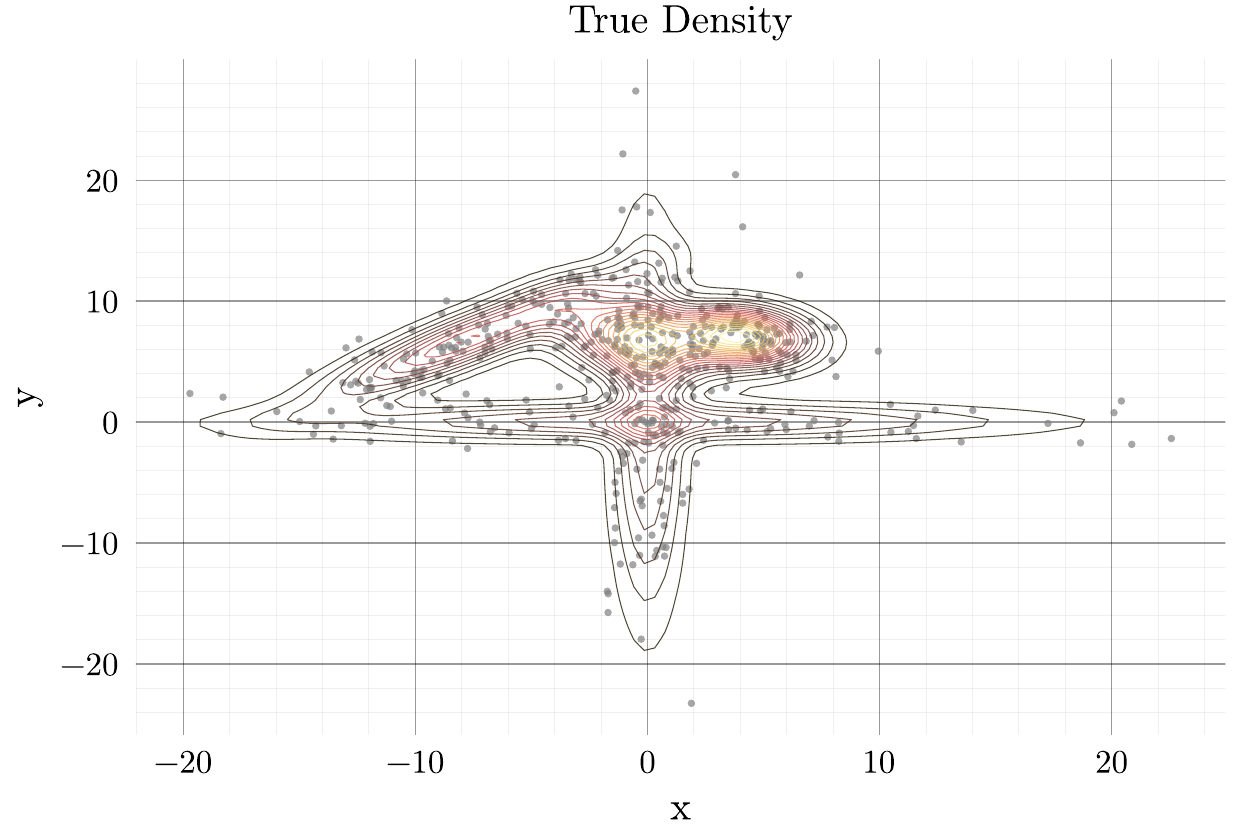} 
    } \quad
    \subfloat{
        \includegraphics[width=0.46\linewidth]{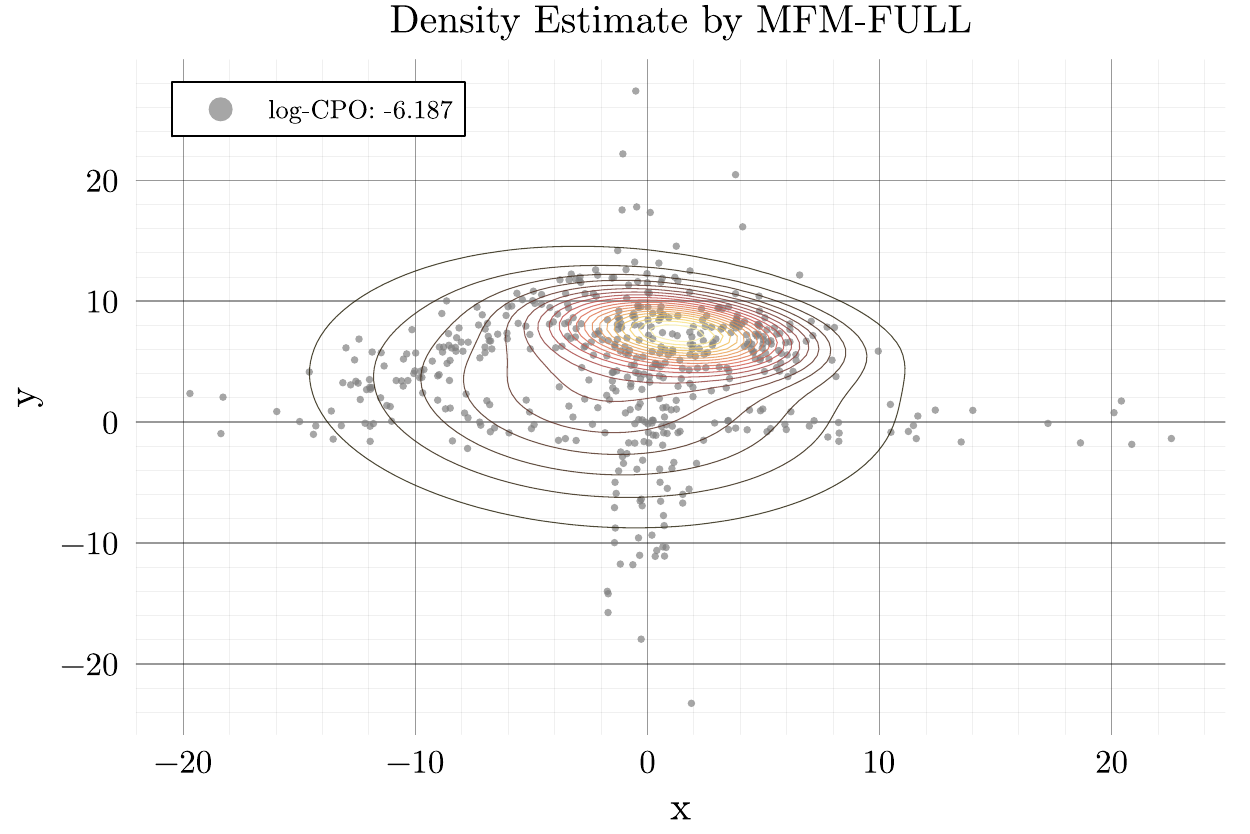} 
    }
    \\
    \subfloat{
        \includegraphics[width=0.46\linewidth]{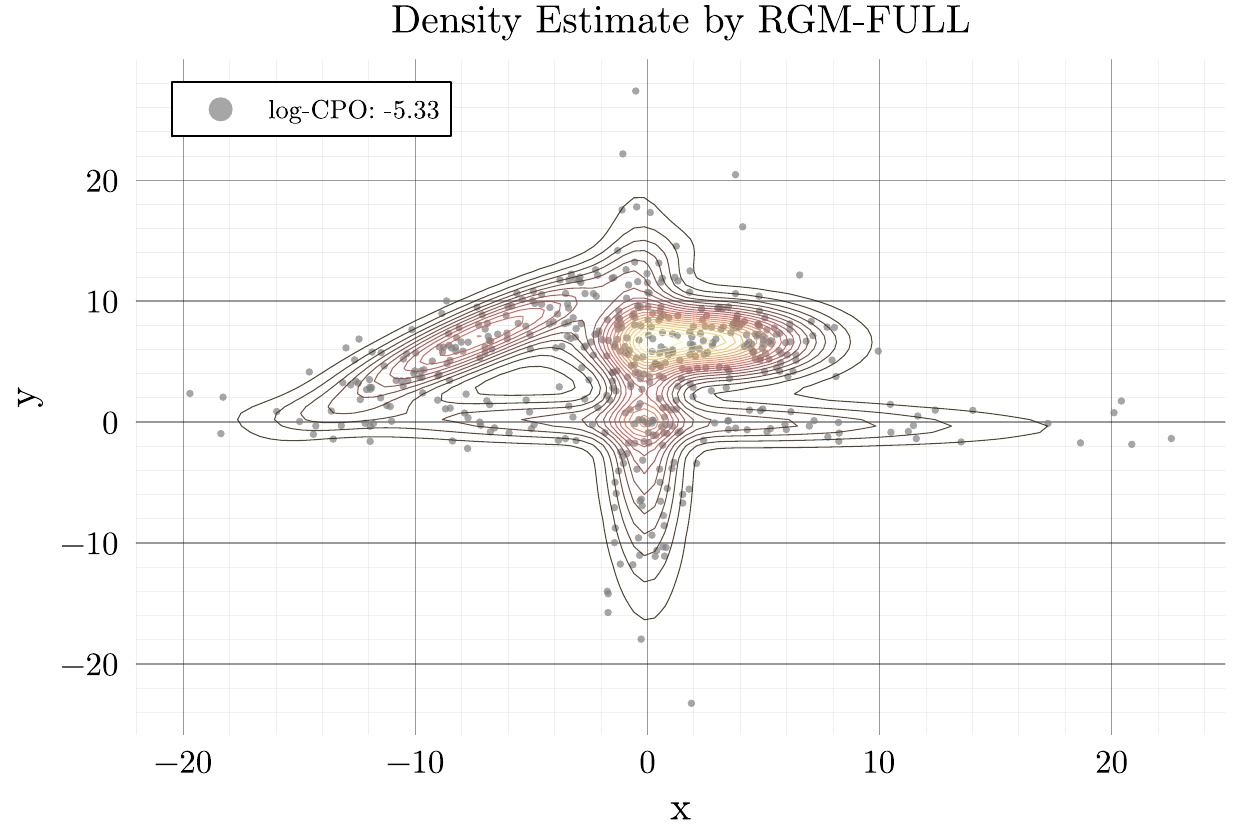} 
    } \quad
    \subfloat{
        \includegraphics[width=0.46\linewidth]{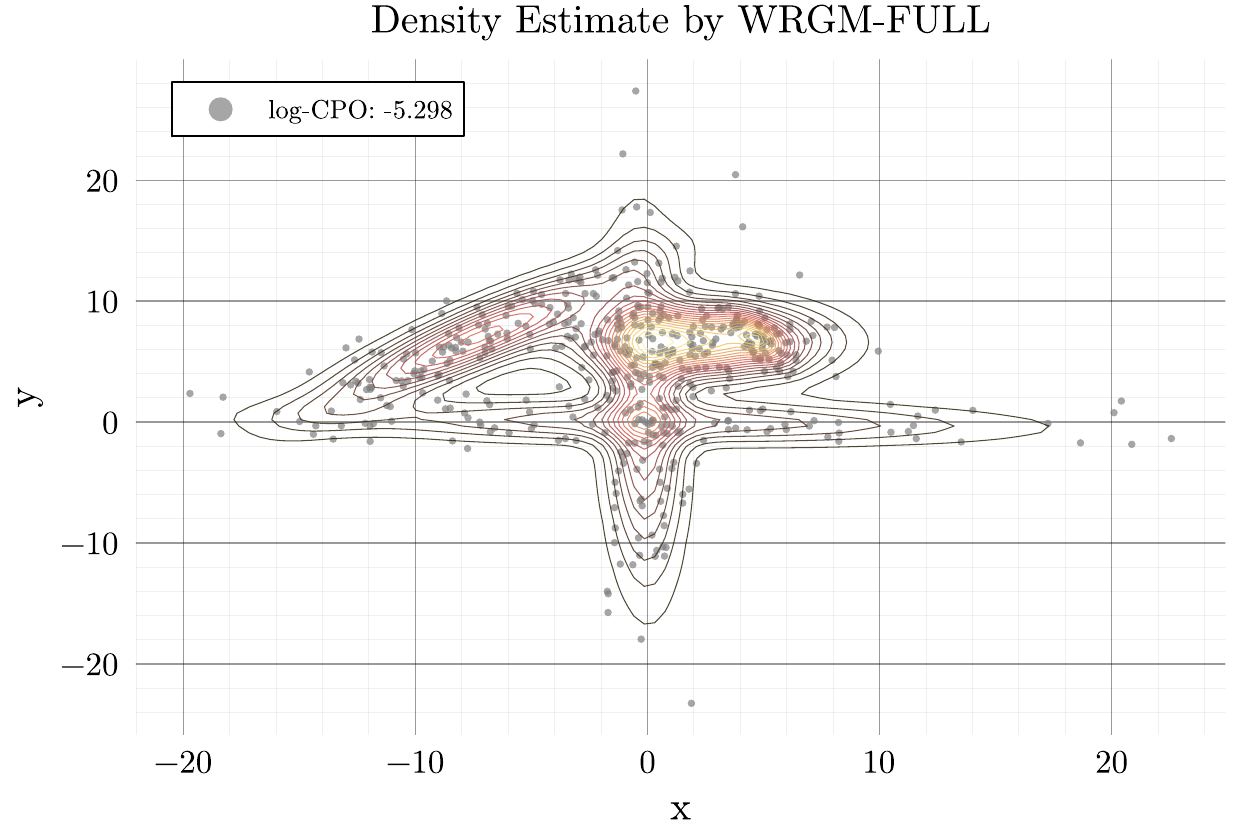} 
    } 
    \caption{Comparison of the true density with the densities estimated by the three models in Simulation 1.}
    \label{fig:sim1_de}
\end{figure}

\begin{figure}[ht]
    \centering
    \subfloat{
        \includegraphics[width=0.46\linewidth]{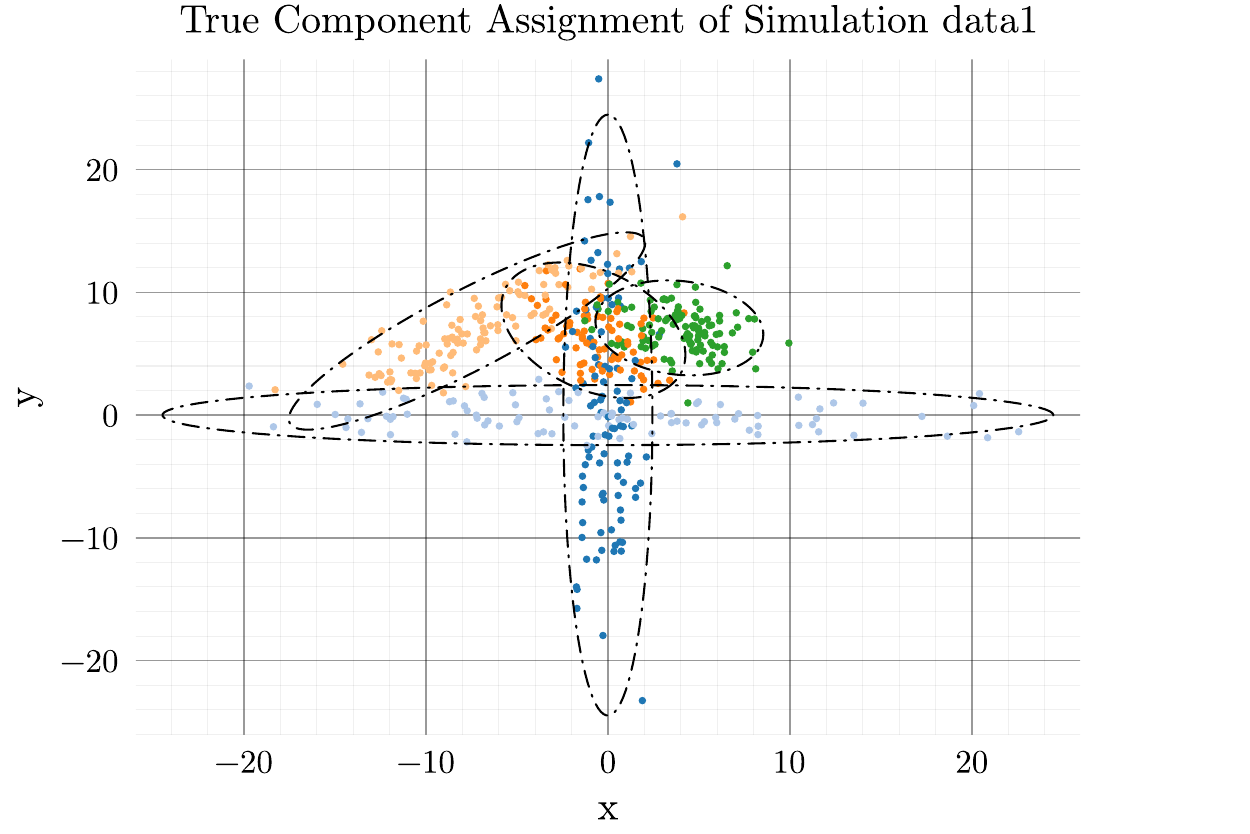} 
    } \quad 
    \subfloat{
        \includegraphics[width=0.46\linewidth]{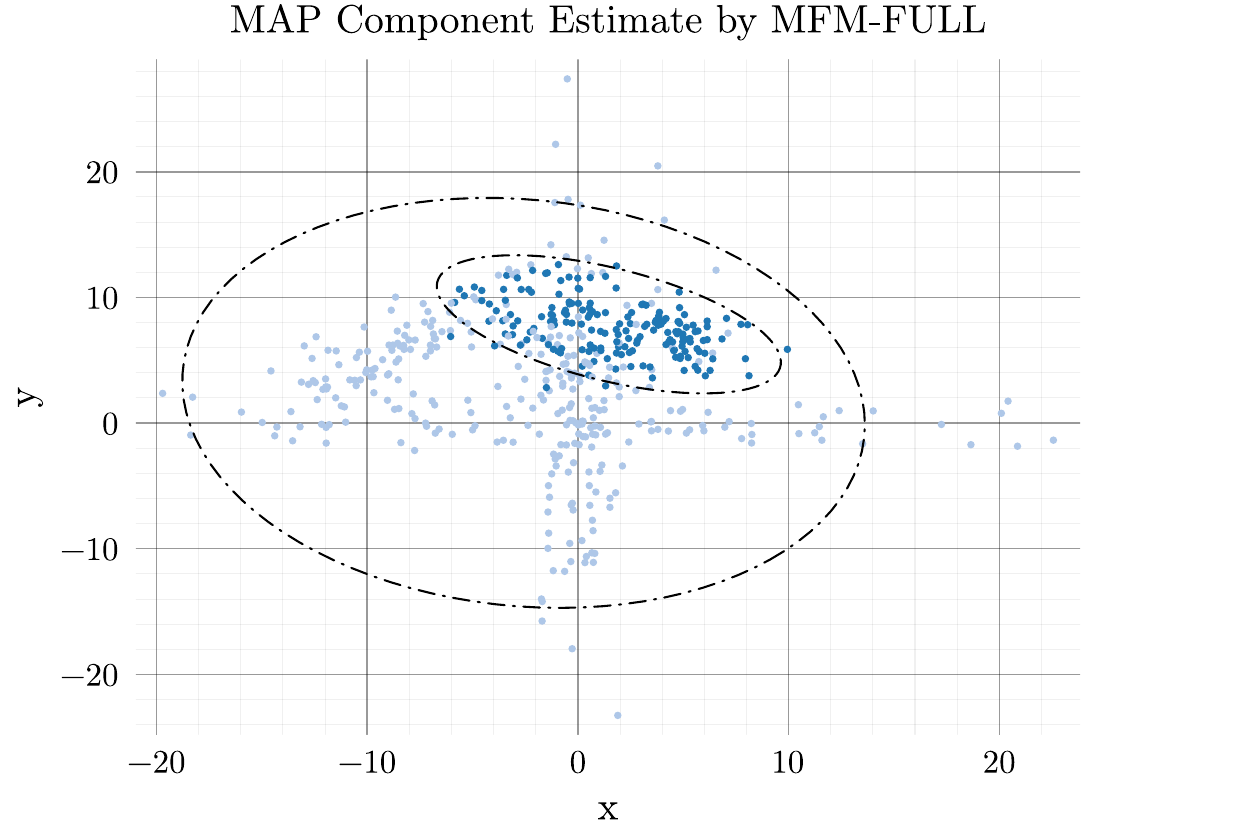} 
    }
    \\ 
    \subfloat{
        \includegraphics[width=0.46\linewidth]{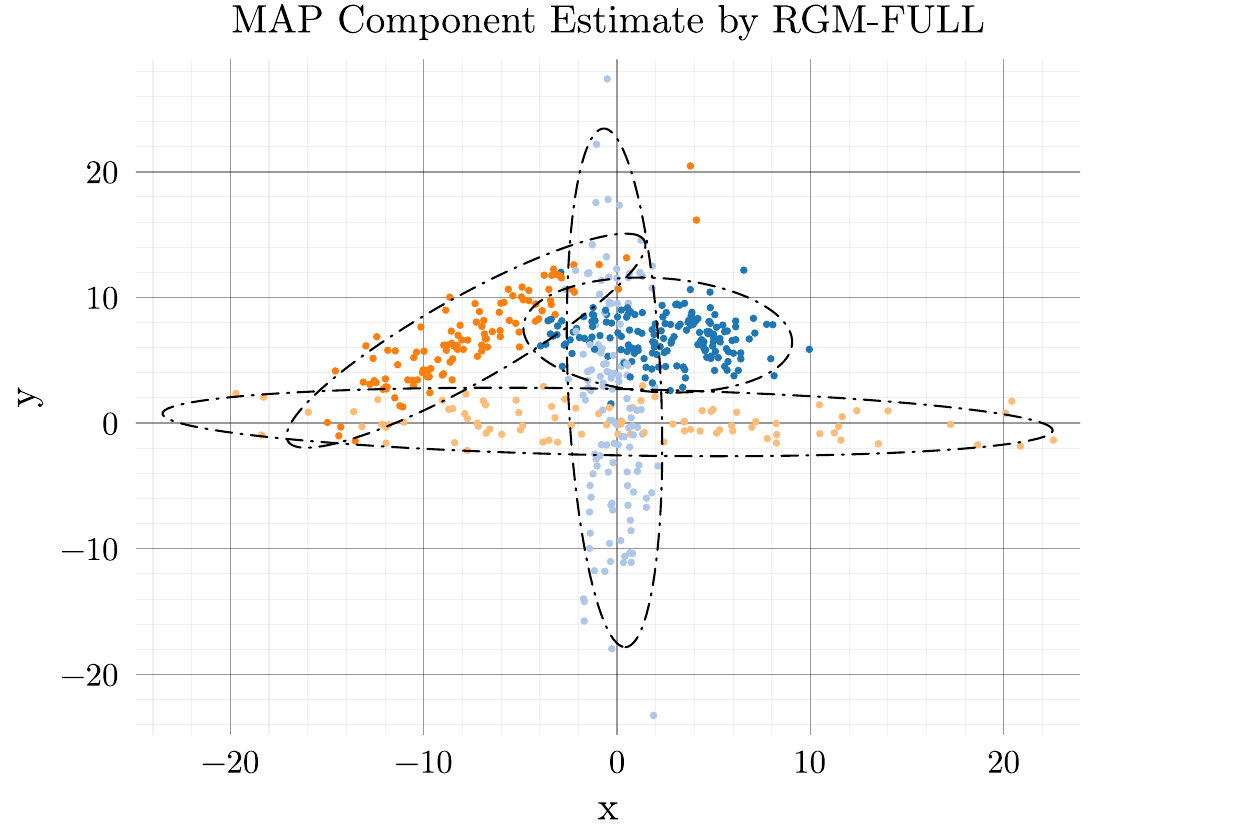} 
    } \quad
    \subfloat{
        \includegraphics[width=0.46\linewidth]{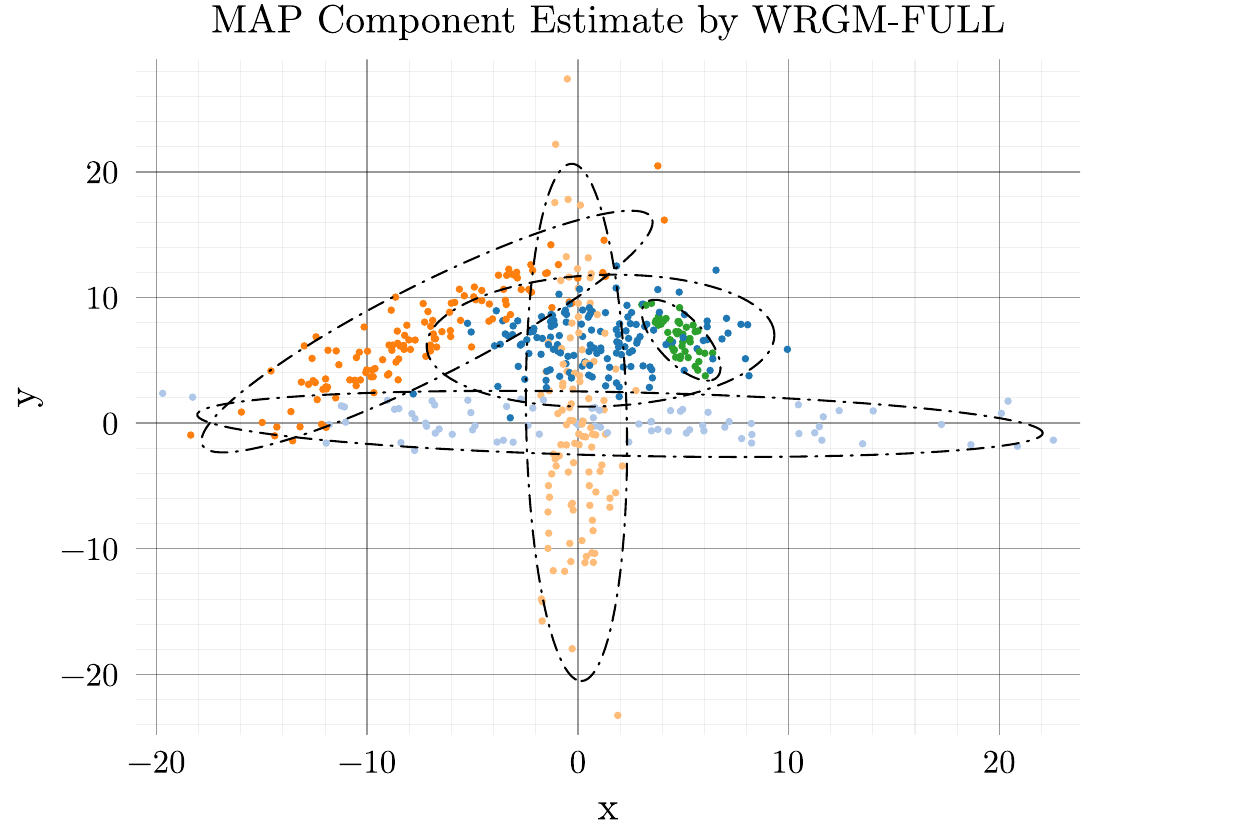} 
    }  
    \caption{Maximum a Posteriori (MAP) component assignments for the three models in Simulation 1.}
    \label{fig:sim1_map}
\end{figure}

\begin{figure}
    \centering
    \subfloat{
        \includesvg[width=0.7\linewidth]{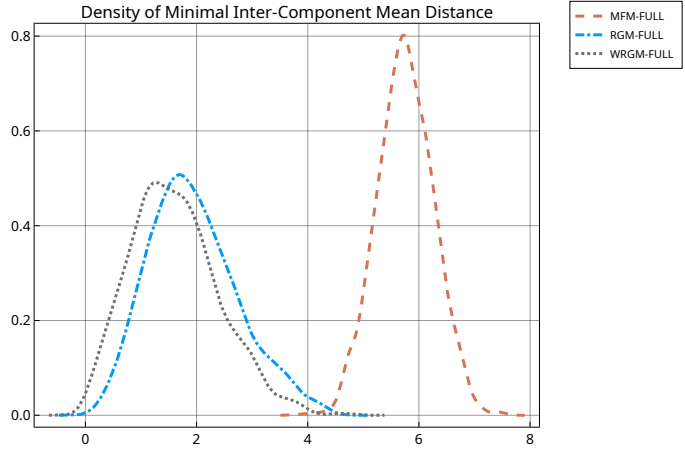} 
    }  
    \caption{Minimum pairwise distances between the component mean parameters for each of the three models in Simulation 1.}
    \label{fig:sim1_min_d}
\end{figure}

\begin{figure}[ht]
    \centering
        \subfloat{
        \includegraphics[width=0.46\linewidth]{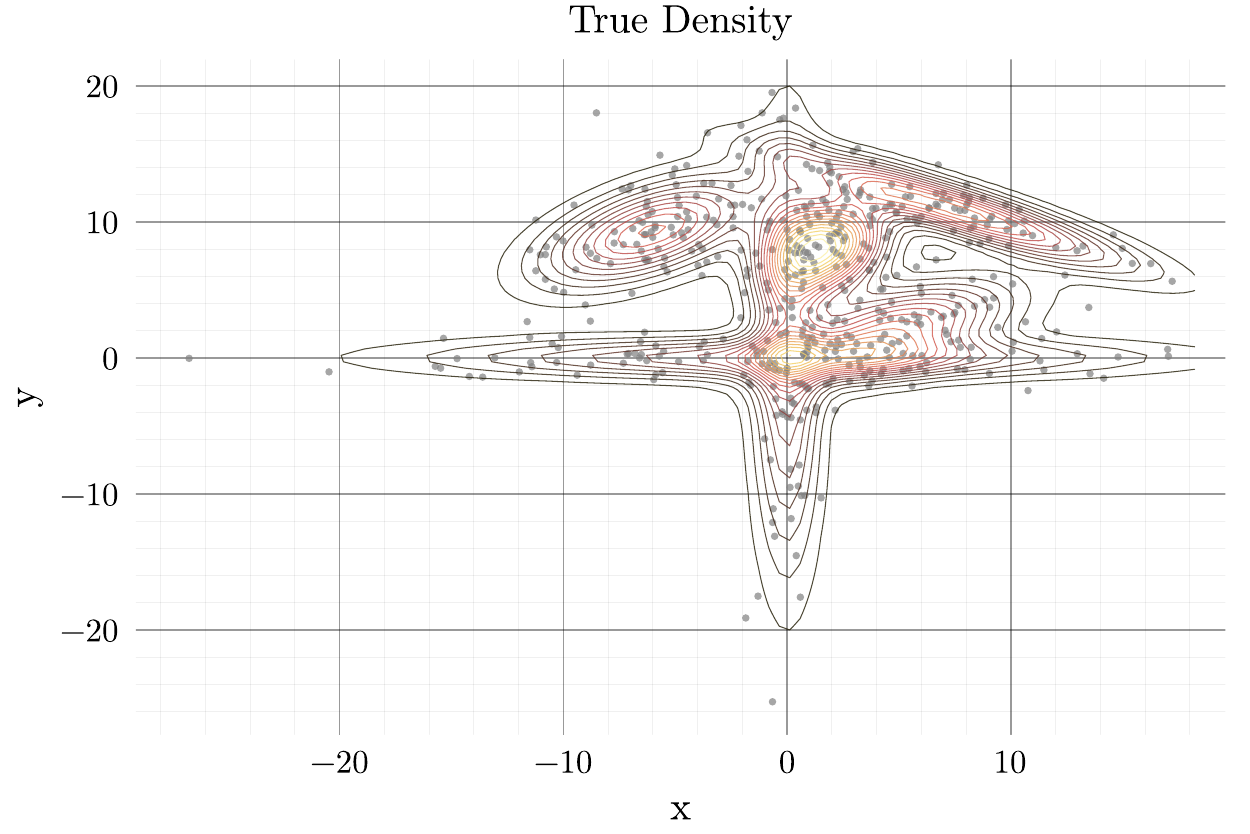} 
    } \quad
    \subfloat{
        \includegraphics[width=0.46\linewidth]{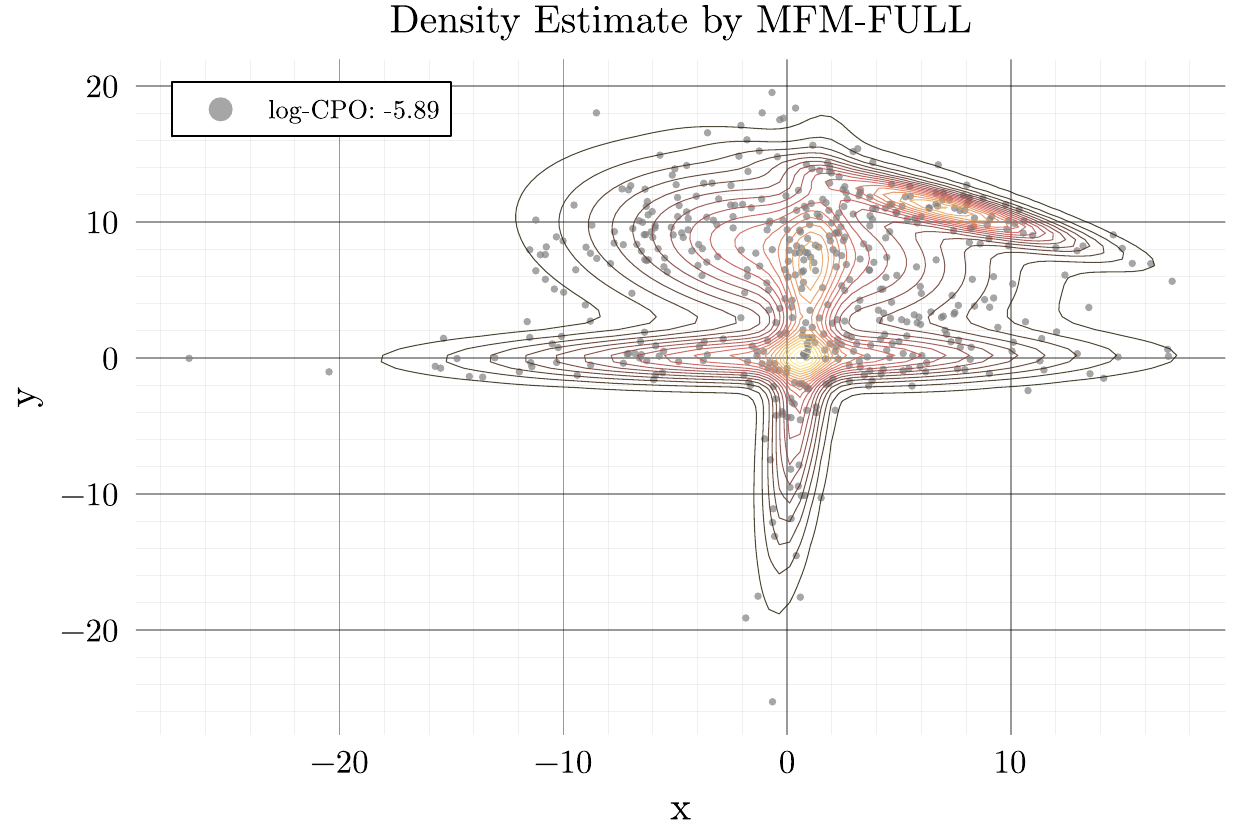} 
    }
    \\
    \subfloat{
        \includegraphics[width=0.46\linewidth]{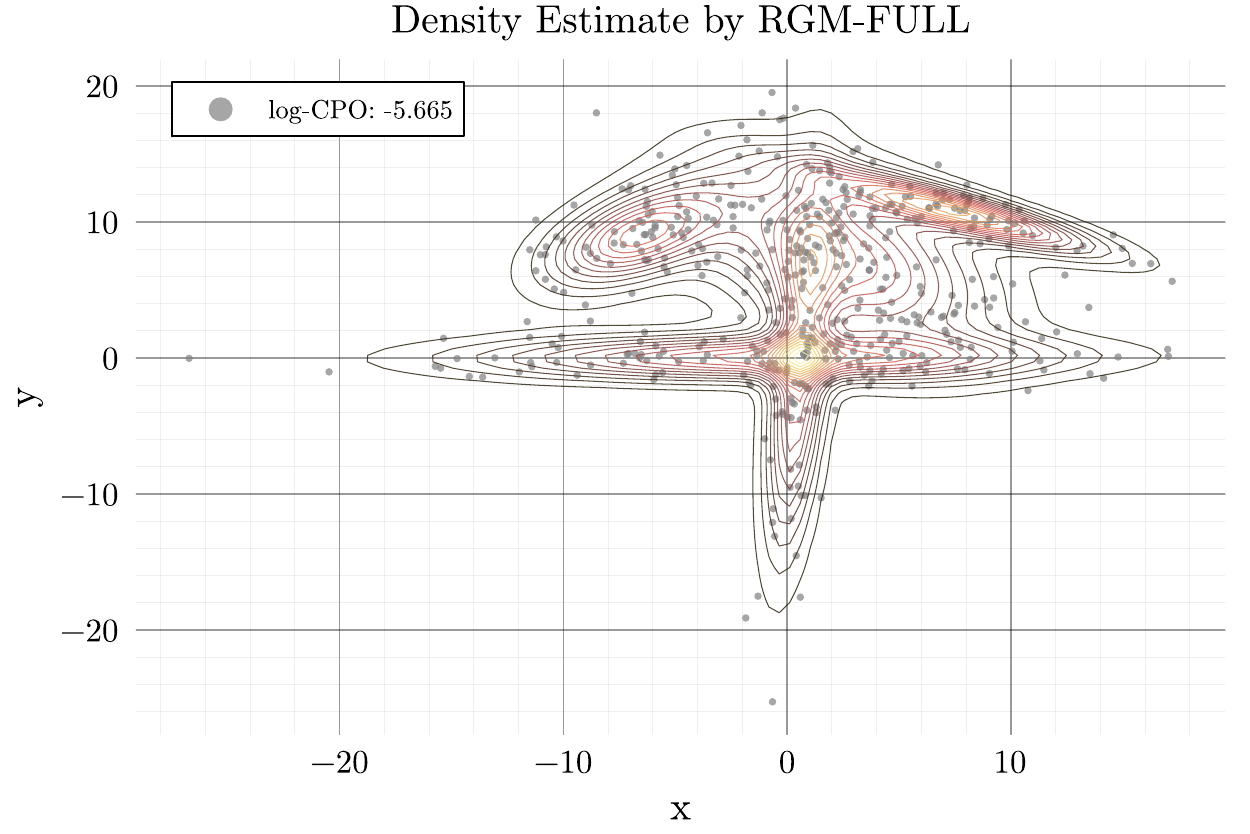} 
    } \quad 
    \subfloat{
        \includegraphics[width=0.46\linewidth]{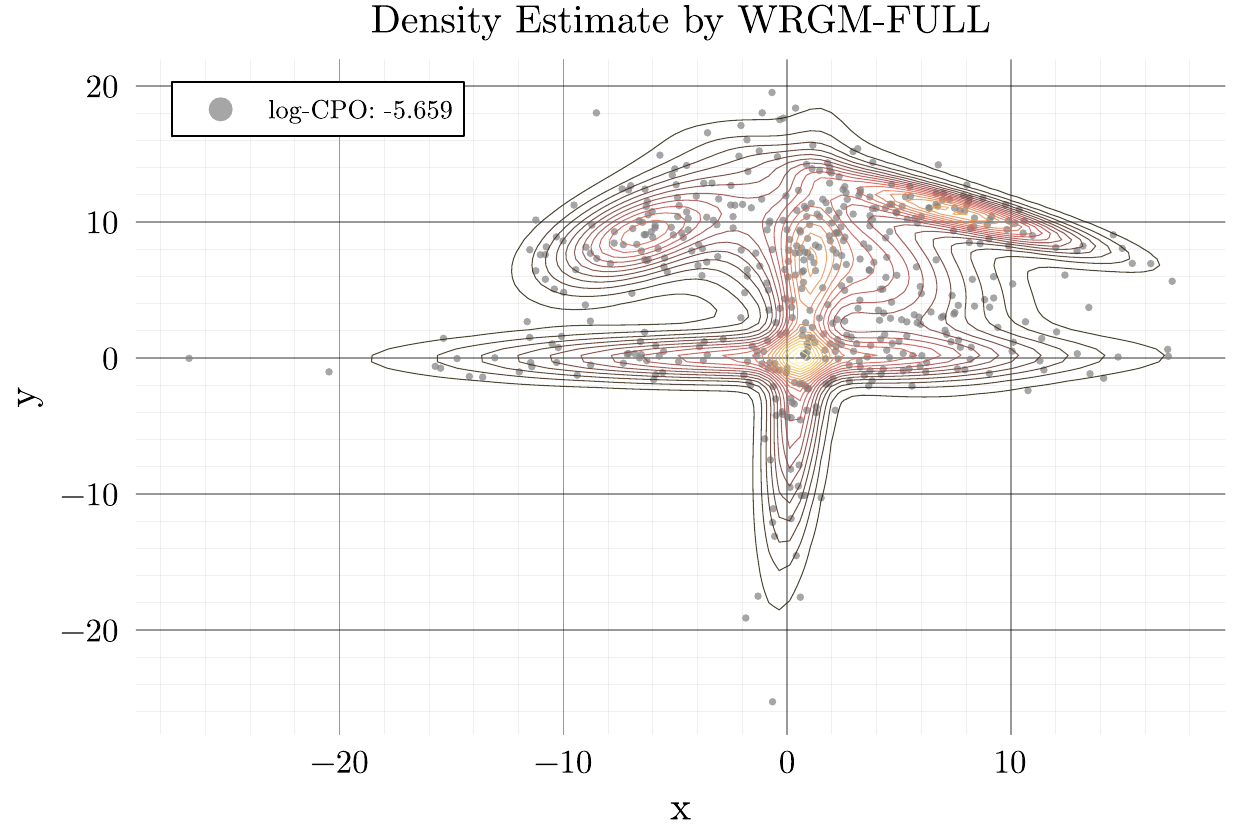} 
    }  
    \caption{Comparison of the true density with the densities estimated by the three models in Simulation 2.}
    \label{fig:sim2_de}
\end{figure}

\begin{figure}[ht]
    \centering
    \subfloat{
        \includegraphics[width=0.46\linewidth]{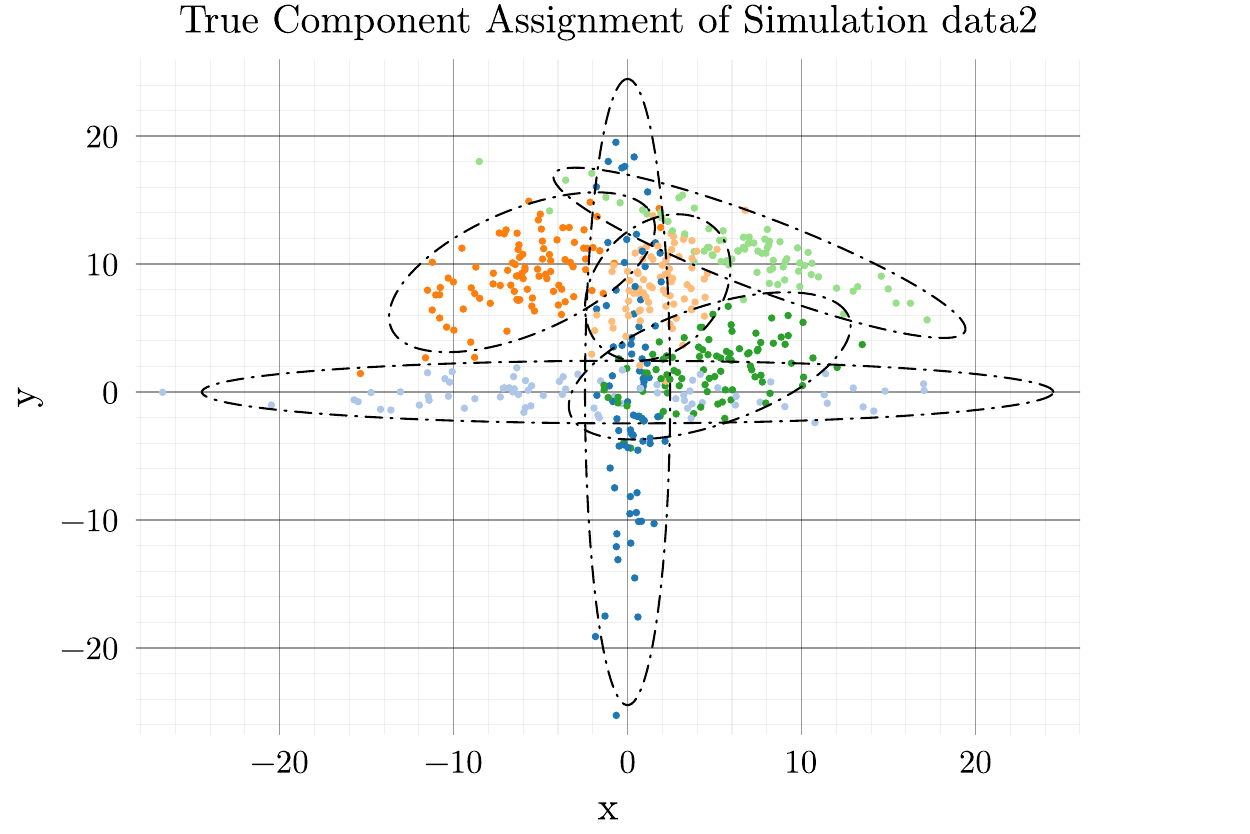} 
    } \quad 
    \subfloat{
        \includegraphics[width=0.46\linewidth]{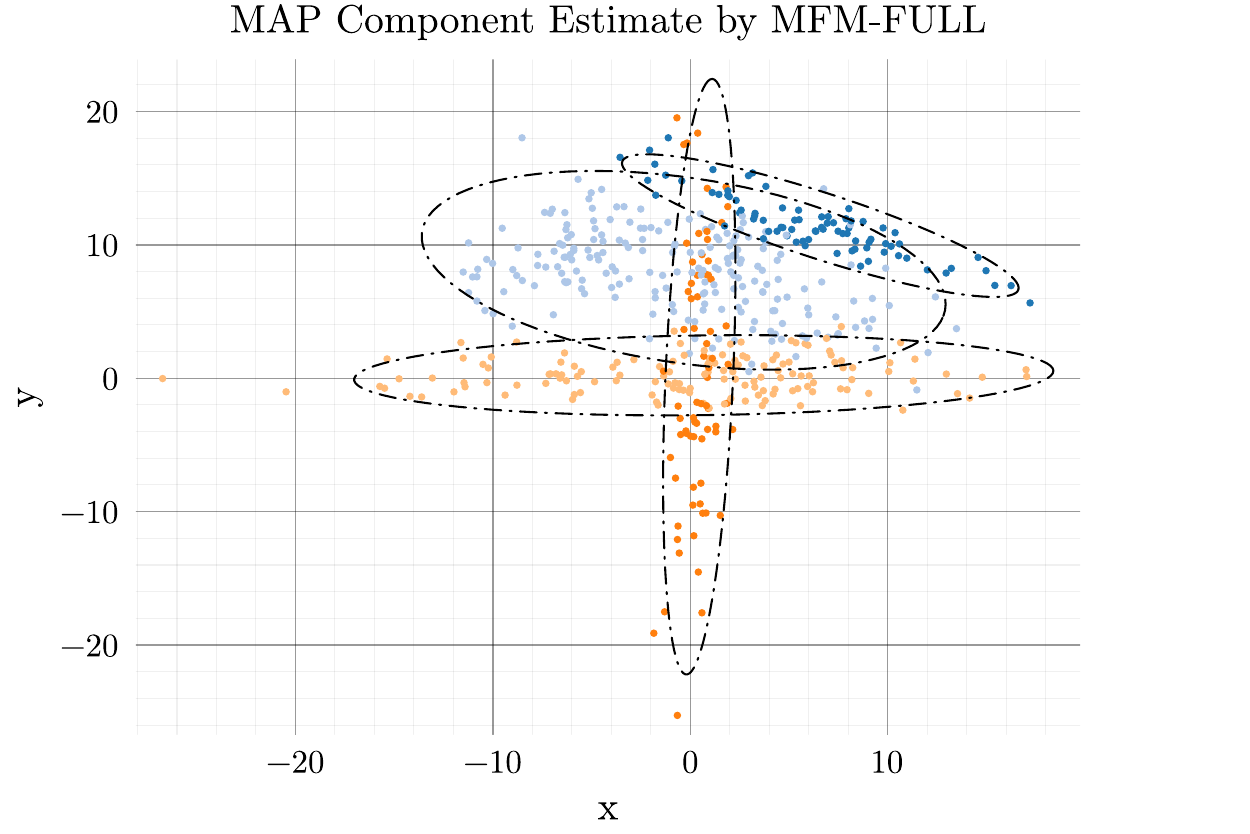} 
    }
    \\
    \subfloat{
        \includegraphics[width=0.46\linewidth]{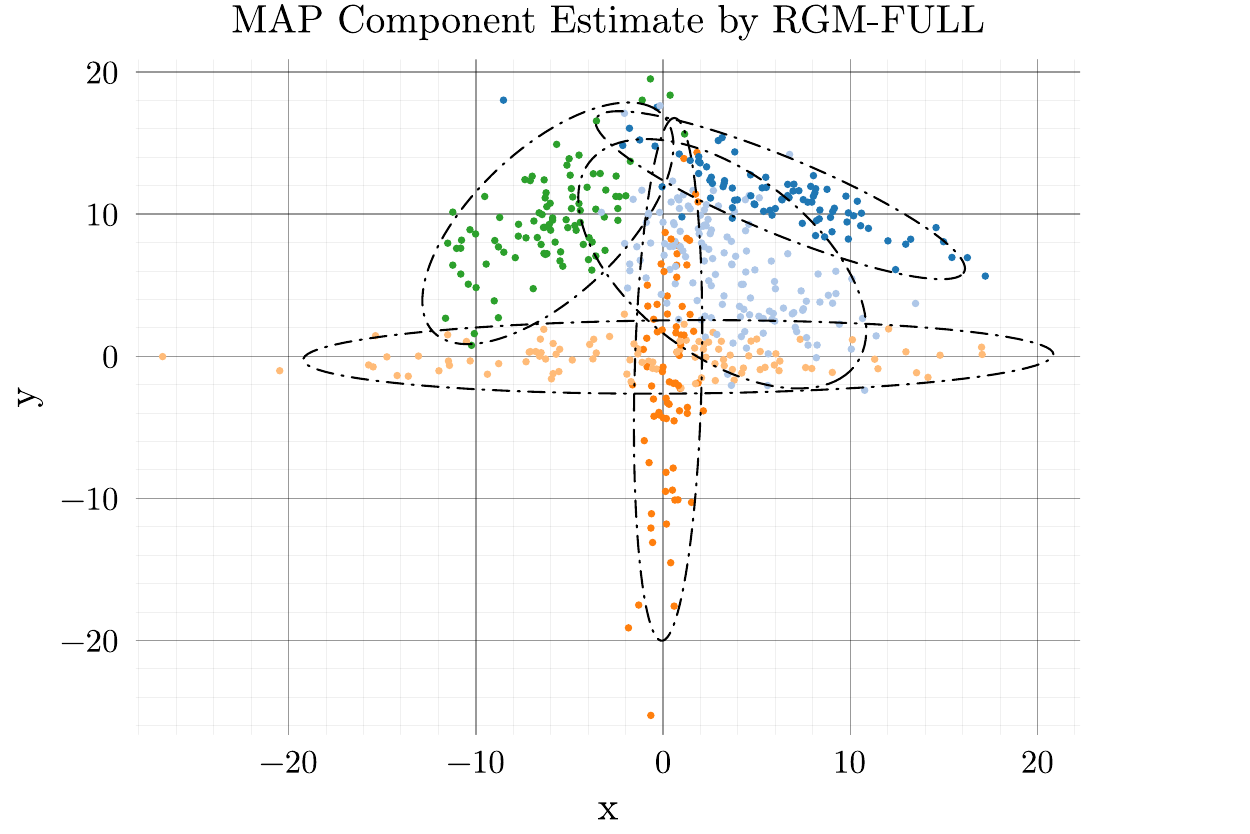} 
    } \quad
    \subfloat{
        \includegraphics[width=0.46\linewidth]{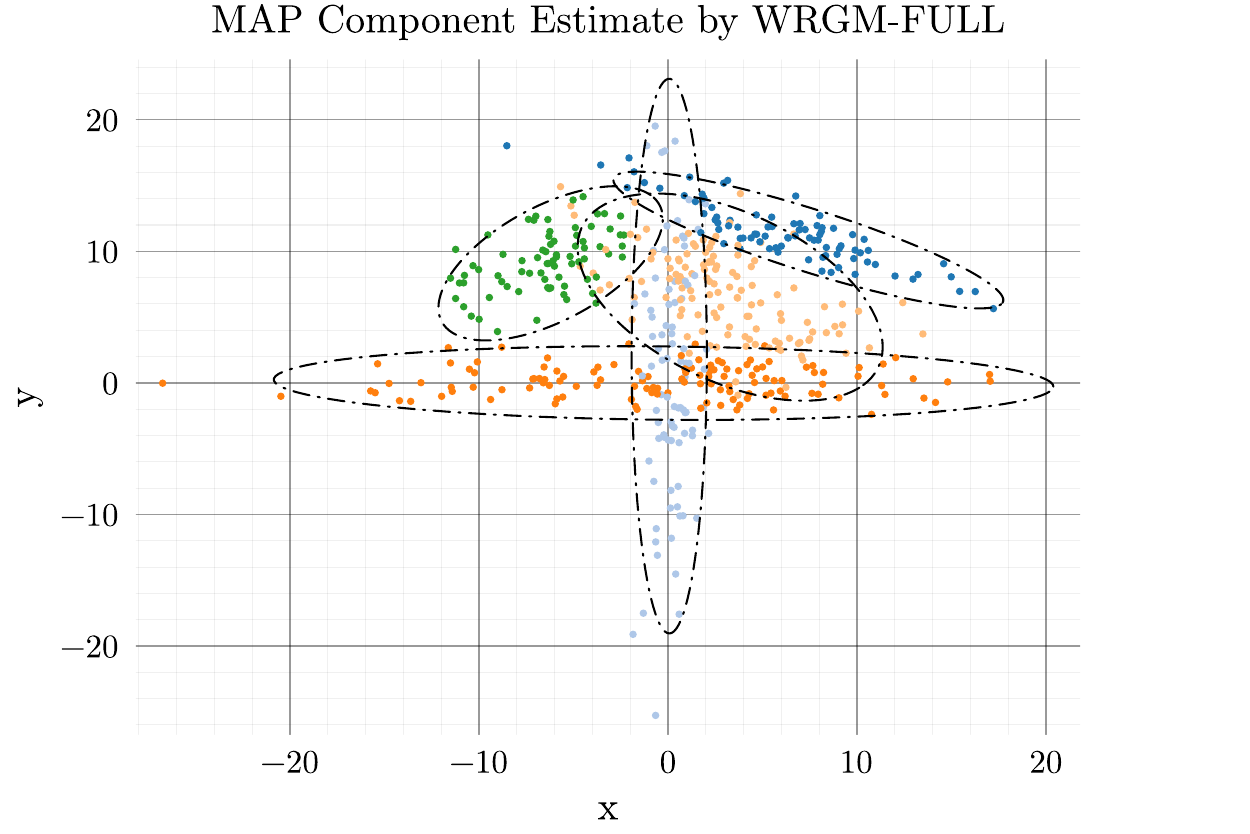} 
    }  
    \caption{Maximum a Posteriori (MAP) component assignments for the three models in Simulation 2.}
    \label{fig:sim2_map}
\end{figure}

\begin{figure}
    \centering
    \subfloat{
        \includesvg[width=0.6\linewidth]{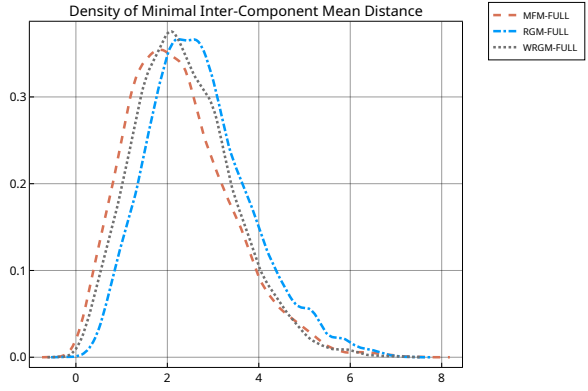} 
    } 
    \caption{Minimum pairwise distances between the component mean parameters for each of the three models in Simulation 2.}
    \label{fig:sim2_min_d}
\end{figure}

\begin{figure}[ht]
    \centering
        \subfloat{
        \includegraphics[width=0.46\linewidth]{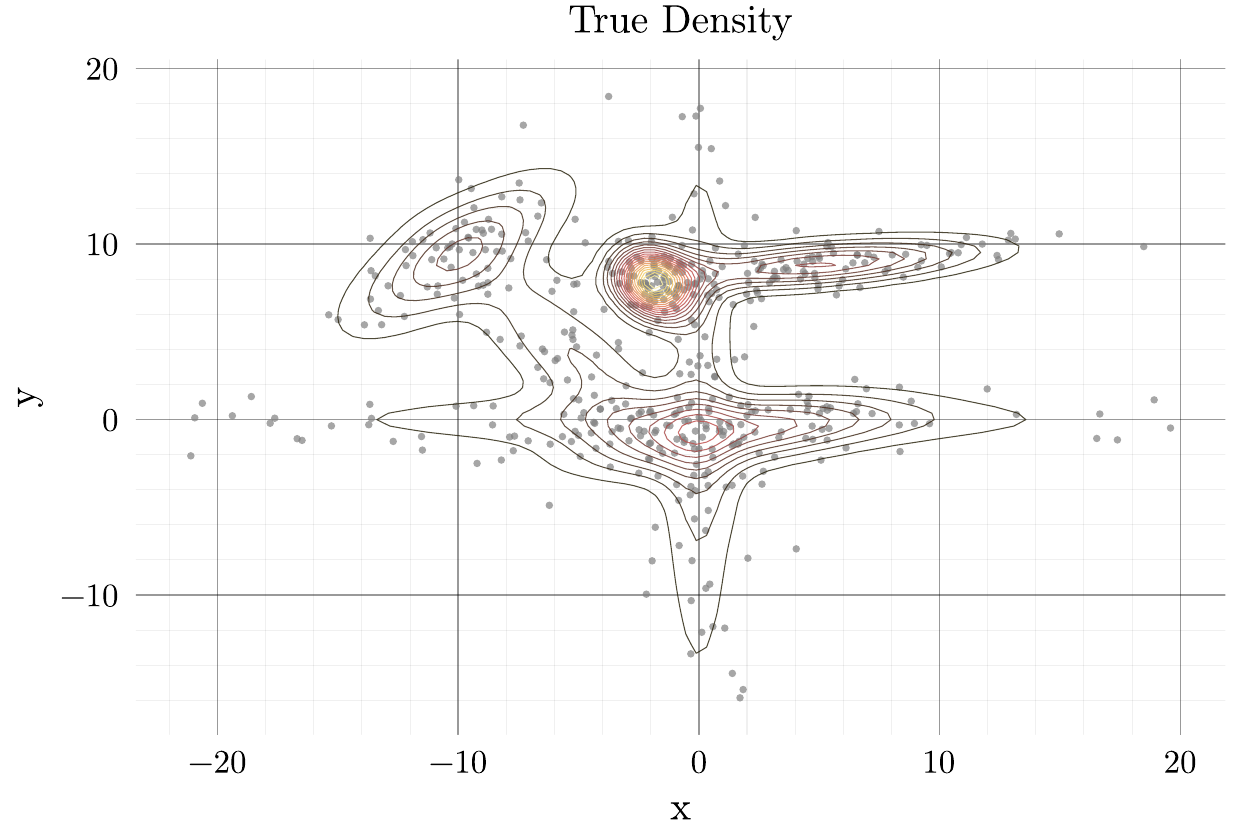} 
    } \quad
    \subfloat{
        \includegraphics[width=0.46\linewidth]{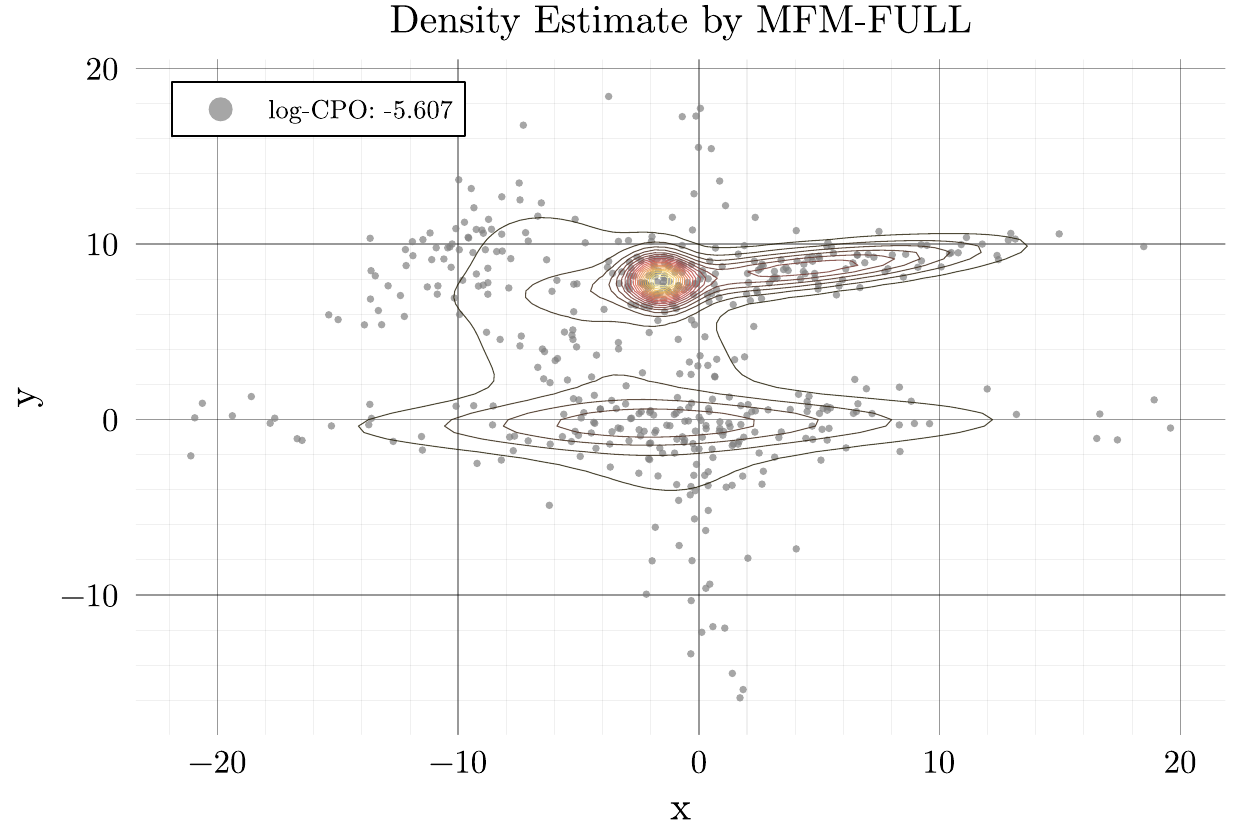} 
    }
    \\
    \subfloat{
        \includegraphics[width=0.46\linewidth]{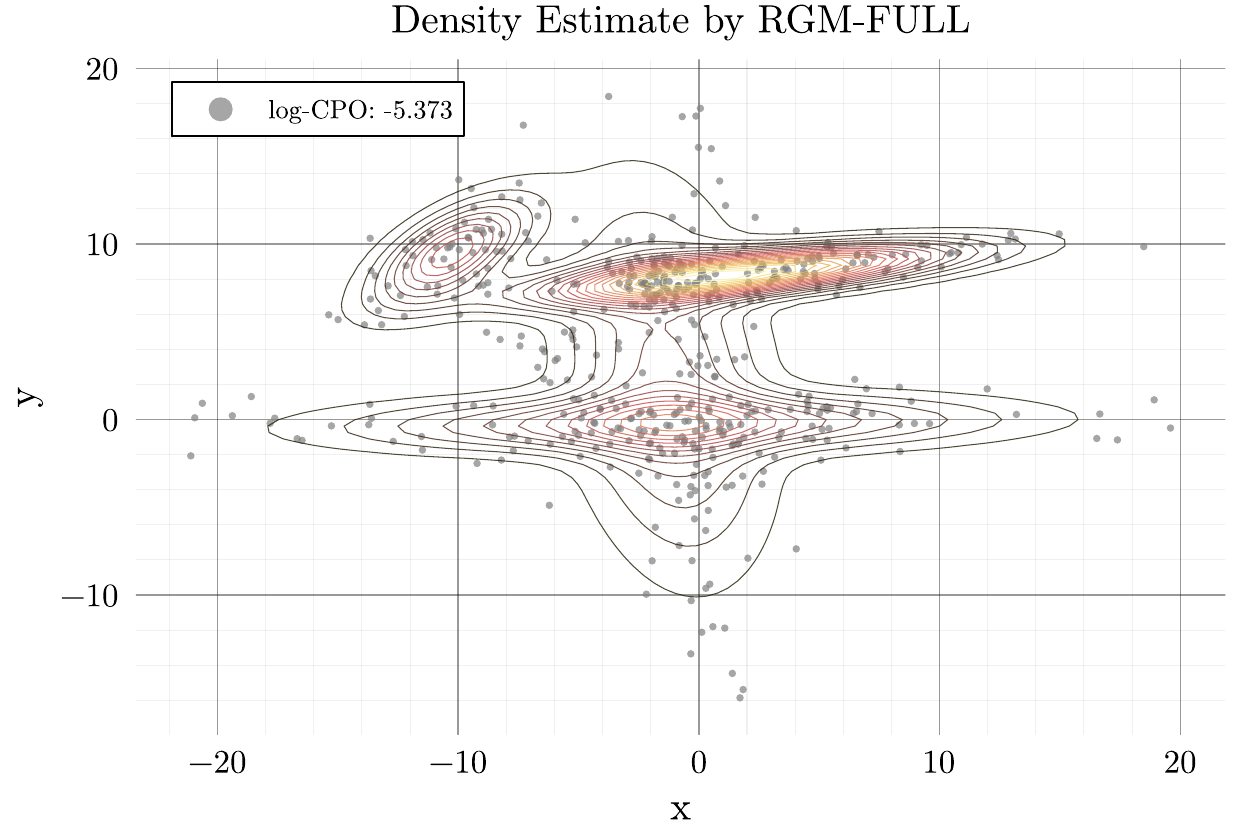} 
    } \quad 
    \subfloat{
        \includegraphics[width=0.46\linewidth]{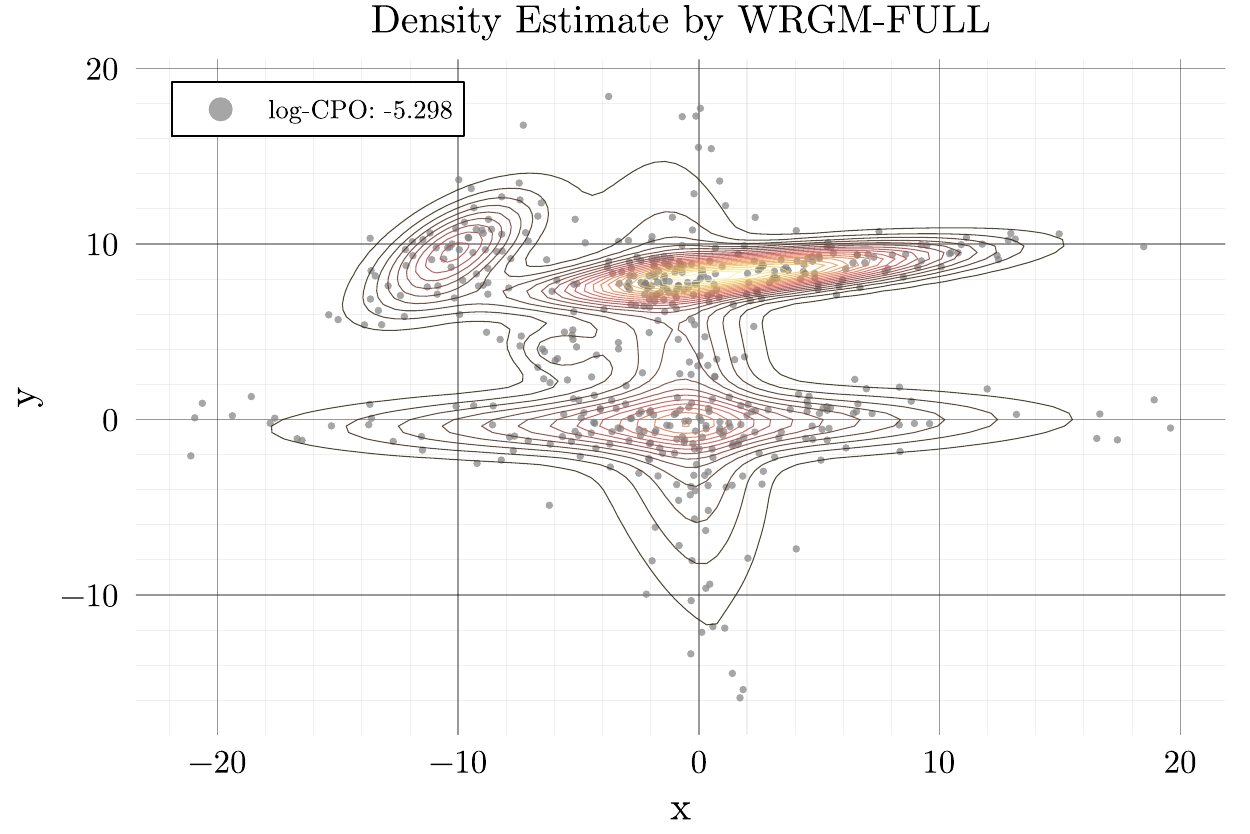} 
    }  
    \caption{Comparison of the true density with the densities estimated by the three models in Simulation 3.}
    \label{fig:sim3_de}
\end{figure}

\begin{figure}[ht]
    \centering
    \subfloat{
        \includegraphics[width=0.46\linewidth]{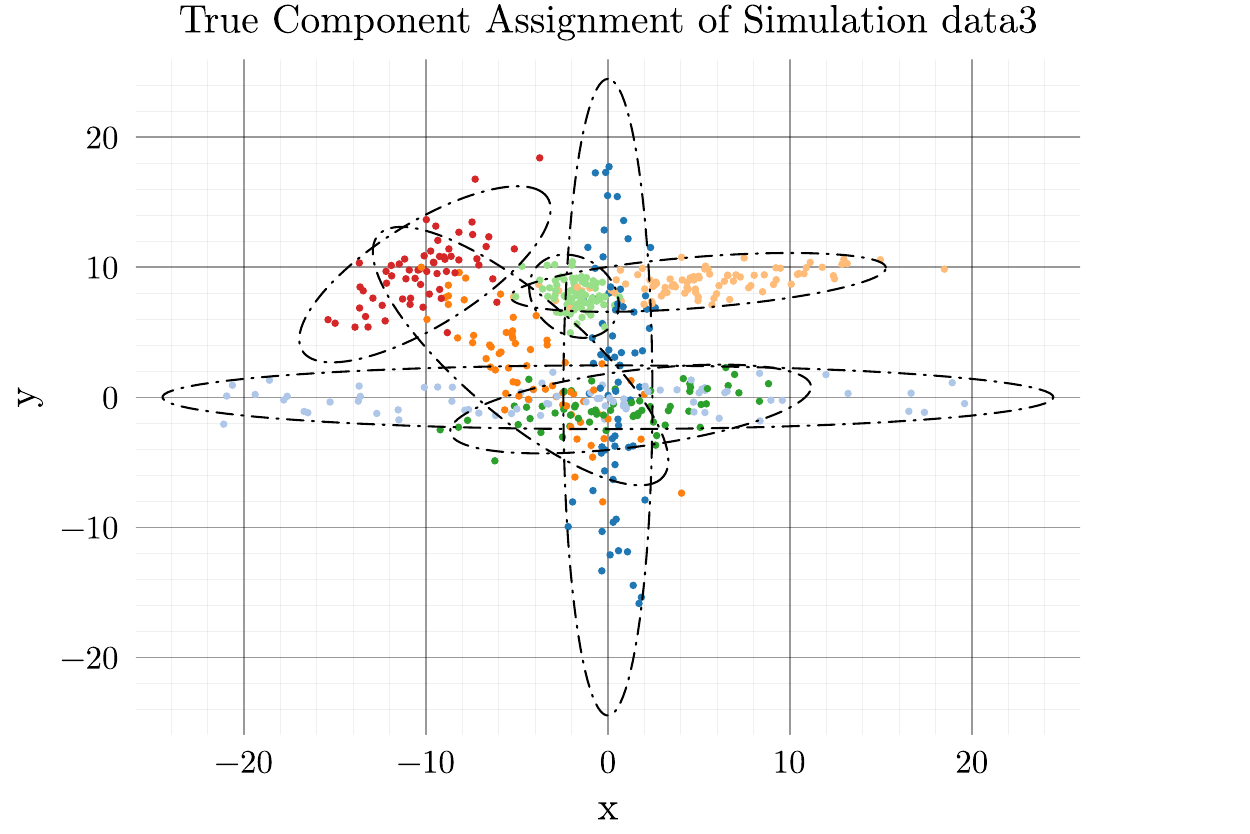} 
    } \quad 
    \subfloat{
        \includegraphics[width=0.46\linewidth]{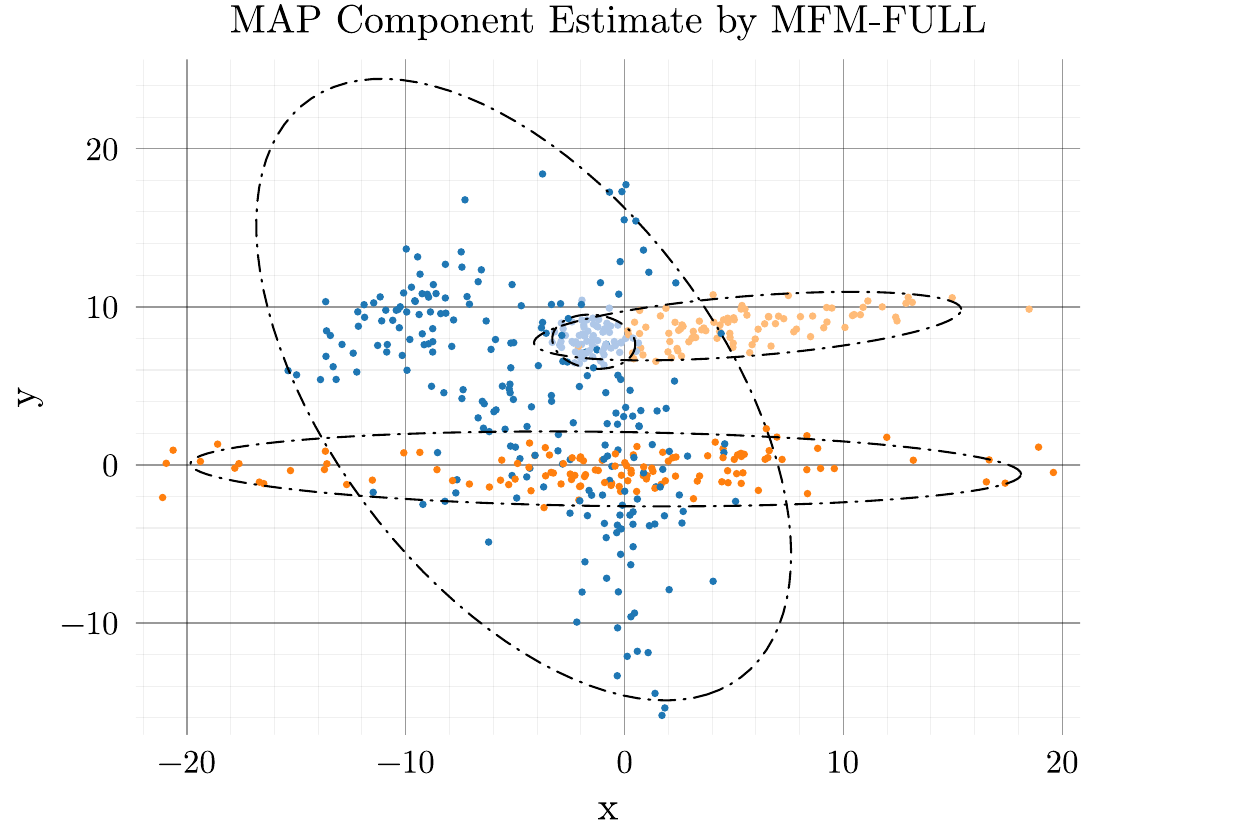} 
    }
    \\
    \subfloat{
        \includegraphics[width=0.46\linewidth]{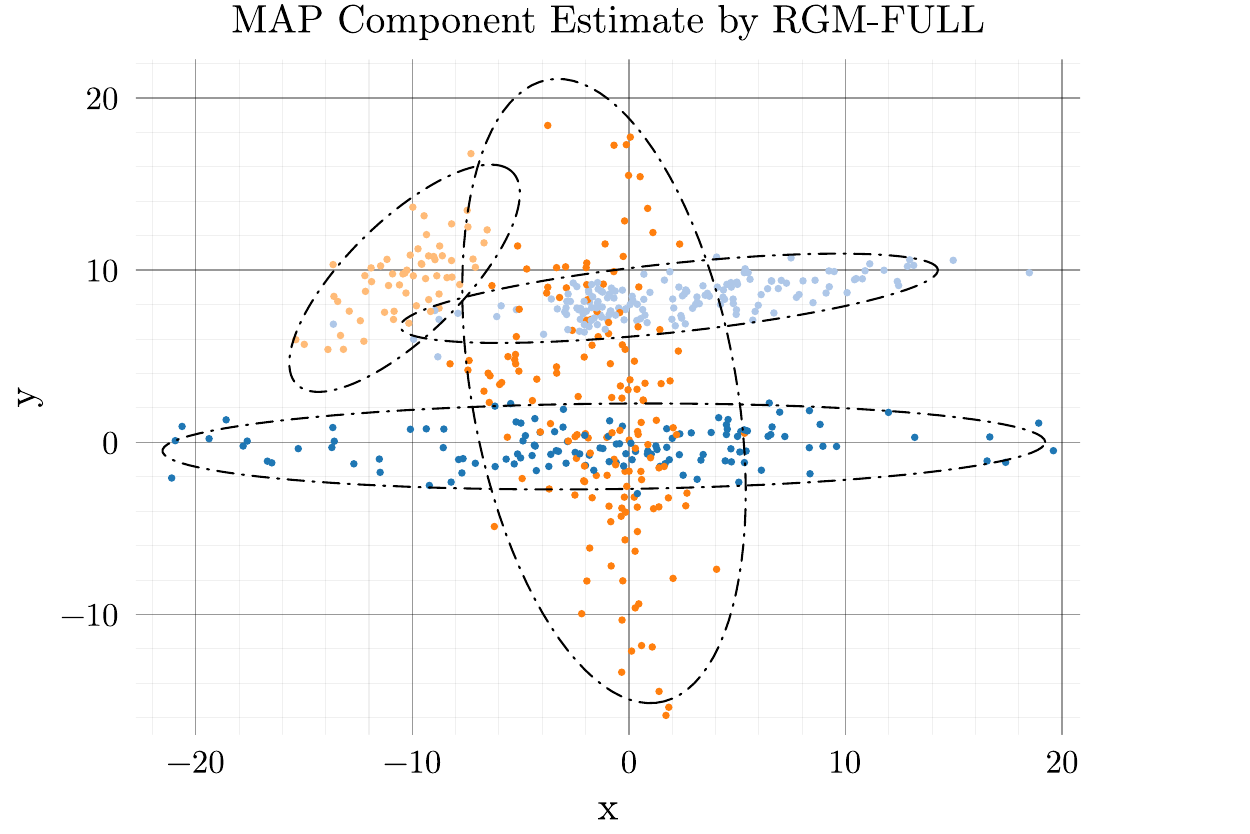} 
    } \quad 
    \subfloat{
        \includegraphics[width=0.46\linewidth]{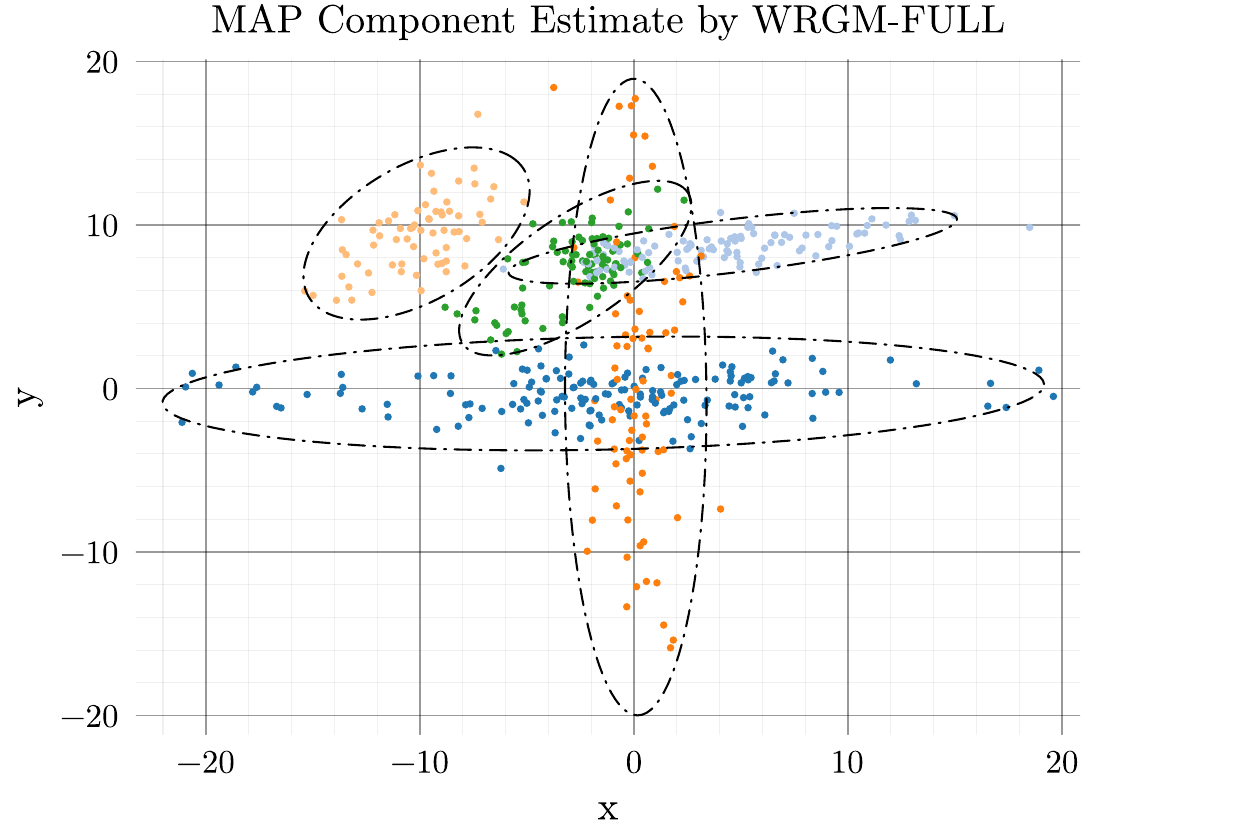} 
    }  
    \caption{Maximum a Posteriori (MAP) component assignments for the three models in Simulation 3.}
    \label{fig:sim3_map}
\end{figure} 

\begin{figure}
    \centering
    \subfloat{
        \includesvg[width=0.6\linewidth]{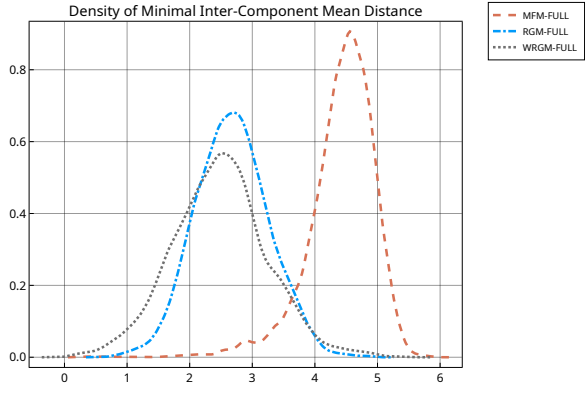} 
    } 
    \caption{Minimum pairwise distances between the component mean parameters for each of the three models in Simulation 3.}
    \label{fig:sim3_min_d}
\end{figure}

\section{Data Applications}
\label{sec_data_app}
We apply the WRGM model to two datasets to evaluate its performance and compare it with the RGM and MFM models. The first dataset, known as the A1 dataset and originally introduced by~\cite{Asets2002}, is a two-dimensional dataset with a sample size of 3,000. It has been widely used in the clustering literature~\cite{denoeux2015ek,du2016study,wang2016automatic,gu2019distance,ros2019munec,zhang2023density}. The second dataset is the Graft-versus-Host Disease (GvHD) flow cytometry dataset introduced in~\citep{brinkman2007high}. It contains measurements on \( p = 4 \) markers: CD3, CD4, CD8b, and CD8. Following the preprocessing approach in~\cite{fuquene2019choosing}, we retained only observations where CD3 \( > 300 \), resulting in a subset of 1,802 samples. For clarity of visualization and interpretation, we focused our analysis on the (CD8, CD4) marker pair.
\\\\Different prior hyperparameters are used for the two datasets due to their varying scales. For the A1 dataset, we set the Dirichlet concentration parameter to \( \beta = 1 \). The prior for the mean parameters of the mixture components, \( p_m \), follows a normal distribution with zero mean and covariance \( \tau^2 I \), where \( \tau = 10^5 \). The repulsion function is \( g(x) = \frac{x}{x + g_0} \), with \( g_0 = 100 \). For the truncated inverse-Wishart prior on the covariance matrices, we set \( \nu = 4 \), and the eigenvalue bounds are \( (10^{-12}, 10^{12}) \). The prior on \( K \) is a zero-truncated Poisson distribution with \( \lambda = 1 \). For the GvHD dataset, we maintain the same prior structure, but with different hyperparameters: \( \beta = 1 \), \( \tau = 100 \), \( g_0 = 50 \), \( \nu = 4 \), and \( \lambda = 1 \).
\\\\
In addition to fitting the three models with full covariance matrices, we also consider their diagonal covariance counterparts for comparison, resulting in a total of six models. For each model, the MCMC is run for 10,000 burn-in iterations, followed by the 5,000 posterior samples.
\\\\
For the A1 dataset, the estimated densities, MAP component assignments, and the minimum pairwise distances between the mean parameters of the components under these six models are presented in Figures~\ref{fig:a1_de}, \ref{fig:a1_map}, and \ref{fig:a1_min_d}, respectively. We observe that both RGM and WRGM outperform MFM in terms of log-CPO under both diagonal and full covariance specifications. WRGM also demonstrates a slight improvement over RGM. Interestingly, both RGM and WRGM tend to produce a larger number of components compared to MFM. The MAP component allocations from RGM and WRGM are generally consistent with each other, but differ noticeably from those obtained under MFM. As shown in Figure~\ref{fig:a1_min_d}, the minimum pairwise distances between the mean parameters of mixture components under WRGM tend to be smaller than those under the other two models. This indicates that WRGM allows components to be closer together in location, which may partially explain why the estimated densities under WRGM—both with diagonal and full covariance structures—contain more components.
\\\\
For the GvHD dataset, the estimated densities, MAP component assignments, and the minimum pairwise distances between the mean parameters of the components for the six models are shown in Figures~\ref{fig:gvhd_de}, \ref{fig:gvhd_map}, and \ref{fig:gvhd_min_d}, respectively. In this case, we find that the estimated densities and MAP assignments are largely consistent across models, likely reflecting the relatively simple structure of the dataset. However, some differences arise when comparing the diagonal covariance models: specifically, three components identified by WRGM appear to be merged into a single component under RGM. This discrepancy may be attributed to the repulsive prior in RGM penalizing the close proximity of component means more strongly than WRGM.

\begin{figure}[ht]
    \centering
        \subfloat{
        \includegraphics[width=0.46\linewidth]{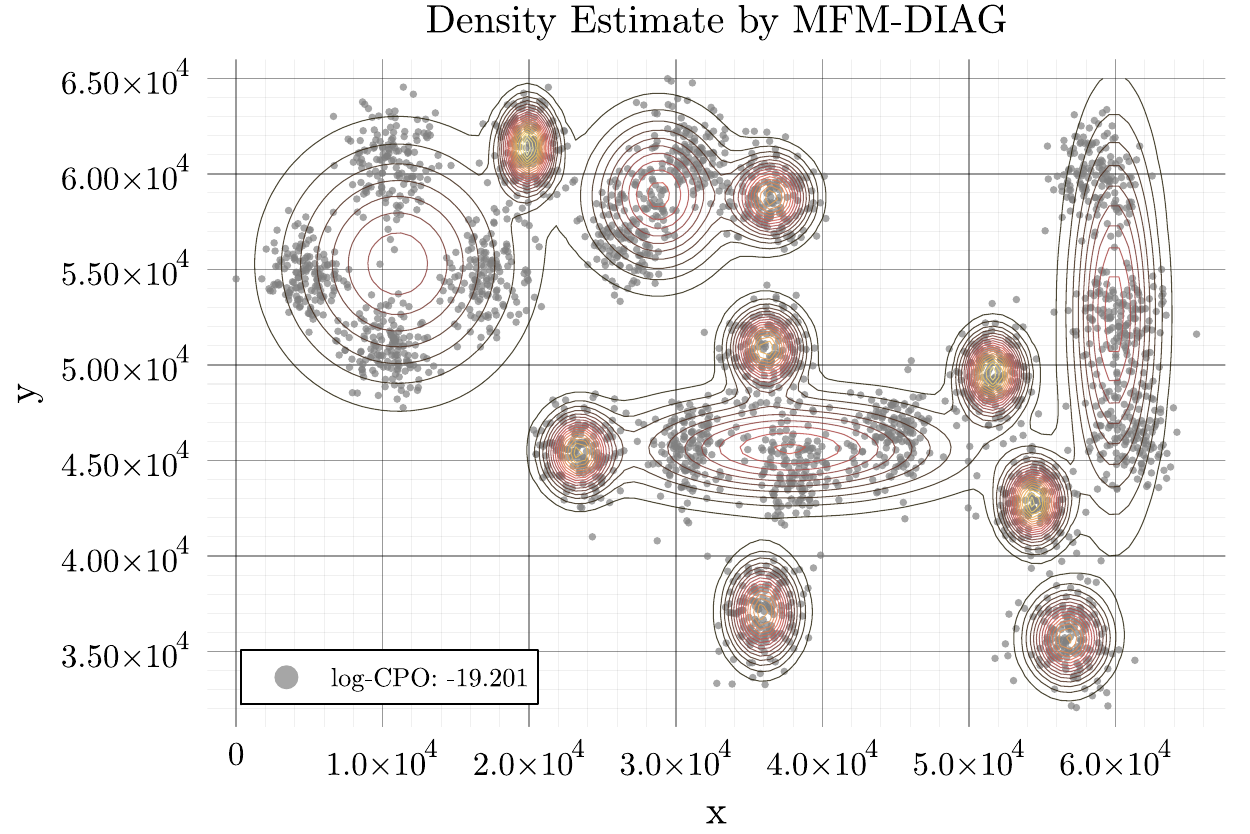} 
    } \quad
    \subfloat{
        \includegraphics[width=0.46\linewidth]{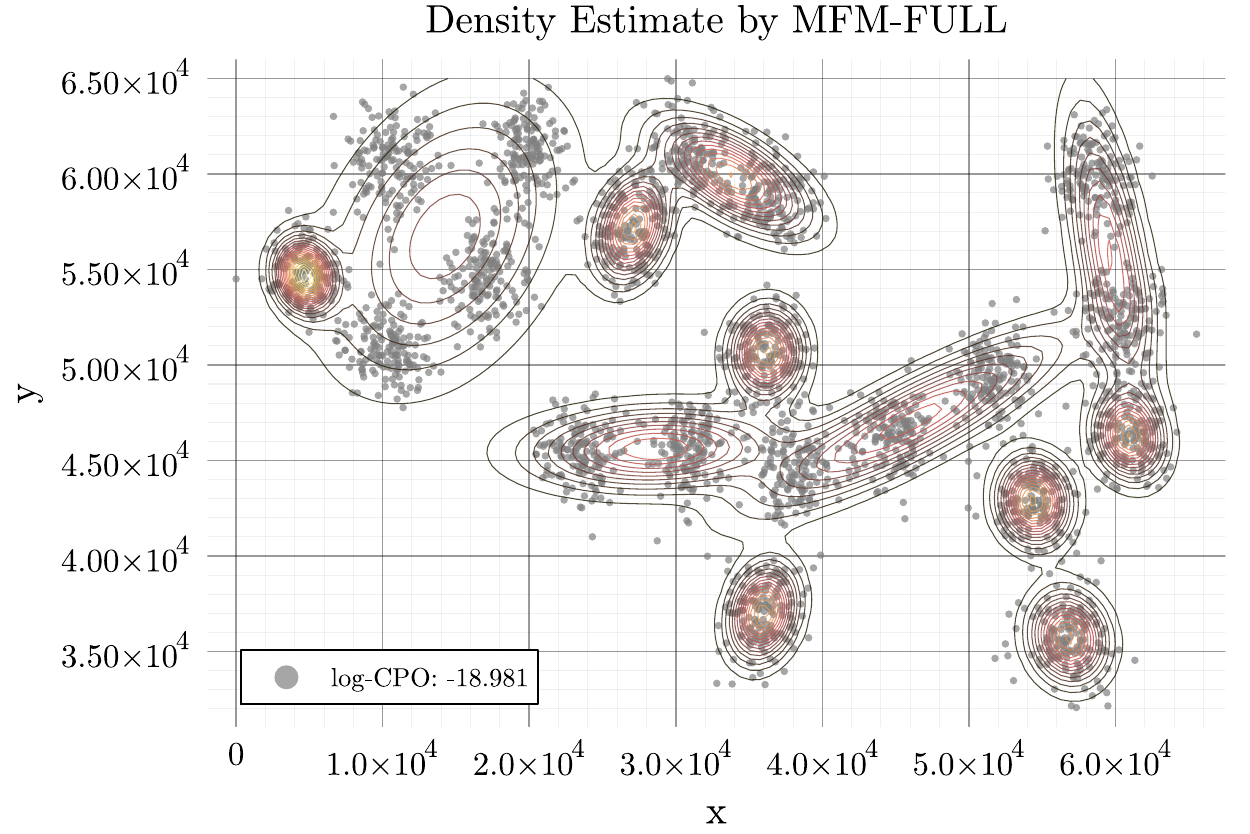} 
    }
    \\
    \subfloat{
        \includegraphics[width=0.46\linewidth]{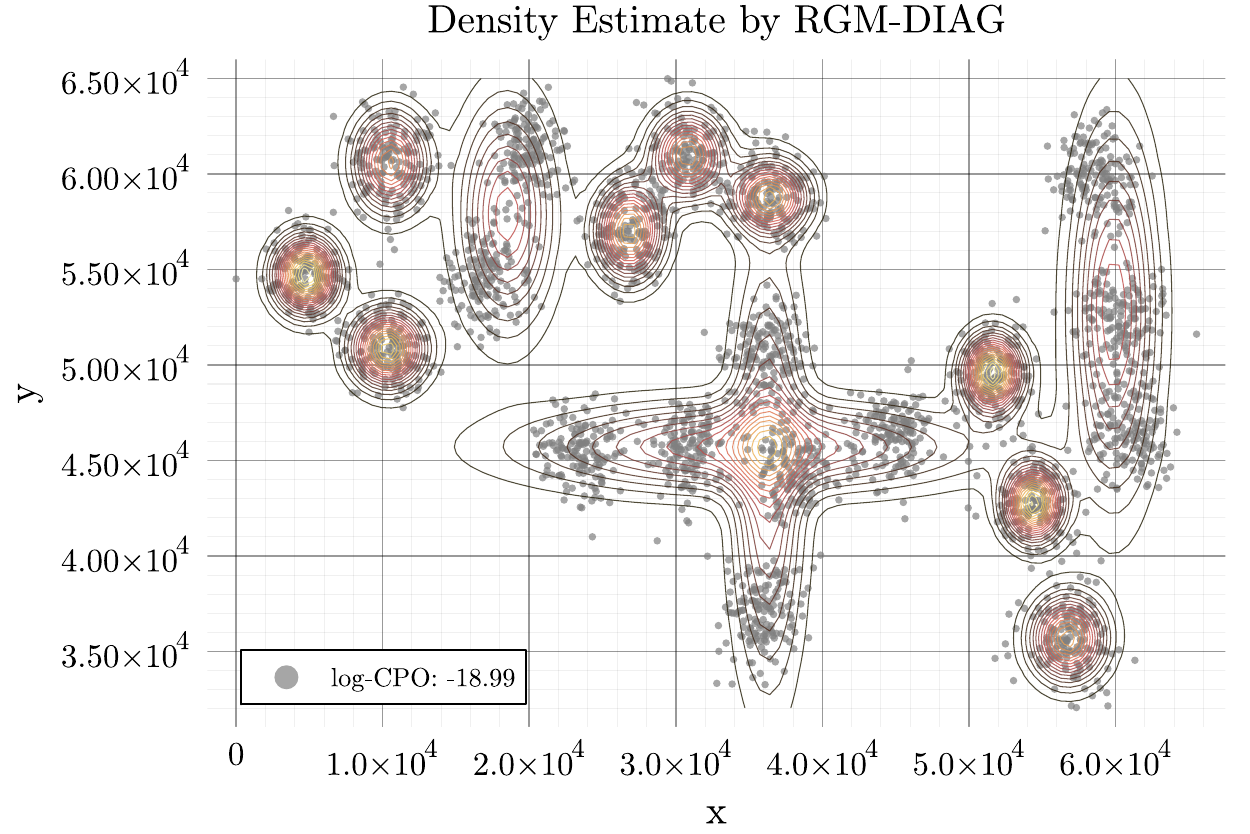} 
    } \quad
    \subfloat{
        \includegraphics[width=0.46\linewidth]{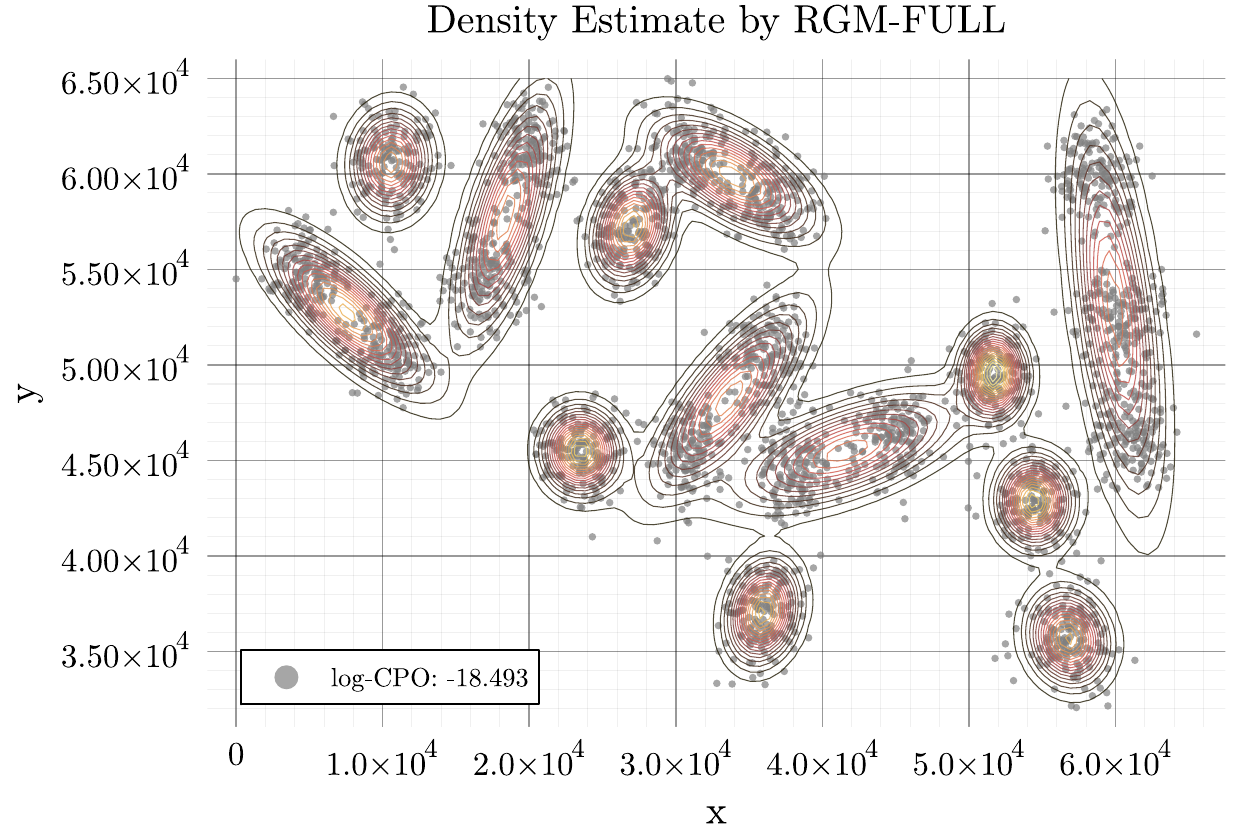} 
    } 
    \\
    \subfloat{
        \includegraphics[width=0.46\linewidth]{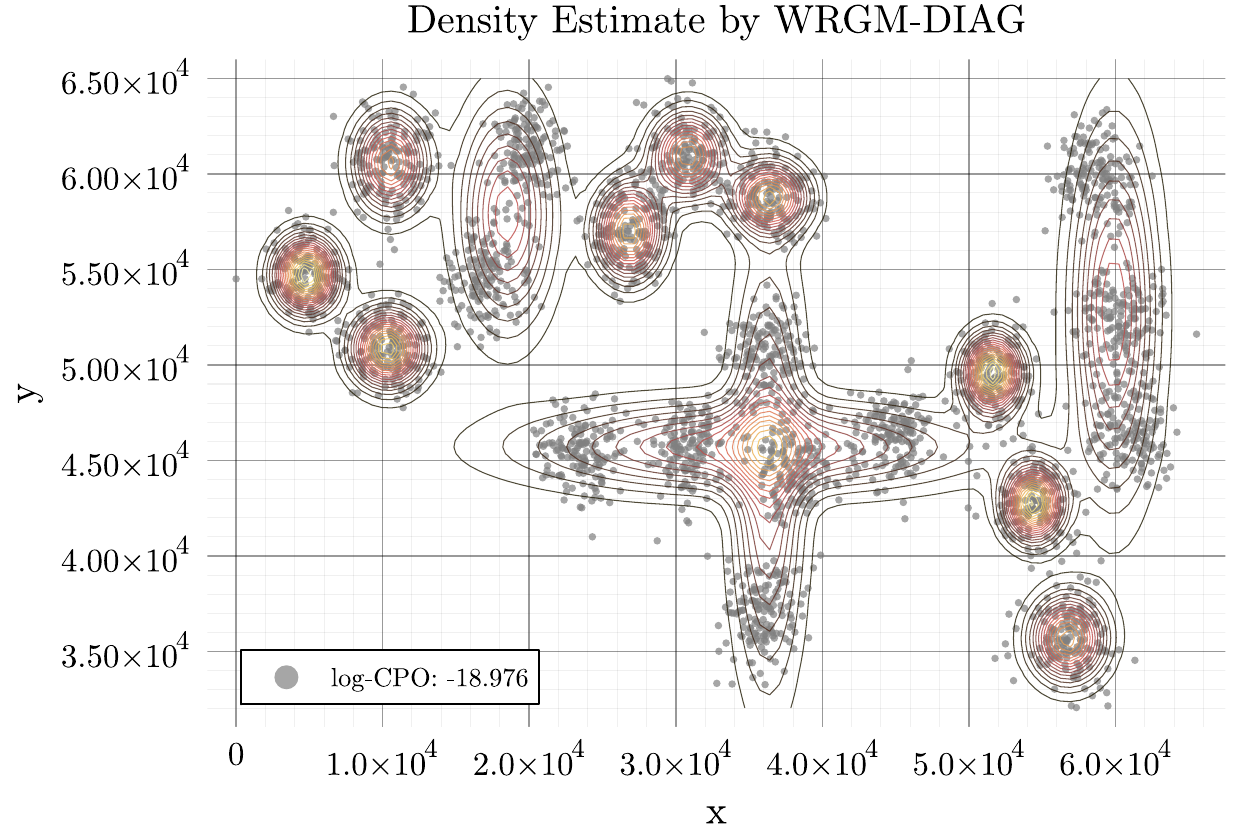} 
    } \quad
    \subfloat{
        \includegraphics[width=0.46\linewidth]{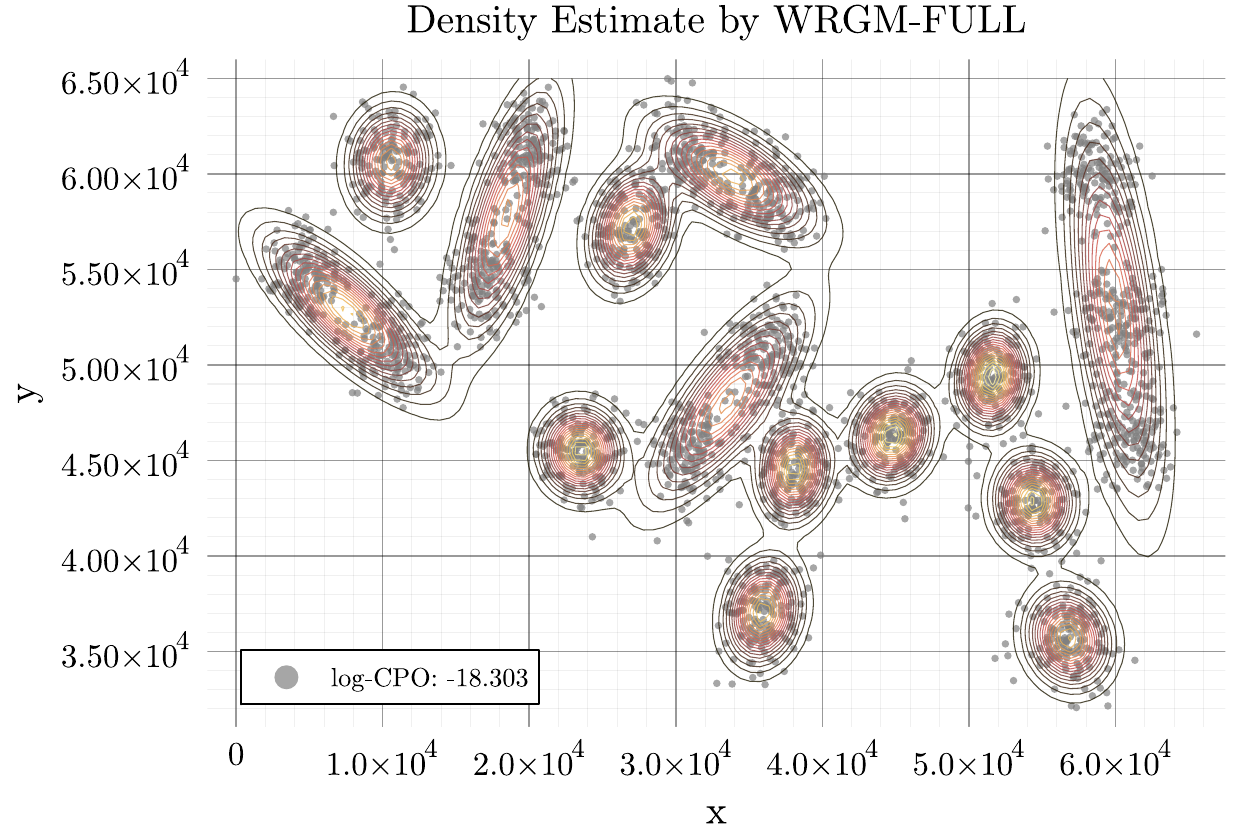} 
    }  
    \caption{Estimated densities by the six models for the A1 dataset}
    \label{fig:a1_de}
\end{figure}

\begin{figure}[ht]
    \centering
    \subfloat{
        \includegraphics[width=0.46\linewidth]{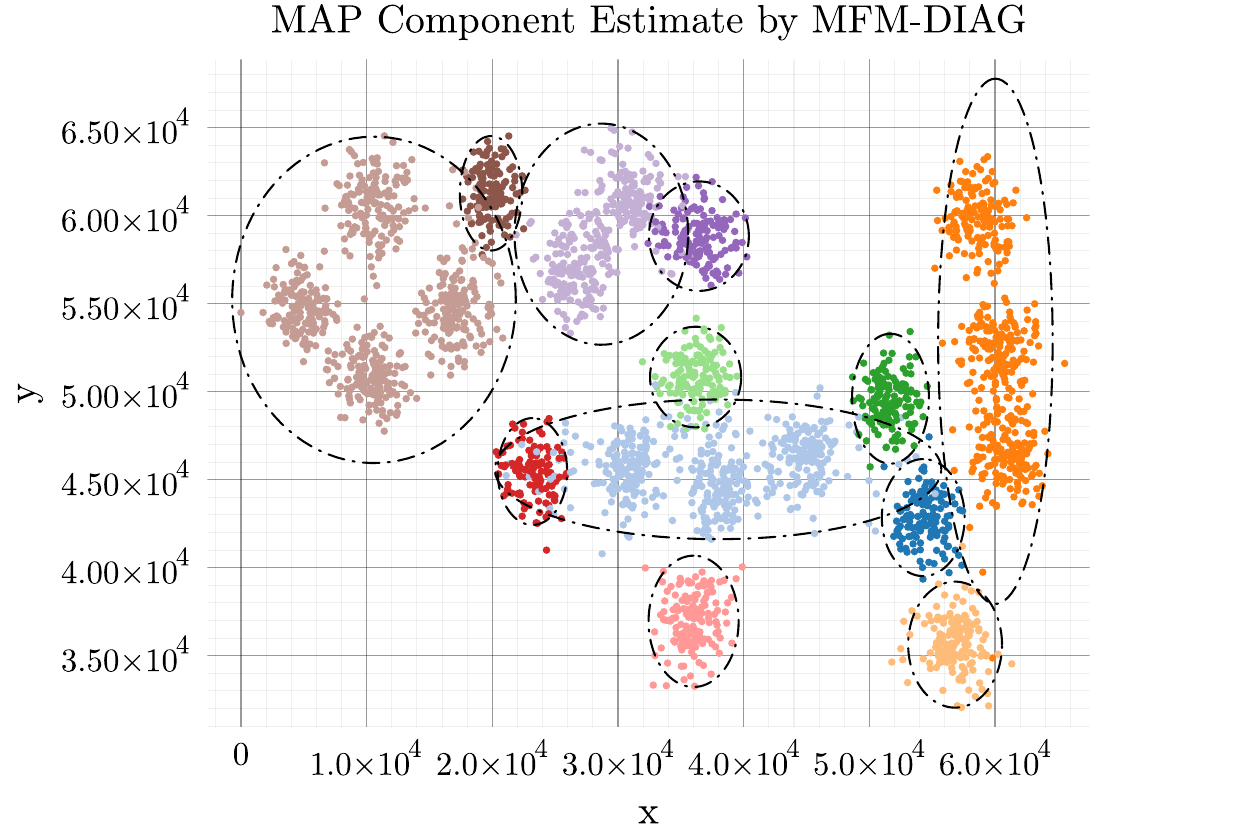} 
    } \quad 
    \subfloat{
        \includegraphics[width=0.46\linewidth]{figures/a1/a1_MFM-DIAG_map.pdf} 
    }
    \\
    \subfloat{
        \includegraphics[width=0.46\linewidth]{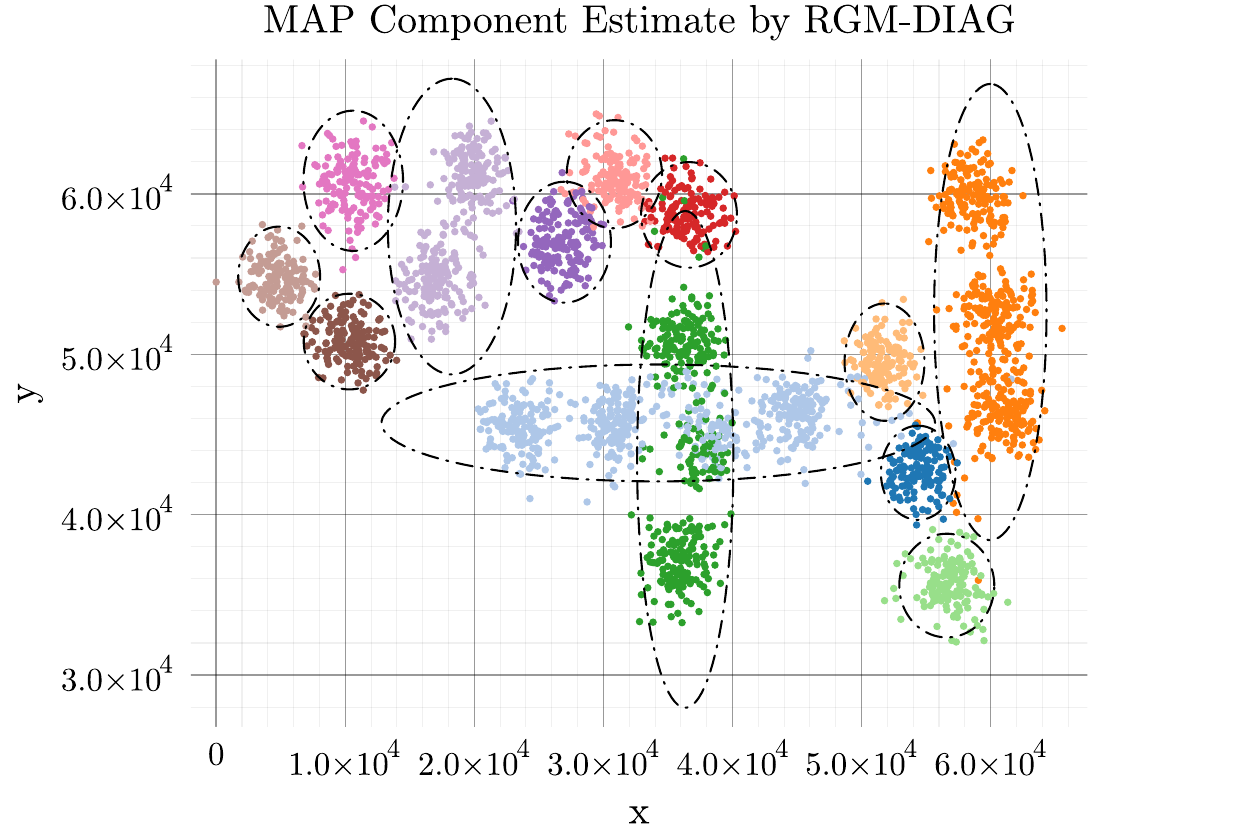} 
    } \quad
    \subfloat{
        \includegraphics[width=0.46\linewidth]{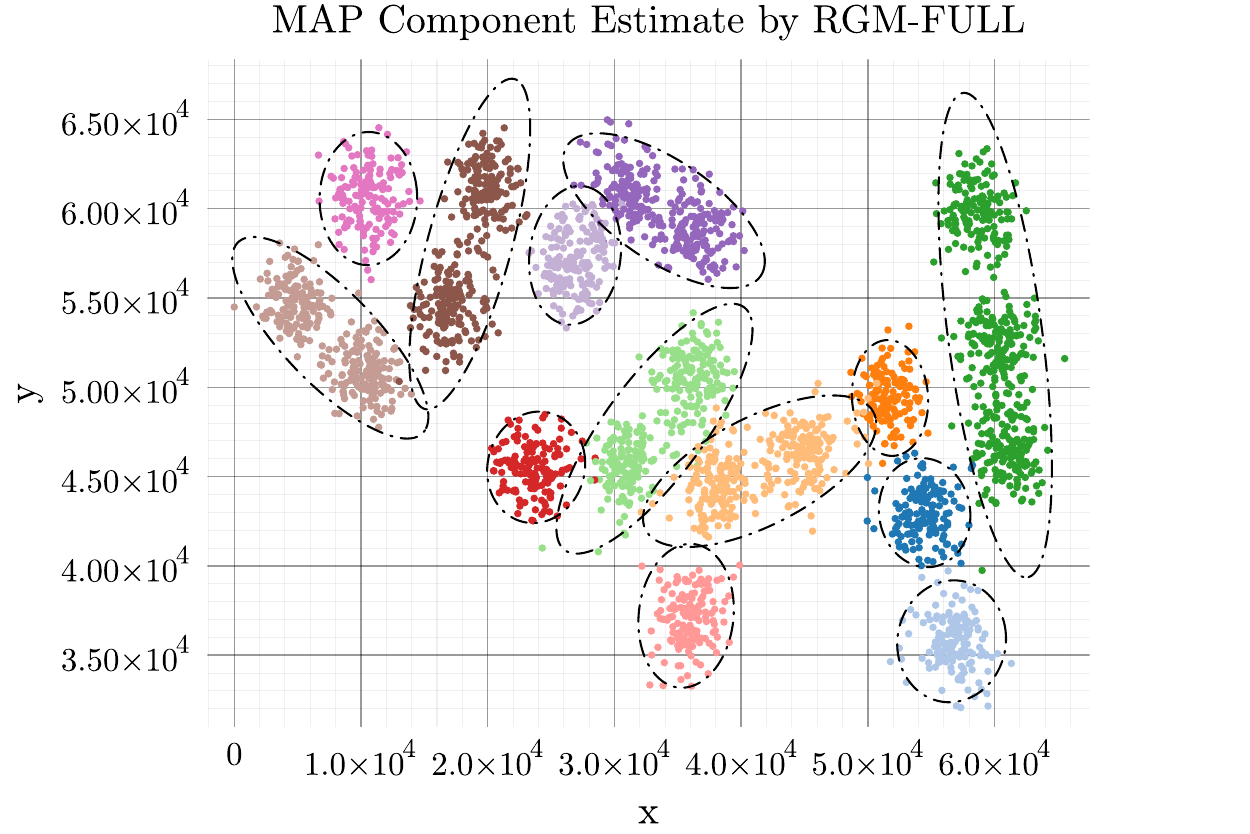} 
    } 
    \\
    \subfloat{
        \includegraphics[width=0.46\linewidth]{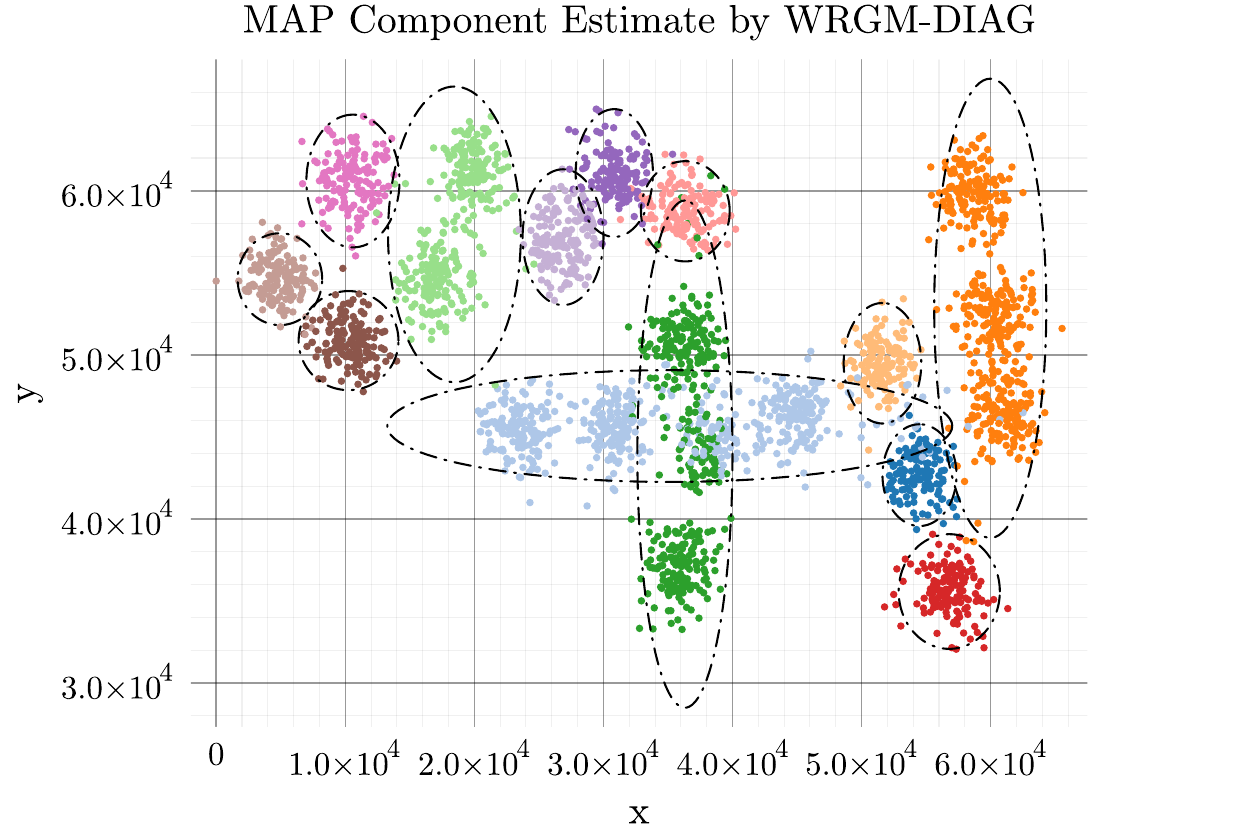} 
    } \quad
    \subfloat{
        \includegraphics[width=0.46\linewidth]{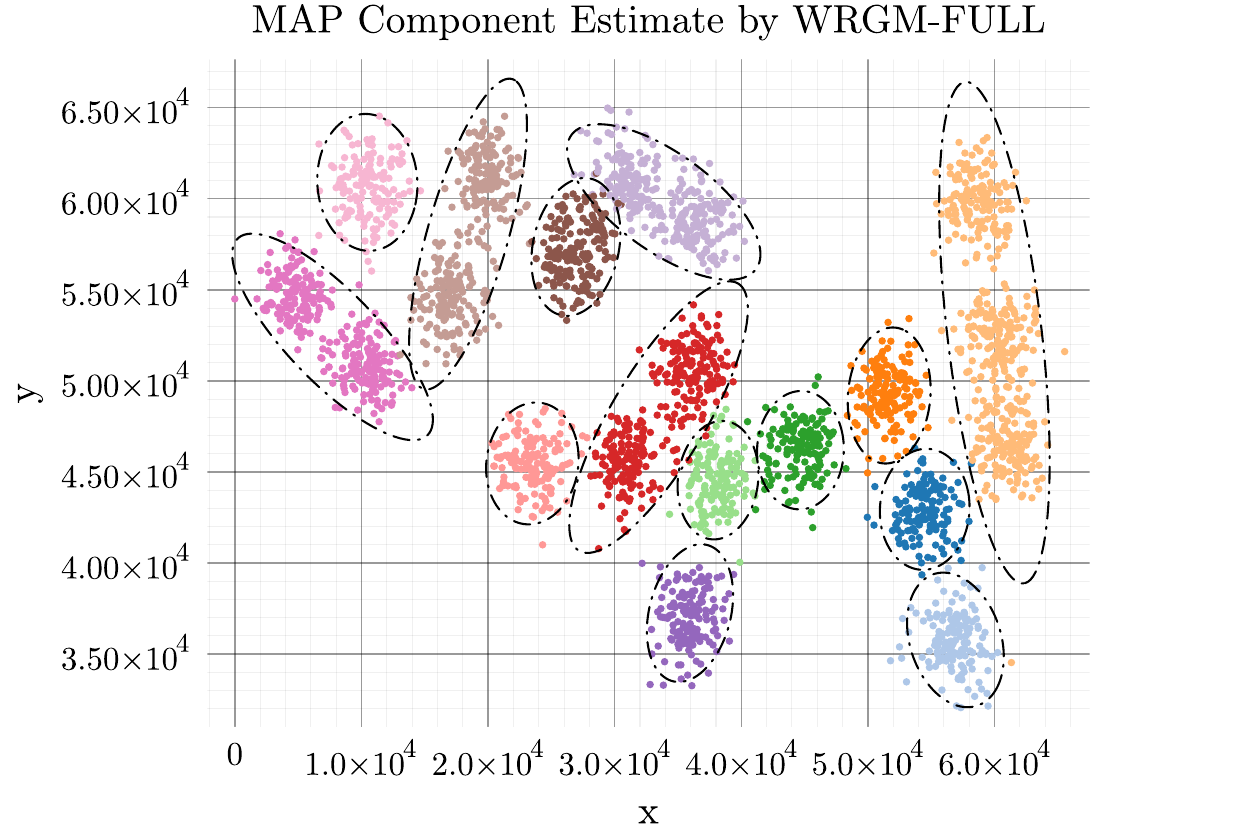} 
    }  
    \caption{MAP Component Assignments for the six models for the A1 dataset}
    \label{fig:a1_map}
\end{figure}

\begin{figure}
    \centering
    \subfloat{
        \includesvg[width=0.46\linewidth]{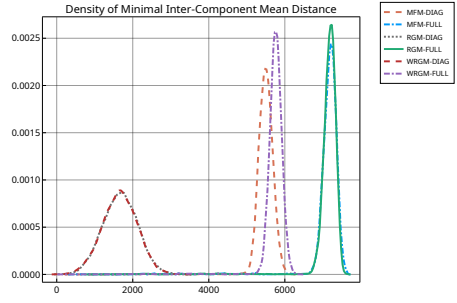} 
    }  
      
    \caption{Minimum pairwise distances between the component mean parameters for each of the six models for the A1 dataset.}
    \label{fig:a1_min_d}
\end{figure}

\begin{figure}[ht]
    \centering
        \subfloat{
        \includegraphics[width=0.46\linewidth]{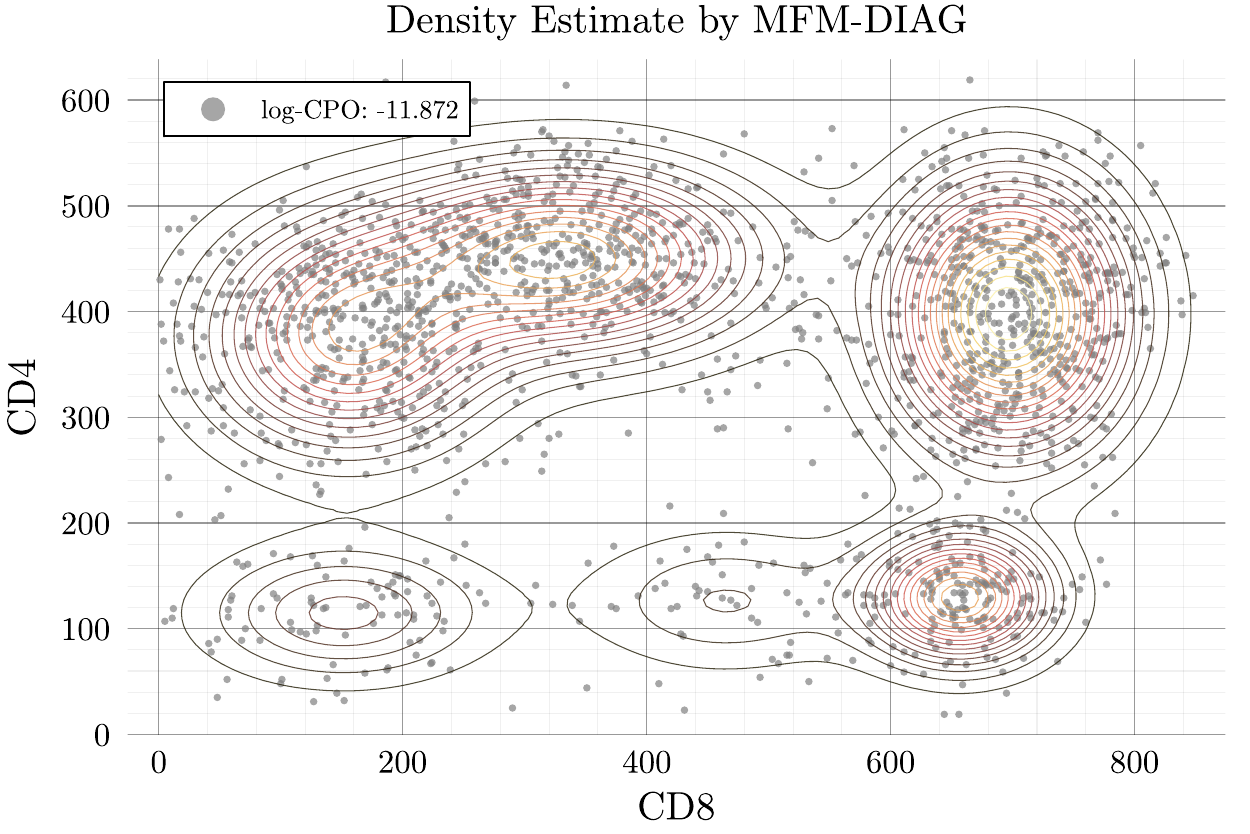} 
    } \quad
    \subfloat{
        \includegraphics[width=0.46\linewidth]{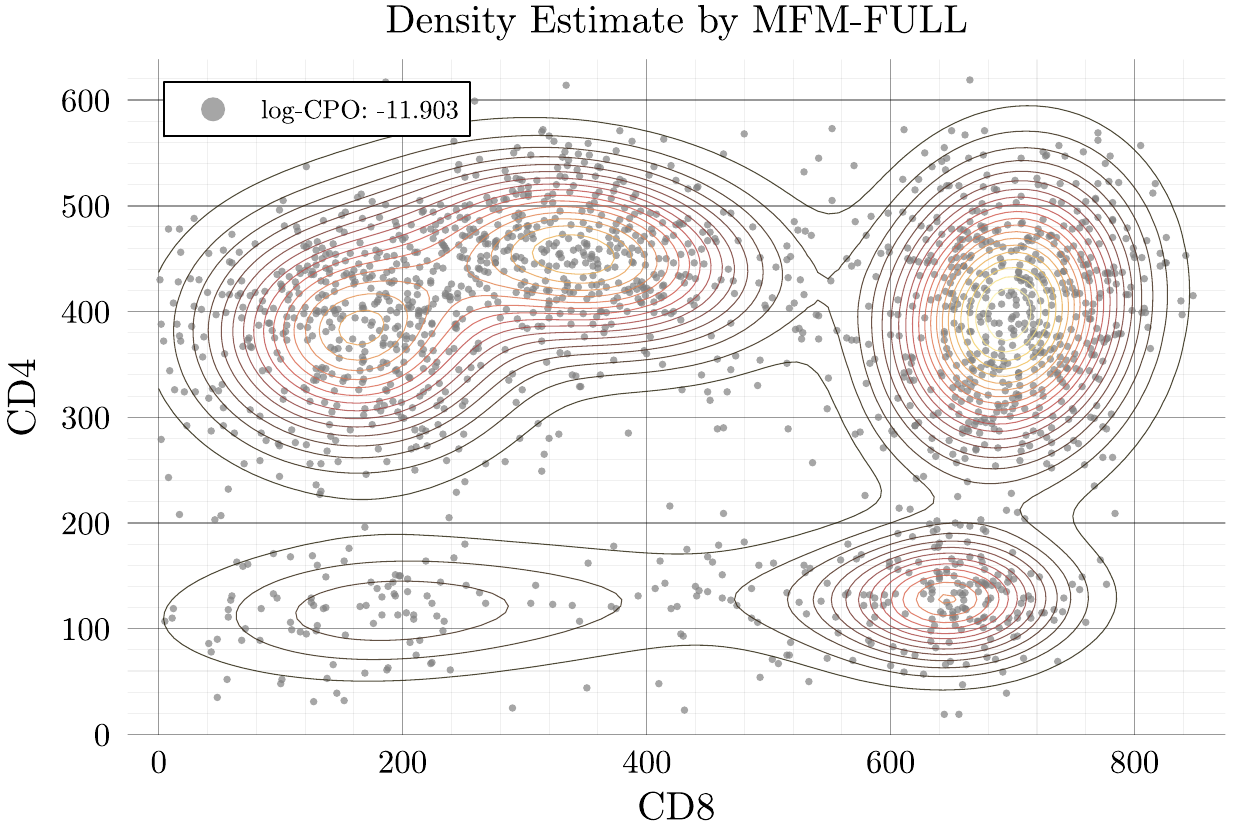} 
    }
    \\
    \subfloat{
        \includegraphics[width=0.46\linewidth]{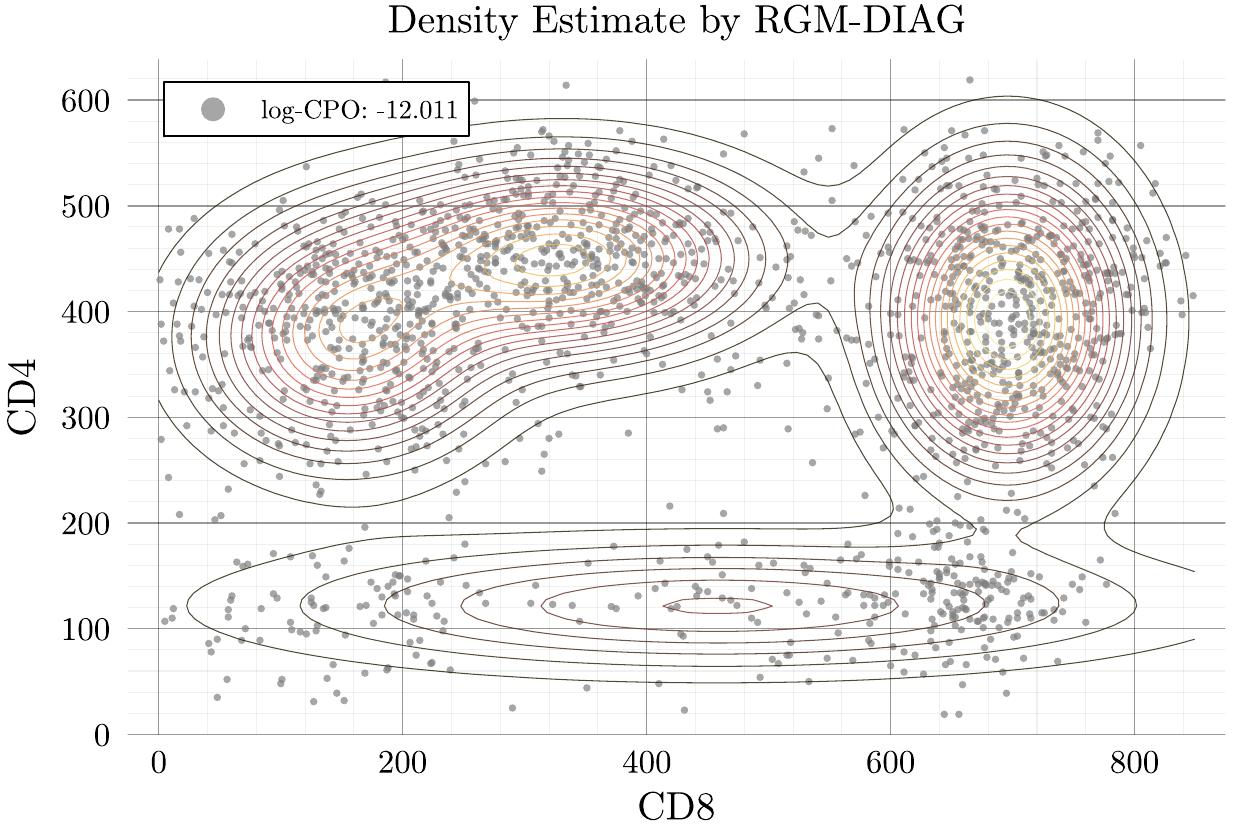} 
    } \quad
    \subfloat{
        \includegraphics[width=0.46\linewidth]{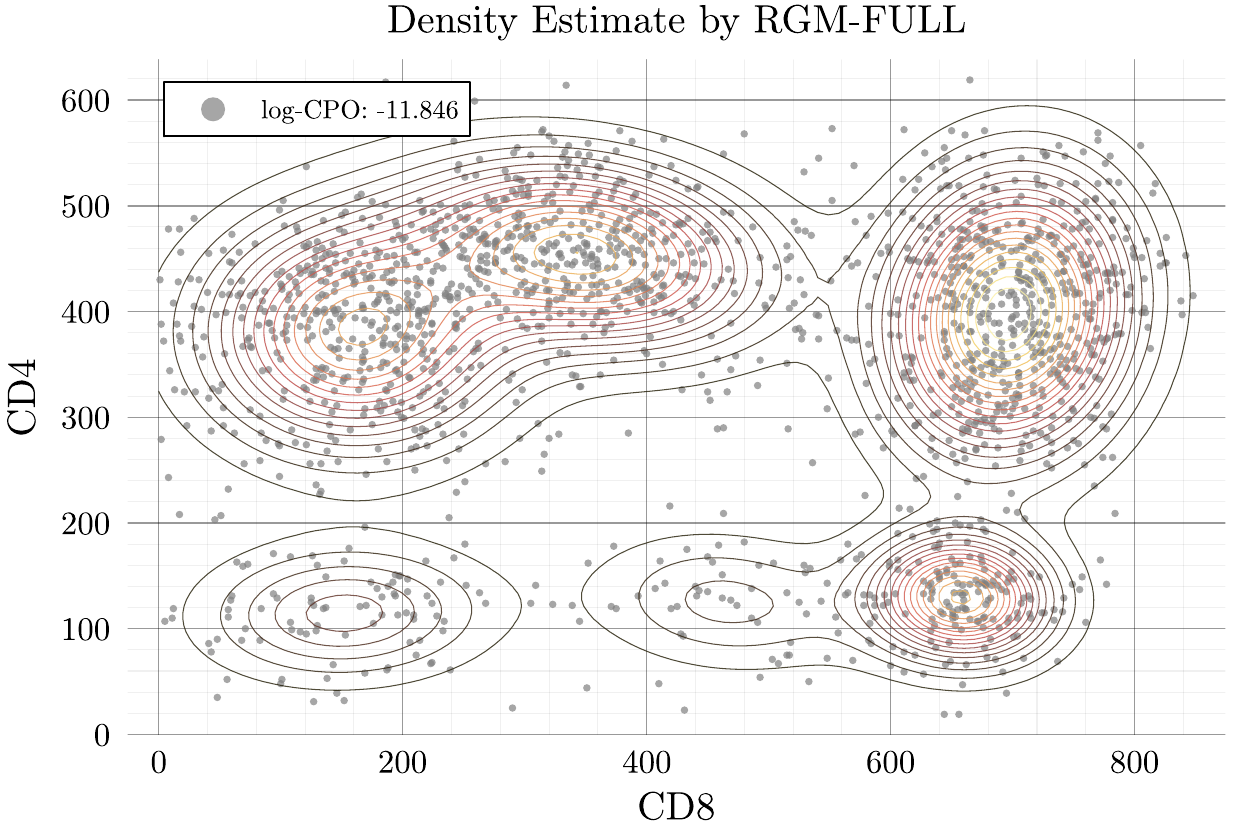} 
    } 
    \\
    \subfloat{
        \includegraphics[width=0.46\linewidth]{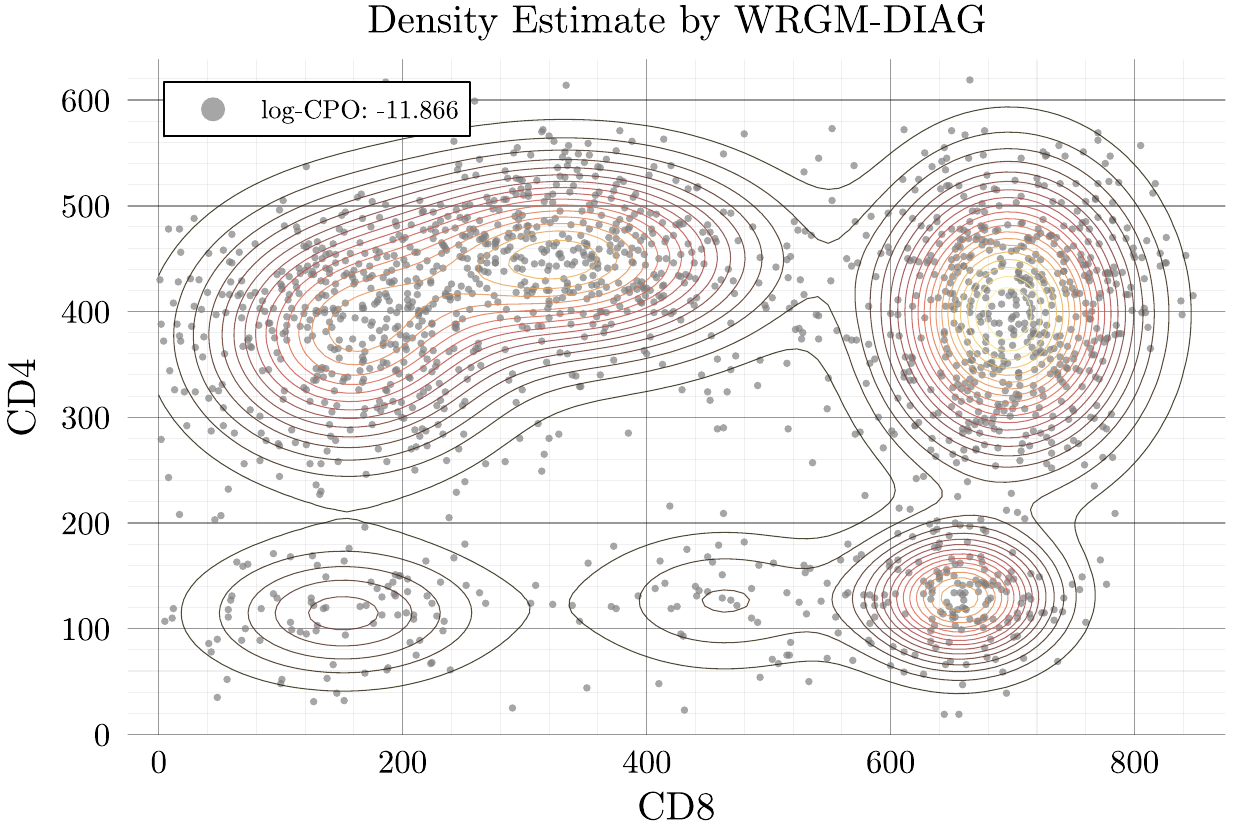} 
    } \quad
    \subfloat{
        \includegraphics[width=0.46\linewidth]{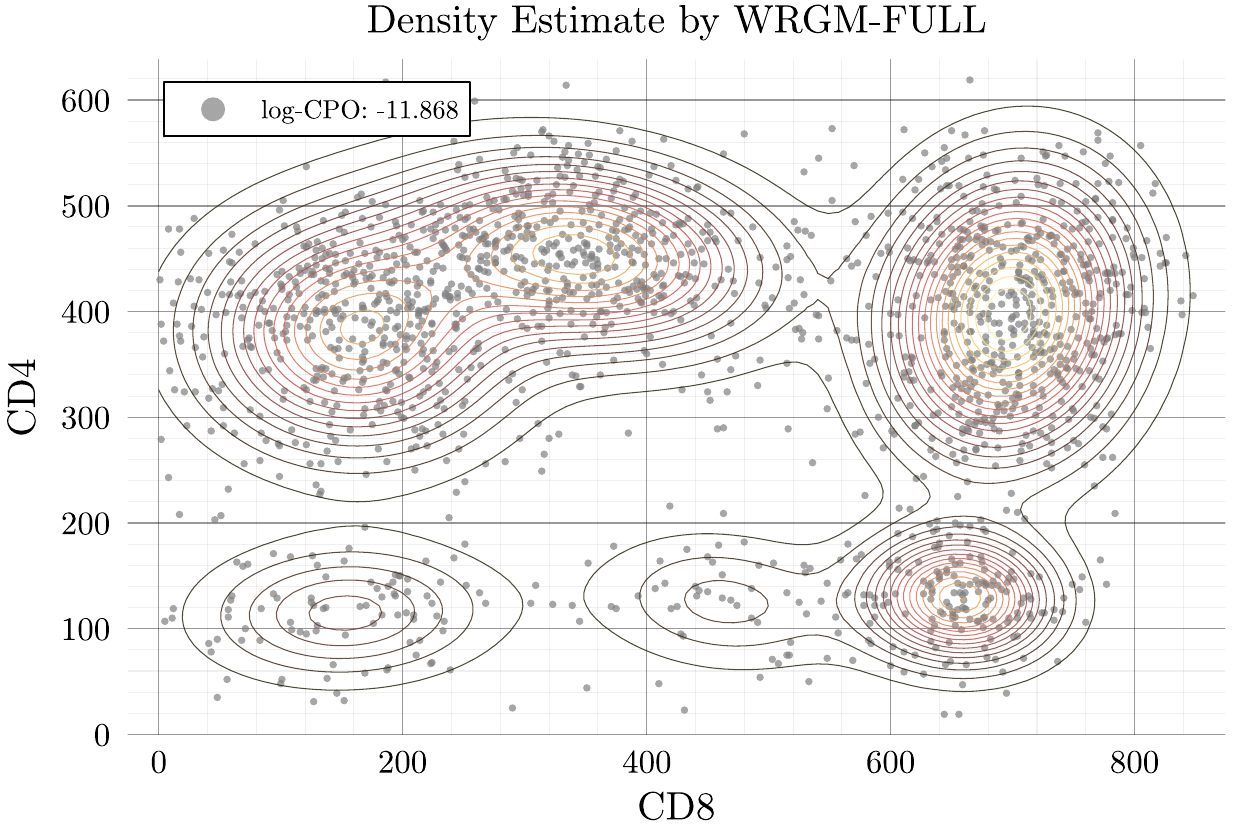} 
    }  
    \caption{Estimated densities by the six models for the GvHD dataset}
    \label{fig:gvhd_de}
\end{figure} 

\begin{figure}[ht]
    \centering
    \subfloat{
        \includegraphics[width=0.46\linewidth]{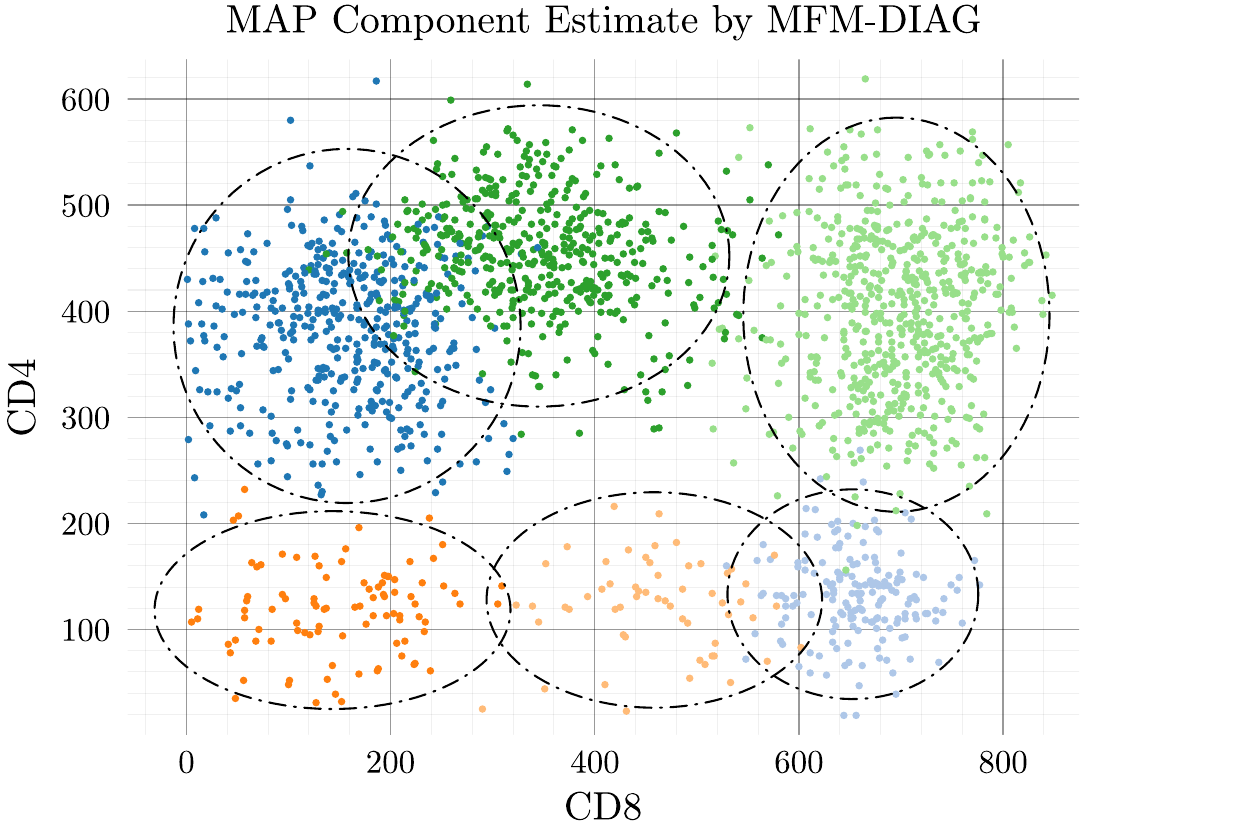} 
    } \quad 
    \subfloat{
        \includegraphics[width=0.46\linewidth]{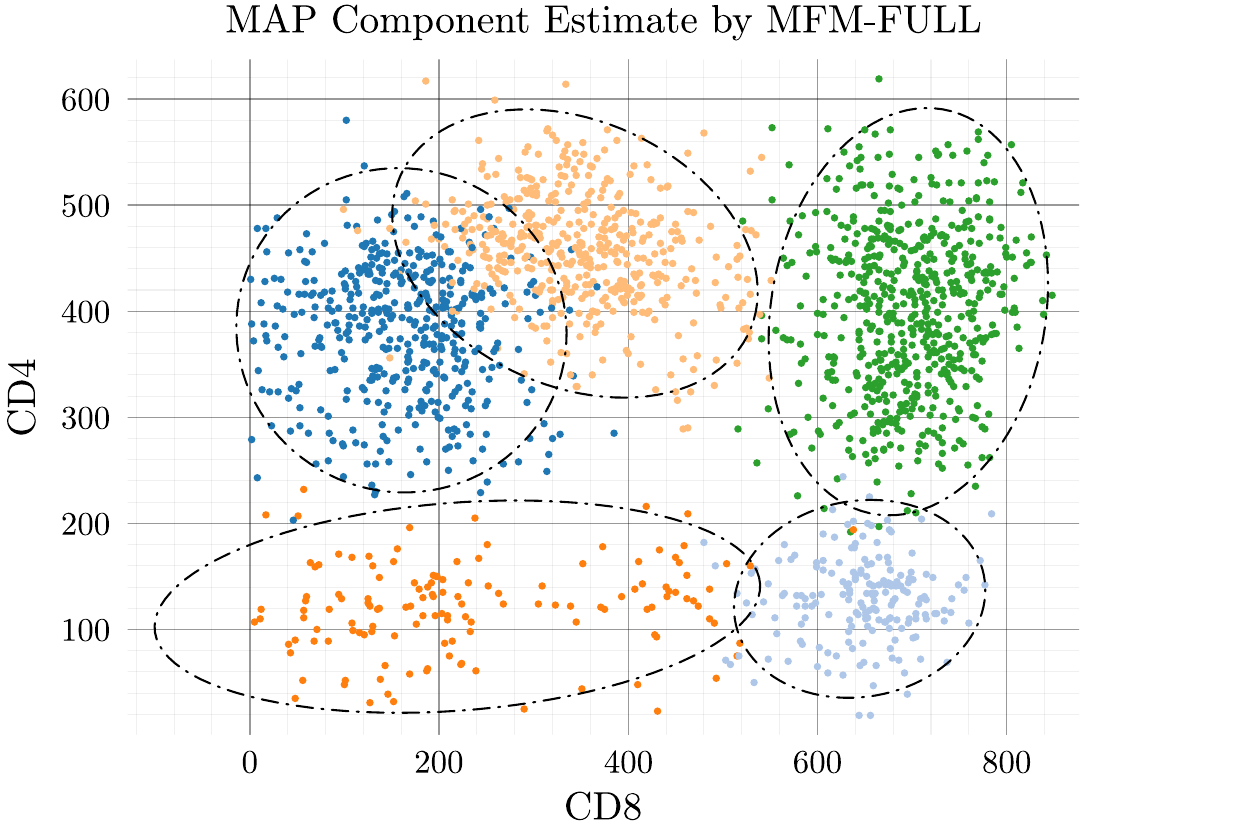} 
    }
    \\
    \subfloat{
        \includegraphics[width=0.46\linewidth]{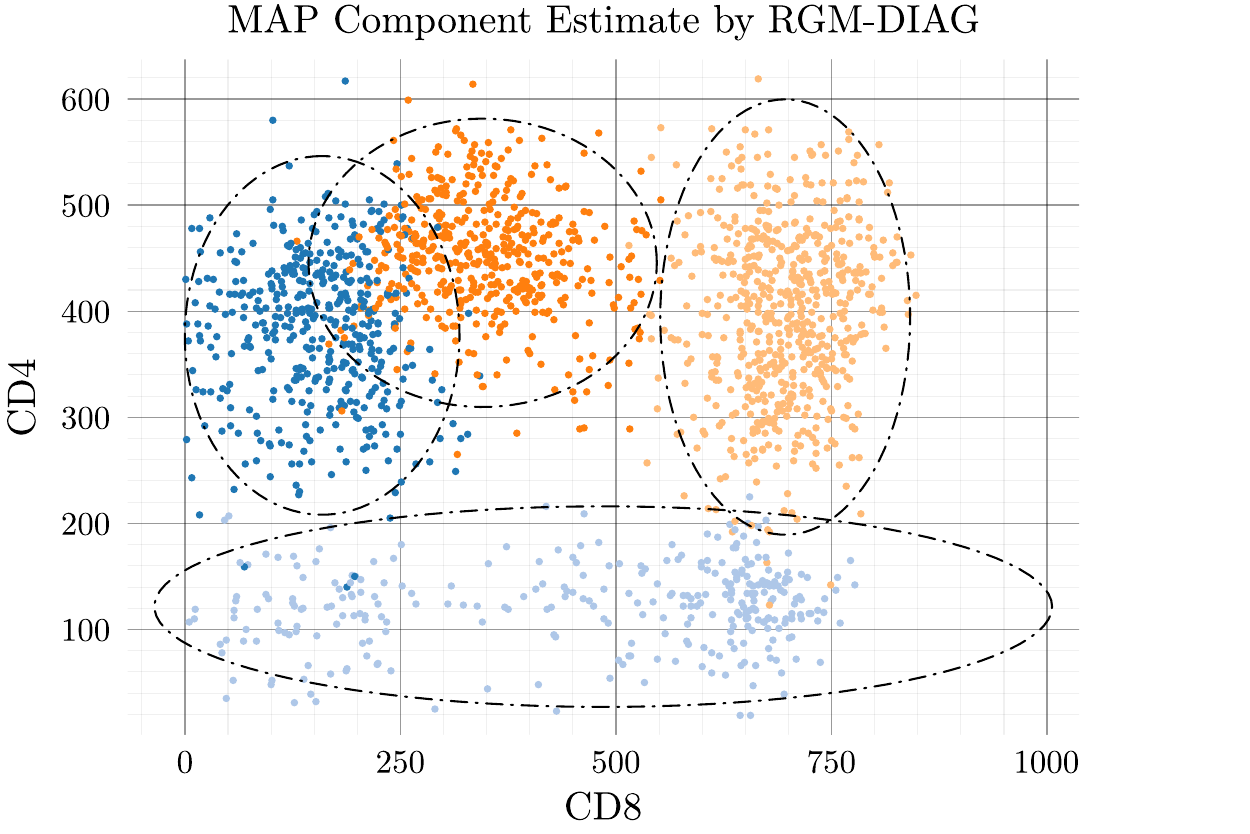} 
    } \quad
    \subfloat{
        \includegraphics[width=0.46\linewidth]{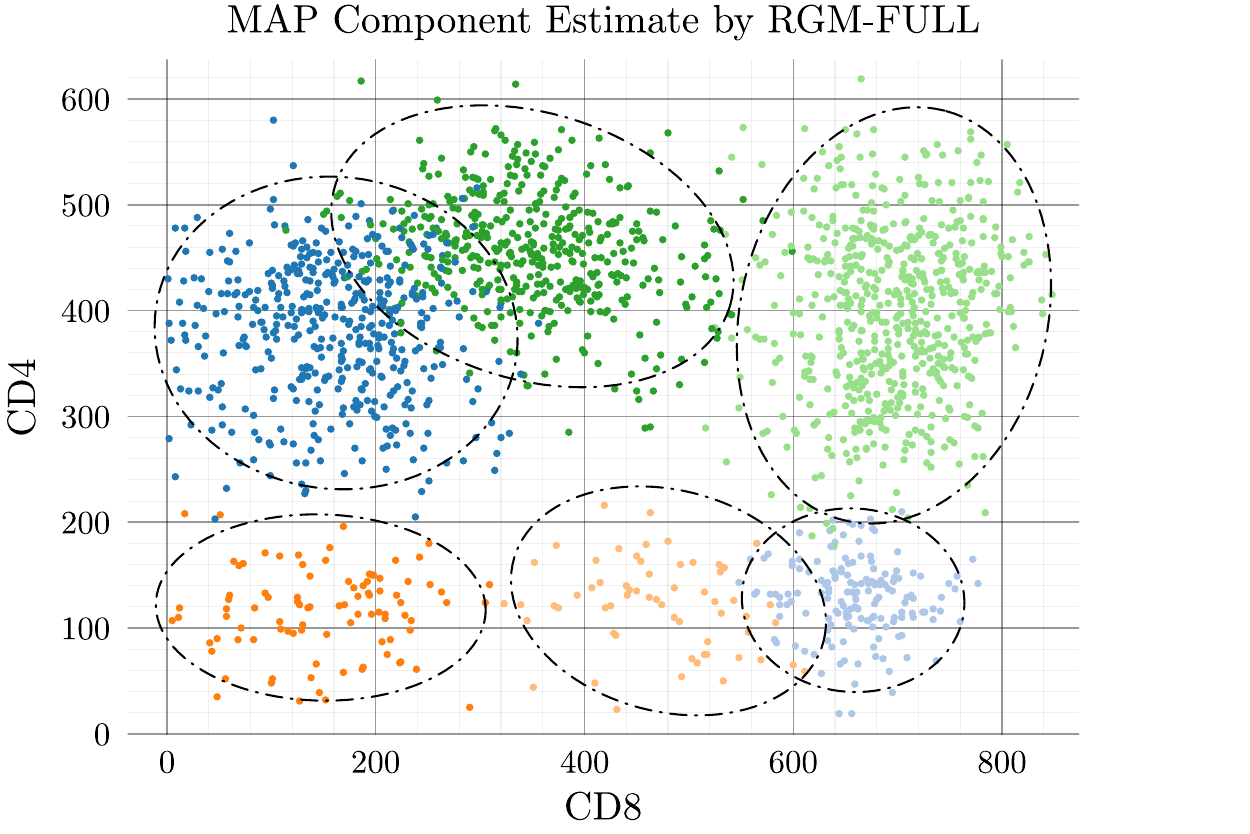} 
    } 
    \\
    \subfloat{
        \includegraphics[width=0.46\linewidth]{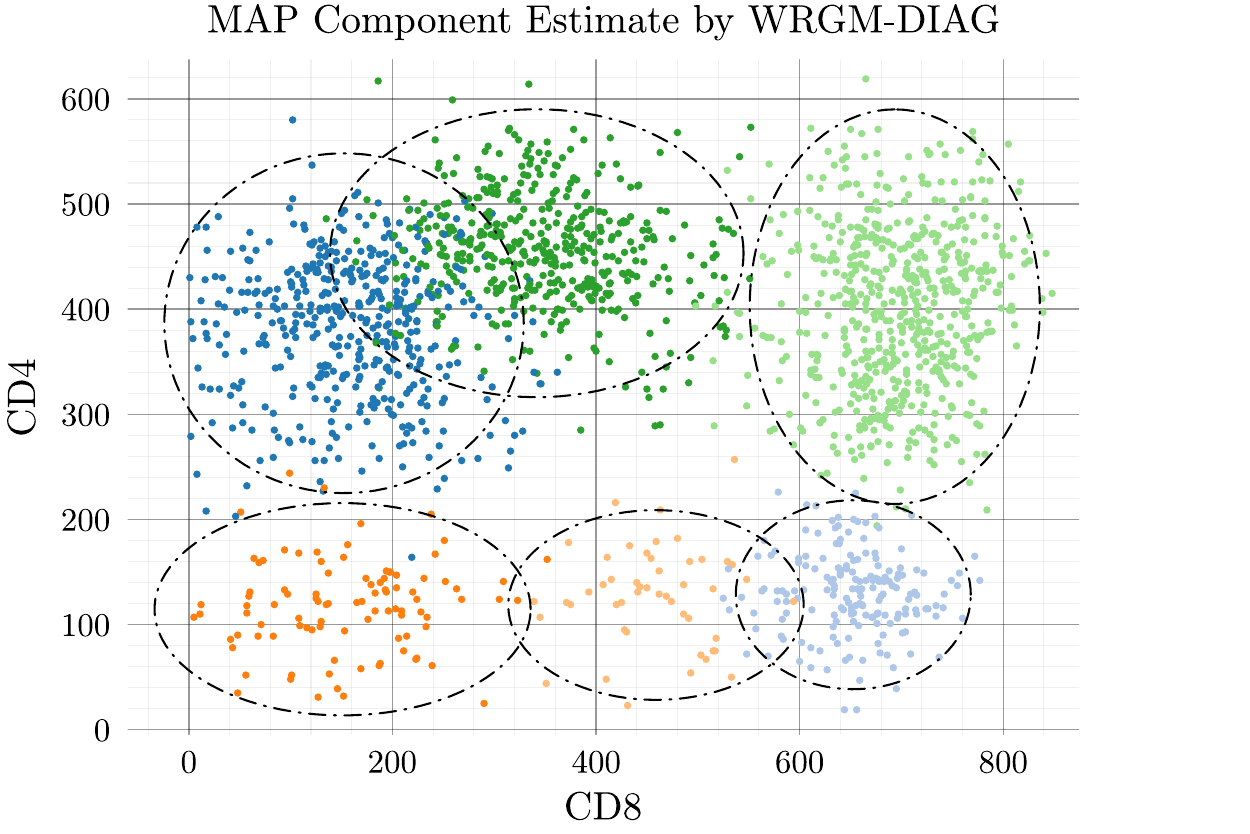} 
    } \quad
    \subfloat{
        \includegraphics[width=0.46\linewidth]{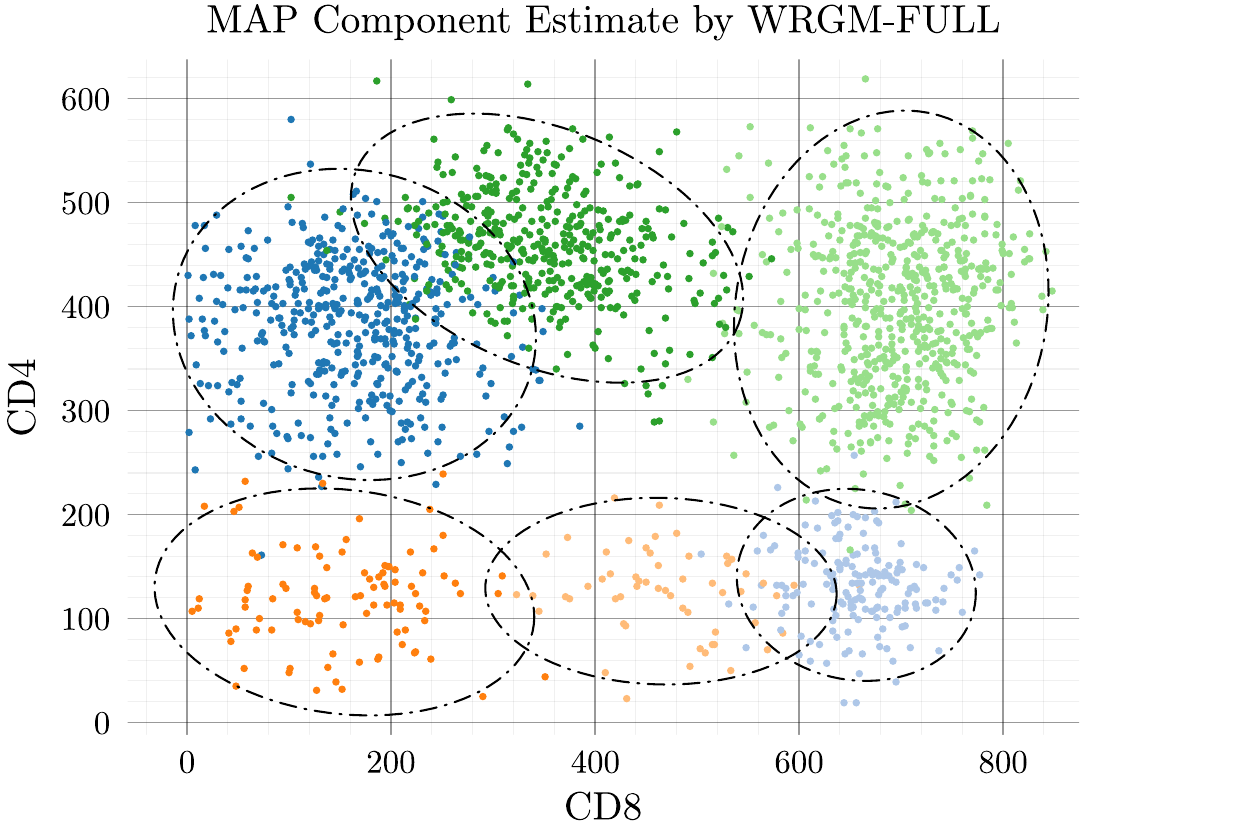} 
    }  
    \caption{MAP Component Assignments for the six models for the GvHD dataset}
    \label{fig:gvhd_map}
\end{figure}

\begin{figure}
    \centering
    \subfloat{
        \includesvg[width=0.46\linewidth]{figures/GvHD/GvHD_Mean_min_dist_kde} 
    } \quad
    \subfloat{
        \includesvg[width=0.46\linewidth]{figures/GvHD/GvHD_Wasserstein_min_dist_kde} 
    }  
    \caption{Minimum pairwise distances between the component mean parameters for each of the six models for the GvHD dataset.}
    \label{fig:gvhd_min_d}
\end{figure}

\section{Discussion}
\label{sec_discussion}
We introduced the WRGM model, an extension of the RGM framework, in which the mean and covariance parameters of each mixture component are no longer assumed to be a priori independent. Instead, they are jointly distributed according to a symmetric repulsive distribution designed to promote separation between the component distributions. We establish a posterior contraction rate for density estimation that matches the rate obtained by \cite{Xie2020}, up to a logarithmic factor. Simulation studies and data applications demonstrate that the proposed WRGM model outperforms both the RGM and MFM models in terms of density estimation accuracy.
\\\\
There are several promising avenues for future research. First, while this paper focuses on posterior contraction rates in the context of density estimation, it would be valuable to study the posterior contraction properties of the Gaussian mixing measures themselves \citep{Ho2016AOS, Ho2016EJS, Guha2021}. Another important direction is the selection of the number of mixture components. In this regard, exploring connections between non-local priors \citep{fuquene2019choosing} and the WRGM framework may offer a principled approach to model selection. Finally, although this work focuses on Gaussian mixtures—primarily due to the analytical convenience of closed-form Wasserstein distances—future work could extend to more general mixture models, potentially leveraging numerical methods to compute the Wasserstein distance.



\begin{appendices}

\section{Proofs of Theoretical Results}
This section provides the proofs of Proposition \ref{thm_normalizing_const} and Theorem \ref{thm_post_contraction}. The proof of Proposition \ref{thm_normalizing_const} follows the same strategy as Theorem 1 in \cite{Xie2020}, while the proof of Theorem \ref{thm_post_contraction} is adapted from the proof of Theorem 4 in \cite{Xie2020}. The technical lemmas supporting the proof of Theorem \ref{thm_post_contraction} are deferred to Appendix \ref{sec_technical_lemmas}.

\subsection{Proof of Proposition \ref{thm_normalizing_const}}
    \begin{proof}
Since $ h_K((m_1, \Sigma_1), \ldots, (m_K, \Sigma_K)) \le 1$, we immediately have that
$$ Z_K \le \int_{\mathbb{R}^p} \cdots \int_{\mathbb{R}^p}  \int_{{\cal S}} \cdots \int_{{\cal S}}  \bigg( \prod_{k=1}^{K} p_{m}(m_k) p_{\Sigma}(\Sigma_k) \bigg) d m_1 \cdots d m_K d \Sigma_1 \cdots d \Sigma_K = 1. $$    
It follows that $- \log h_K \ge 0$.
\\\\
Consider $m_1, \ldots, m_K \sim p_{m}$, and $\Sigma_1, \ldots, \Sigma_K \sim p_{\Sigma}$. Suppose $h_K$ is of the form \eqref{eqn_repulsive_fun1}. Denote $\mu_k = N(m_k, \Sigma_k), k=1,\ldots,K$. Then by Jensen's inequality, we have
\begin{eqnarray*}
-\log Z_K &=& -\log \mathbb{E} \Big[ \min_{1 \le k < k' \le K} g( W_{2}(\mu_k, \mu_k') ) \Big] 
\le \mathbb{E} \Big[ \max_{1 \le k < k' \le K} -\log  g( W_{2}(\mu_k, \mu_k') ) \Big].
\end{eqnarray*}
Since
\begin{eqnarray*}
 \Big[ \max_{1 \le k < k' \le K} -\log g( W_{2}(\mu_k, \mu_k') )  \Big]^2 = \max_{1 \le k < k' \le K} \big( \log g( W_{2}(\mu_k, \mu_k') ) \big)^2 ,   
\end{eqnarray*}
we obtain that
\begin{eqnarray*}
    -\log Z_K &\le& \mathbb{E} \bigg\{ \Big( \max_{1 \le k < k' \le K}\big[ \log g( W_{2}(\mu_k, \mu_k') )  \big]^2  \Big)^{\frac{1}{2}} \bigg\} \\
    & \le & \bigg\{ \mathbb{E} \bigg( \max_{1 \le k < k' \le K}\big[ \log g( W_{2}(\mu_k, \mu_k') )  \big]^2  \bigg) \bigg\}^{\frac{1}{2}} \\
    & \le & \bigg\{ \sum_{1 \le k < k' \le K} \mathbb{E} \bigg( \Big[ \log g(W_{2}(\mu_k, \mu_k')) \Big]^2 \bigg)  \bigg\}^{\frac{1}{2}} \\
    &=& \bigg\{  \frac{1}{2} K (K-1) \mathbb{E} \bigg( \Big[ \log g(W_{2}(\mu_1, \mu_2)) \Big]^2 \bigg) \bigg\}^{\frac{1}{2}} \\
    & \le & c_1 K 
\end{eqnarray*}
where the second inequality follows from Jensen's inequality, and $c_1$ is a constant which can be taken as
$$ c_1 = \frac{1}{2} \mathbb{E} \big( \big[ \log g(W_{2}(\mu_k, \mu_k')) \big]^2 \big) .$$
The proof when $h_K$ takes the form \eqref{eqn_repulsive_fun2} is similar and is thus ommited.

\end{proof}

\subsection{Proof of Theorem \ref{thm_post_contraction}}
\begin{proof}
    To obtain posterior contraction rate, we consider the sequence of submodels $({\cal F}_{K_n})_{n=1}^{\infty}$ of ${\cal M}(\mathbb{R}^p)$, where ${\cal F}_{K_n}$ is defined as:
$$ {\cal F}_{K_n} = \bigg\{f_F: F = \sum_{k=1}^{K} w_k \delta_{(m_k, \Sigma_k)}, K \le K_n, m_k \in \mathbb{R}^p, \Sigma_k \in {\cal S} \bigg\} ,$$
and $(K_n)_{n}$ is a sequence of increasing integers.
The submodel ${\cal F}_{K_n}$ is partitioned as follows:
$${\cal G}_K(\mathbf{a}_K) = {\cal F}_K \bigg( \prod_{k=1}^{K} (a_k, a_k + 1]) \bigg), \quad \mathbf{a}_K = (a_1, \ldots, a_K) \in \mathbb{N}^K, \quad K=1,\ldots, K_n,$$
where
$${\cal F}_K \bigg( \prod_{k=1}^{K} (a_k, b_k] \bigg) = \bigg\{f_F: F = \sum_{k=1}^{K} w_k \delta_{(m_k, \Sigma_k)}, ||m_k||_{\infty} \in (a_k, b_k] \bigg\} .$$

Applying Theorem 3 from \cite{Kruijer2010}, computing the posterior contraction rate reduces to identifying two sequences 
$(\overline{\epsilon}_n)_{n=1}^{\infty}$, $(\underline{\epsilon}_n)_{n=1}^{\infty}$ such that
\begin{eqnarray}
    \label{eqn_contraction_cond1}
    \Pi({\cal F}_{K_n}^{c}) \lesssim \exp(-c_1 n \underline{\epsilon}_n^2),
\end{eqnarray}
for some $c_1 > 0$,
\begin{eqnarray}
\label{eqn_contraction_cond2}
 \exp(- n \overline{\epsilon}_n^2)    \sum_{K=1}^{K_n} \sum_{a_1 = 0}^{\infty} \cdots \sum_{a_K=0}^{\infty} \sqrt{ {\cal N}(\overline{\epsilon}_n, {\cal G}_K( \mathbf{a}_K ), ||\cdot||_1 ) } \sqrt{ \Pi({\cal G}(\mathbf{a}_K) ) } \rightarrow 0,
\end{eqnarray}
\begin{eqnarray}
    \label{eqn_contraction_cond3}
    \Pi\Bigg( f: \int f_0 \log \frac{f_0}{f} \le \underline{\epsilon}_n^2, \int f_0 \bigg( \log \frac{f_0}{f} \bigg)^2 \le \underline{\epsilon}_n^2 \Bigg) \ge \exp(- n \underline{\epsilon}_n^2) .
\end{eqnarray}

We set $\underline{\epsilon}_n = (\log n)^{t_0} / \sqrt{n}, \overline{\epsilon}_n = (\log n)^t / \sqrt{n}$ where $t$ and $t_0$ satisfy $t > t_0 + \frac{1}{2}$, $t_0 > \frac{p^2}{2} + p + \frac{\alpha}{4}$, and $K_n =  \floor*{c_0 (\log n)^{2t-1} } $, for some sufficiently small constant $c_0 > 0$. Lemma \ref{lemma_entropy_remaining_mass_cond} verifies Conditions \eqref{eqn_contraction_cond1} and \eqref{eqn_contraction_cond2}, while Lemma \ref{lemma_prior_mass_cond} verifies Condition \eqref{eqn_contraction_cond3}.  
\end{proof}

The following lemma bounds the covering number of $ {\cal F}_K \bigg( \prod_{k=1}^{K} (a_k, b_k] \bigg)$ with respect to $||\cdot||_1$.
\begin{lemma}
\label{lemma_entropy_bound}
    \begin{eqnarray*}
{\cal N}\bigg( \epsilon, {\cal F}_K \bigg( \prod_{k=1}^{K} (a_k, b_k] \bigg), ||\cdot||_1 \bigg) & \lesssim &  \frac{1}{\epsilon^{Kq}} \bigg( \prod_{k=1}^{K} b_k \bigg)^{p}
\end{eqnarray*}
where $q = 1+4p+p(p-1)/2$.
\end{lemma}
\begin{proof}
It is easy to see that we have the following bounds 
$$ {\cal N}(\epsilon, \triangle^K, ||\cdot||_1) \lesssim \bigg( \frac{1}{\epsilon} \bigg)^K .$$
$$ {\cal N}(\epsilon, [\underline{\sigma}^2, \overline{\sigma}^2], |\cdot|) \lesssim \frac{1}{\epsilon} ,$$
$$ {\cal N}(\epsilon, \{m_k: ||m_k||_{\infty} \in (a_k, b_k]\}, ||\cdot||_{\infty}) \lesssim \bigg( \frac{b_k}{\epsilon} \bigg)^p .$$
By Lemma \ref{lemma_entropy_orthogonal_group}, we can also bound the covering number of $O(p)$ with respect to the operator norm:
$${\cal N}(\epsilon, O(p), ||\cdot||_{op}) \lesssim \bigg( \frac{1}{\epsilon} \bigg)^{p(p-1)/2} .$$

Consider $f_{F_1}, f_{F_2} \in {\cal F}_K(\prod_{k=1}^{K} (a_k, b_k]) $ with $F_1 = \sum_{k=1}^{K} w_{1,k} \delta_{(m_{1,k}, \Sigma_{1,k})}$, $F_2 = \sum_{k=1}^{K} w_{2,k} \delta_{(m_{2,k}, \Sigma_{2,k})}$, where $||m_{1,k}||_{\infty}, || m_{2,k} ||_{\infty} \in (a_k, b_k]$, and the covariance matrices have the following decompositions
$$ \Sigma_{1,k} = U_{1,k} \mbox{diag}(\lambda_1(\Sigma_{1,k}), \ldots, \lambda_p(\Sigma_{1,k})) U_{1,k}^{T},$$
$$ \Sigma_{2,k} = U_{2,k} \mbox{diag}(\lambda_1(\Sigma_{2,k}), \ldots, \lambda_p(\Sigma_{2,k})) U_{2,k}^{T}.$$

For each $k=1,\ldots,K$, by Lemma \ref{lemma_hellinger}, 
the conditions $||m_{1,k} - m_{2,k}|| < p \epsilon$, $|\lambda_j(\Sigma_{1,k}) - \lambda_j(\Sigma_{2,k})| < \epsilon, j=1,\ldots,p$, and $||U_{1,k} - U_{2,k}||_{op} < \epsilon$  imply the following bound on the Hellinger distance between two normal distributions:
$$ H(N(m_{1,k}, \Sigma_{1,k}), N(m_{2,k}, \Sigma_{2,k})) \lesssim \epsilon .$$
Given $\sum_{k=1}^{K} |w_{1,k} - w_{2,k}| < \epsilon$, by triangle inequality and upper bounding the $L_1$ norm of two Gaussian densities by Hellinger distance, we obtain  
\begin{eqnarray*}
    ||f_{F_1} - f_{F_2}||_1 &\le& \sum_{k=1}^{K} w_{1,k} ||\phi_{\Sigma_{1,k}}(y - m_{1,k}) - \phi_{\Sigma_{2,k}}(y - m_{2,k})||_1 + \sum_{k=1}^{K} |w_{1,k} - w_{2,k}| \\
    &\le& \sum_{k=1}^{K} 2\sqrt{2} w_{1,k} H(\phi_{\Sigma_{1,k}}(y - m_{1,k}), \phi_{\Sigma_{2,k}}(y - m_{2,k})) + \epsilon \\
    & \lesssim & \epsilon.
\end{eqnarray*}

It thus follows that 
\begin{eqnarray*}
{\cal N}\bigg( \epsilon, {\cal F}_K \bigg( \prod_{k=1}^{K} (a_k, b_k] \bigg), ||\cdot||_1 \bigg) & \lesssim & {\cal N}(\epsilon, \triangle^K, ||\cdot||_1) \times {\cal N}(\epsilon, [\underline{\sigma}, \overline{\sigma}], |\cdot|)^{pK} \\
&& \times {\cal N}(\epsilon, \{m_k: ||m_k||_{\infty} \in (a_k, b_k]\}, ||\cdot||_{\infty})^{K} \\
&& \times {\cal N}(\epsilon, O(p), ||\cdot||_{op})^K  \\
& \lesssim & \bigg( \frac{1}{\epsilon} \bigg)^K \bigg( \frac{1}{\epsilon} \bigg)^{pK}  \bigg(\frac{1}{\epsilon} \bigg)^{p(p-1)K/2} \prod_{k=1}^{K} \bigg( \frac{b_k}{\epsilon} \bigg)^{p} \\
&= & \frac{1}{\epsilon^{Kq}} \bigg( \prod_{k=1}^{K} b_k \bigg)^{p}
\end{eqnarray*}
where $q = 1+2p+p(p-1)/2$.

\end{proof}

\begin{lemma}
\label{lemma_entropy_remaining_mass_cond}
   Assume conditions A0-A3 and B1-B5 hold. Let $\underline{\epsilon}_n = (\log n)^{t_0} / \sqrt{n}, \overline{\epsilon}_n = (\log n)^t / \sqrt{n}$ where $t$ and $t_0$ satisfy $t > t_0 + \frac{1}{2} > \frac{1}{2}$, and $K_n =  \floor*{c_0 (\log n)^{2t-1} } $, for some sufficiently small constant $c_0 > 0$. Then \eqref{eqn_contraction_cond1} and \eqref{eqn_contraction_cond2} hold.
\end{lemma}
\begin{proof}
We first derive an upper bound on the quantity below:
$$ \sum_{K=1}^{K_n} \sum_{a_1 = 0}^{\infty} \cdots \sum_{a_K=0}^{\infty} \sqrt{ {\cal N}(\epsilon, {\cal G}_K( \mathbf{a}_K ), ||\cdot||_1 ) } \sqrt{ \Pi({\cal G}(\mathbf{a}_K) ) } $$
We have
\begin{eqnarray*}
    \Pi({\cal G}_K(\mathbf{a}_K)) &\le& \Pi(m_1, \ldots, m_K: ||m_k|| \ge \sqrt{p} a_k, k=1,\ldots, K|K) p_K(K)  \\
    & \le & \frac{p_K(K)}{Z_K} \int \cdots \int \prod_{k=1}^{K} I(||m_k||^2 \ge p a_k^2) p_m(m_1) dm_1 \cdots p_m(m_K) dm_K \\
    && \qquad \qquad \qquad p_{\Sigma}(\Sigma_1) d \Sigma_1 \cdots p_{\Sigma}(\Sigma_K) d \Sigma_K \\
    & \le & e^{c_1 K} \prod_{k=1}^{K} \int_{||m_k||^2 \ge p a_k^2} p(m_k) dm_k \\
    & \le & e^{c_1 K} B^{K/2} \prod_{k=1}^{K} \exp \left( - p b_2 a_k^2 \right).
\end{eqnarray*}
where the second last inequality follows from Proposition \ref{thm_normalizing_const} and the last inequality follows from Condition B2. 
We have by Lemma \ref{lemma_entropy_bound} that
$${\cal N}(\epsilon, {\cal G}_K(\mathbf{a}_K), ||\cdot||_1) \le c_2 \bigg( \frac{1}{\epsilon^q} \bigg)^K \prod_{k=1}^{K} (a_k + 1)^p ,$$
for some $c_2 > 0$. It thus follows that
\begin{eqnarray*}
 \sum_{K=1}^{K_n} \sum_{a_1 = 0}^{\infty} \cdots \sum_{a_K=0}^{\infty} \sqrt{ {\cal N}(\epsilon, {\cal G}_K( \mathbf{a}_K ), ||\cdot||_1 ) } \sqrt{ \Pi({\cal G}(\mathbf{a}_K) ) }  &\le &\sum_{K=1}^{K_n} \bigg( \frac{ S \sqrt{B_2 c_2 e^{c_1} }}{ \epsilon^{q/2}} \bigg)^K \\
 & \le & K_n \bigg( \frac{M}{\epsilon^{q/2}} \bigg)^{K_n}
\end{eqnarray*}
for some $M>0$. Thus, for $c_0$ small enough we have

\begin{eqnarray}
 \exp(- n \overline{\epsilon}_n^2)    \sum_{K=1}^{K_n} \sum_{a_1 = 0}^{\infty} \cdots \sum_{a_K=0}^{\infty} \sqrt{ {\cal N}(\overline{\epsilon}_n, {\cal G}_K( \mathbf{a}_K ), ||\cdot||_1 ) } \sqrt{ \Pi({\cal G}_K(\mathbf{a}_K) ) } \le \exp\bigg(-\frac{1}{2} (\log n)^{2t} \bigg) \rightarrow 0 .
\end{eqnarray}
For the remaining mass condition, we observe that by Condition B5,
$$ \Pi({\cal F}_{K_n}^{c}) = \Pi(K > K_n) \le\exp(-B_4 K_n \log K_n) \le \exp(-c_1 n \underline{\epsilon}_n^2) ,$$
for some $c_1 > 0$.

\end{proof}

\begin{lemma}
\label{lemma_prior_mass_cond}
Suppose Assumptions A0 - A3 and B1 - B5 hold. Let $\underline{\epsilon}_n = (\log n)^{t_0} / \sqrt{n}$ with $t_0 > \frac{p^2}{2} + p + \frac{\alpha}{4}$. Then Condition \ref{eqn_contraction_cond3} is satisfied. 
\end{lemma}
\begin{proof}
Motivated by Lemma \ref{lemma_mixture_approx_2}, we are interested in finding the prior probability of the following event:
\[ \tilde{B}(F^{*},\epsilon) := 
\left\{ f_F : F = \sum_{k=1}^{N} w_k \delta_{(m_k, \Sigma_k)}, \quad (m_k, \Sigma_k) \in E_k, \quad \sum_{k=1}^{N} |w_k - w^*_k| < \epsilon \right\},
\]
where
\[
E_k = \left\{ (m, \Sigma) \in \mathbb{R}^p \times \mathcal{S} : \|m - m^*_k\|_{\infty} < \frac{\epsilon}{2}, \quad || \Sigma - \Sigma^*_k||_{op} < \frac{\epsilon}{2} \right\},
\]
and the measure $
F^* = \sum_{k=1}^{N} w^*_k \delta_{(m^*_k, \Sigma^*_k)}
$
satisfies
\[||m_k^{*}||_{\infty} \lesssim \bigg( \log \frac{1}{\epsilon} \bigg)^{\frac{1}{2}},
\|m^*_k - m^*_{k'}\|_{\infty} \geq 2\epsilon, \quad ||\Sigma^*_k - \Sigma^*_{k'}||_{op} \geq 2\epsilon \text{ whenever } k \neq k', \quad j = 1, \dots, p,
\]
and
\[
N \lesssim \left( \log \frac{1}{\epsilon} \right)^{p^2+2p}.
\]
The prior probability of $\tilde{B}(F^{*},\epsilon)$ can be factorized as:
\[
\Pi(\tilde{B}(F^*, \epsilon)) = p_K(N) \Pi \left( \bigcap_{k=1}^{N} \{(m_k, \Sigma_k) \in E_k\} \mid K = N \right) \Pi \left( \|\mathbf{w} - \mathbf{w}^*\|_1 < \epsilon \mid K = N \right).
\]
Applying condition A1, B3, Lemma \ref{lemma_wishart}, and following similar lines of the proof of Theorem 4 in \cite{Xie2020}, we can show that
$$ - N \bigg( \frac{1}{\epsilon} \bigg)^{\frac{\alpha}{2}} \lesssim \log \Pi((m_k, \Sigma_k) \in E_k, k=1,\ldots,N).$$
Applying Condition B5 and the assumption $N \lesssim \bigg( \log \frac{1}{\epsilon} \bigg)^{p^2+2p}$, it follows that
$$- \bigg( \log \frac{1}{\epsilon} \bigg)^{p^2 + 2p + \frac{\alpha}{2}} \lesssim \log p_K(N) +  \log \Pi((m_k, \Sigma_k) \in E_k, k=1,\ldots,N).$$
By Lemma A.2 in Ghosal and Van Der Vaart (2001), we have
\[
- \bigg( \log \frac{1}{\epsilon} \bigg)^{p^2 + 2p+1} \lesssim - N \bigg( \log \frac{1}{\epsilon} \bigg)
\lesssim \log \Pi \left( w_1, \dots, w_N : \sum_{k=1}^{N} |w_k - w^*_k| < \epsilon \right).
\]
It follows that for $\alpha \ge 2$,
\[
\exp\bigg( -c_3 \bigg( \log \frac{1}{\epsilon} \bigg)^{p^2+2p+\frac{\alpha}{2}} \bigg) \lesssim \Pi(\tilde{B}(F^{*}, \epsilon)) \lesssim \Pi( B\bigg(f_0, \eta \epsilon^{\frac{1}{2}} \bigg(\log \frac{1}{\epsilon} \bigg)^{\frac{p+4}{4}} \bigg) ,
\]
where the second inequality follows from Lemma \ref{lemma_mixture_approx_2}.
It thus follows that
$$ \exp\bigg( -c_3 \bigg( \log \frac{1}{\epsilon} \bigg)^{p^2+2p+\frac{\alpha}{2}} \bigg) \lesssim  \Pi(B(f_0, \epsilon)) .$$
Setting $\underline{\epsilon}_n = \frac{(\log n)^{t_0}}{\sqrt{n}} $ with $t_0 > \frac{p^2}{2} + p + \frac{\alpha}{4}$, \eqref{eqn_contraction_cond3} is satisfied.
\end{proof}

\section{Technical Lemmas}
\label{sec_technical_lemmas}
\begin{lemma}
\label{lemma_hellinger}
Let $N(m_1, \Sigma_1)$ and $N(m_2, \Sigma_2)$ be two arbitrary $p$-dimensional multivariate Gaussian distributions, and $U_1, U_2$ be two orthogonal matrices such that  $U_1^{T} \Sigma_1 U_1 = \text{diag}((\lambda_j(\Sigma_1)_{j=1}^{p})$ and $U_2^{T} \Sigma_2 U_2 = \text{diag}((\lambda_j(\Sigma_2)_{j=1}^{p})$. Suppose the eignenvalues satisfy $0 < \underline{\sigma}^2 < \lambda_j(\Sigma_1) < \overline{\sigma}^2$ and $0 < \underline{\sigma}^2 < \lambda_j(\Sigma_2) < \overline{\sigma}^2$ for all $j=1,\ldots,p$. 
Furthermore, Suppose $||m_1 - m_2|| < \sqrt{p} \epsilon$, $| \lambda_j(\Sigma_1) - \lambda_j(\Sigma_2)| < \epsilon$, $j=1,\ldots,p$, and $||U_1 - U_2||_{op} < \epsilon$, where $||\cdot||_{op}$ is the operator norm. 
\\\\
Then the Hellinger distance between $N(m_1, \Sigma_1)$ and $N(m_2, \Sigma_2)$ can be upper bounded as
\[
H^2(N(m_1, \Sigma_1), N(m_2, \Sigma_2)) \le C \epsilon^2 .
\]
for all sufficiently small $\epsilon > 0$, and where the constant $C > 0$ does not depend on $N(m_1, \Sigma_1)$ or $N(m_2, \Sigma_2)$.

\end{lemma}
\begin{proof}
    The (squared) Hellinger distance between two multivariate Gaussian distributions \( N(m_1, \Sigma_1) \) and \( N(m_2, \Sigma_2) \) is given by:

\[
H^2(N(m_1, \Sigma_1), N(m_2, \Sigma_2)) = 1 - \frac{\det(\Sigma_1)^{1/4} \det(\Sigma_2)^{1/4}}{\det\left(\frac{\Sigma_1 + \Sigma_2}{2} \right)^{1/2}} e^{-\frac{1}{8} (m_1 - m_2)^T \left(\frac{\Sigma_1 + \Sigma_2}{2}\right)^{-1} (m_1 - m_2)}
\]
Using the given bounds on the eigenvalues:
\[
\underline{\sigma}^2 < \lambda_j(\Sigma_1), \lambda_j(\Sigma_2) < \overline{\sigma}^2, \quad j=1, \dots, p
\]
and the condition
\[
| \lambda_j(\Sigma_2) - \lambda_j(\Sigma_1) | < \epsilon, \quad j=1,\ldots,p,
\]
we have
\begin{eqnarray*}
    \det(\Sigma_2) \le \det(\Sigma_1) (1+c_1 \epsilon)^p \le \det(\Sigma_1) (1+c_2 p \epsilon) 
\end{eqnarray*}
for some constant $c_1, c_2 > 0$. 
\\\\
We now move on to the determinant of $\frac{\Sigma_1 + \Sigma_2}{2}$. Since $\Sigma_1 = U_1 D_1 U_1^{T}$ and $\Sigma_2 = U_2 D_2 U_2^{T}$ for some diagonal matrices $D_1, D_2$,
We have
\[ \frac{\Sigma_1 + \Sigma_2}{2} = U_1 \bigg( \frac{D_1 + U_1^{T} U_2 D_2 U_2^{T} U_1 }{2}  \bigg) U_1^{T} . \]
Given the condition $||U_1 - U_2||_{op} < \epsilon$, we can write $U_2 = U_1 + \triangle U$ for some $\triangle U$ with $||\triangle U||_{op} < \epsilon$, and consequently 
$$ U_1^{T} U_2 = I + U_1^{T} (\triangle U) .$$
It then follows that 
$$ U_1^{T} U_2 D_2 U_2^{T} U_1 = D_2 + \Xi$$
for some symmetric ``perturbation matrix'' $\Xi$ with $||\Xi||_{op} \le C_1 \epsilon$, for some constant $C_1 > 0$.
\\
We now consider the eigenvalues of $\frac{\Sigma_1 + \Sigma_2}{2}$:
$$ \lambda_j\bigg( \frac{\Sigma_1 + \Sigma_2}{2} \bigg) = \lambda_j \bigg( 
U_1\frac{D_1 + D_2 + \Xi}{2} U_1^{T} \bigg).$$
\\
We first consider the eigenvalues of $U_1 \frac{D_1 + D_2}{2} U_1^T$:
$$ \lambda_j \bigg(U_1 \frac{D_1 + D_2}{2} U_1^T \bigg) = \lambda_j \bigg(  \frac{D_1 + D_2}{2} \bigg) = \frac{1}{2}(\lambda_j(D_1) + \lambda_j(D_2)) = \frac{1}{2}(\lambda_j(\Sigma_1) + \lambda_j(\Sigma_2)). $$
\\
By Weyl's inequality, we can bound the eigenvalue differences by: 
\begin{eqnarray}
\label{eqn_eigen_diff}
\Bigg|\lambda_j \bigg( \frac{\Sigma_1 + \Sigma_2}{2}\bigg) - \frac{1}{2}(\lambda_j(\Sigma_1) + \lambda_j(\Sigma_2)) \Bigg| \le ||U_1 \Xi U_1^{T}||_{op} = ||\Xi||_{op}  \le C_1 \epsilon .    
\end{eqnarray}
\\
Define 
$$ \bar{\lambda}_j = \frac{1}{2} (\lambda_j(\Sigma_1) + \lambda_j(\Sigma_2)),$$
we can express $\lambda_j(\Sigma_1)$ and $\lambda_j(\Sigma)$ as perturbations around $\bar{\lambda}_j$:
$$ \lambda_j(\Sigma_1) = \bar{\lambda}_j + \delta_j, \quad \lambda_j(\Sigma_1) = \bar{\lambda}_j - \delta_j $$
for some $\delta_j$ satisfying $|\delta_j| < \epsilon$ by assumption. Applying Taylor's theorem we obtain:
$$ \log \lambda_j(\Sigma_1) = \log \bar{\lambda}_j + \log\bigg( 1 + \frac{\delta_j}{\bar{\lambda}_j} \bigg) \le \log \bar{\lambda}_j + \frac{\delta_j}{\bar{\lambda}_j} + C_2 \epsilon^2 ,$$
$$ \log \lambda_j(\Sigma_2) = \log \bar{\lambda}_j + \log\bigg( 1 - \frac{\delta_j}{\bar{\lambda}_j} \bigg) \le \log \bar{\lambda}_j - \frac{\delta_j}{\bar{\lambda}_j} + C_2 \epsilon^2$$
for some constant $C_2 > 0$. It follows that

$$\log\bigg( \frac{1}{2} (\lambda_j(\Sigma_1) + \lambda_j(\Sigma_2))\bigg)  =  \log \bar{\lambda}_j \le \frac{1}{2} \bigg( \log \lambda_j(\Sigma_1) + \log \lambda_j(\Sigma_2) \bigg) + C_2 \epsilon^2 . $$
Consequently,
\begin{eqnarray*}
    \log \det\bigg(\frac{\Sigma_1 + \Sigma_2}{2} \bigg) &\le& \frac{1}{2} \sum_{j=1}^{p} (\log \lambda_j(\Sigma_1) + \log \lambda_j(\Sigma_2) ) + C_2 p \epsilon^2 \\
    &=& \frac{1}{2} (\log \det(\Sigma_1) + \log \det(\Sigma_2)) + C_2 p \epsilon^2 .
\end{eqnarray*}
Therefore,
$$ \det \bigg( \frac{ \Sigma_1 + \Sigma_2 }{2} \bigg)^{1/2} \le (\det(\Sigma_1) \det(\Sigma_2))^{1/4} \exp\bigg(\frac{1}{2} C_2 p \epsilon^2 \bigg) .$$
Thus, we obtain:
\begin{eqnarray}
    \label{eqn_determinant_bound}
    \frac{\det(\Sigma_1)^{1/4} \det(\Sigma_2)^{1/4}}{\det\left(\frac{\Sigma_1 + \Sigma_2}{2} \right)^{1/2}} \ge \exp\bigg(-\frac{1}{2}C_2 p\epsilon^2\bigg).
\end{eqnarray}
We next bound the exponential term, for which we use the given condition:
\[
||\mu_1 - \mu_2|| < \sqrt{p} \epsilon.
\]
From \eqref{eqn_eigen_diff}, \( (\Sigma_1 + \Sigma_2)/2 \) has eigenvalues bounded in \( [\underline{\sigma}^2 - C_1 \epsilon, \overline{\sigma}^2 + C_1 \epsilon ]\) for some constant $C_1>0$, the operator norm of its inverse satisfies:

\[
\left\|\left(\frac{\Sigma_1 + \Sigma_2}{2}\right)^{-1} \right\|_{op} \leq \frac{1}{\underline{\sigma}^2 - C_1 \epsilon}.
\]

Thus, the quadratic form in the exponent satisfies:

\[
\frac{1}{8} (\mu_1 - \mu_2)^T \left(\frac{\Sigma_1 + \Sigma_2}{2}\right)^{-1} (\mu_1 - \mu_2) \leq \frac{\sqrt{p}^2 \epsilon^2}{8 (\underline{\sigma}^2- C_1 \epsilon)} = \frac{p \epsilon^2}{8 (\underline{\sigma}^2 - C_1 \epsilon)}.
\]

So the exponential term satisfies:
\begin{eqnarray}
    \label{eqn_exp_bound}
    e^{-\frac{1}{8} (\mu_1 - \mu_2)^T \left(\frac{\Sigma_1 + \Sigma_2}{2}\right)^{-1} (\mu_1 - \mu_2)} \ge 1 - C_3 \epsilon^2.
\end{eqnarray}
for some $C_3 > 0$.
\\
Combining \eqref{eqn_determinant_bound} and \eqref{eqn_exp_bound}, we obtain:

\[
H^2(N(\mu_1, \Sigma_1), N(\mu_2, \Sigma_2)) \le C \epsilon^2 .
\]
for some $C > 0$.

\end{proof}

We have the following result concerning the metric entropy of $O(p)$ with respect to the operator norm.
\begin{lemma}
\label{lemma_entropy_orthogonal_group}
The metric entropy of $O(p)$ can be upper bounded as:
$$    {\cal N}(\epsilon, O(p), ||\cdot||_{op}) \lesssim \bigg(\frac{1}{\epsilon}\bigg)^{p(p-1)/2} .$$
\end{lemma}
\begin{proof}
    Since $O(p)$ is a real compact Lie group of dimension $p(p-1)/2$, this result follows from Theorem 7 of \cite{Szarek1998}.
\end{proof}

The following result is an extension of Lemma D.1 of \cite{Xie2020}.
\begin{lemma}
\label{lemma_mixture_approx_1}
    Let $F$ be a probability distribution compactly supported on a subset of 
\[
\{(m, \Sigma) \in \mathbb{R}^p \times \mathcal{S} : ||m||_{\infty} \leq a\}
\]
with $a \lesssim \bigg( \log \frac{1}{\epsilon}\bigg)^{\frac{1}{2}}$. Then for sufficiently small $\epsilon > 0$, there exists a discrete probability distribution $F^*$ on a subset of 
\[
\{(m, \Sigma) \in \mathbb{R}^p \times \mathcal{S} : ||m||_{\infty} \leq a\}
\]
with at most 
\[
N \lesssim \left(\log \frac{1}{\epsilon} \right)^{p^2 + 2p}
\]
support points, such that 
\[
||f_F - f_{F^*}||_{\infty} \lesssim \epsilon, \quad \text{and} \quad ||f_F - f_{F^*}||_1 \lesssim \epsilon \left(\log \frac{1}{\epsilon} \right)^{p/2}.
\]
\end{lemma}
\begin{proof}
The proof is an adaptation of the proof of Lemma D.1 of \cite{Xie2020}.  We define the quadratic form:
\[
Q_{\Sigma}(y) = y^T \Sigma^{-1} y.
\]
We have
\begin{eqnarray*}
Q_{\Sigma}(y - m) &=& (y - m)^T \Sigma^{-1} (y - m) \\
&=&  \sum_{i,j} (y_i - m_i) (\Sigma^{-1})_{ij} (y_j - m_j)
\end{eqnarray*}

Raising this expression to the power \( j \) and applying the multinomial theorem, we have

\[
Q_{\Sigma}^{j}(y - m) = \sum_{\substack{r_{i_1,i_2} \geq 0 \\ \sum_{i_1,i_2} r_{i_1,i_2} = j}} \binom{j}{\{r_{i_1,i_2}\}} \prod_{i_1,i_2} \left[ (\Sigma^{-1})_{i_1,i_2} (y_{i_1} - m_{i_1}) (y_{i_2} - m_{i_2}) \right]^{r_{i_1,i_2}}.
\]
We aim to find a discrete distribution $F^{*}$ such that
$$ \int \text{det}(\Sigma)^{-\frac{1}{2}} Q_{\Sigma}^j(y-m) (dF - dF^{*}) = 0,$$
for $j=1,\ldots,J-1$. Since
\begin{eqnarray*}
    \prod_{i_1,i_2} \left[ (\Sigma^{-1})_{i_1,i_2} (y_{i_1} - m_{i_1}) (y_{i_2} - m_{i_2}) \right]^{r_{i_1,i_2}} &=&
\left( \prod_{i_1,i_2} (\Sigma^{-1})_{i_1,i_2}^{r_{i_1,i_2}} \right) \\
&&
\times \sum_{k_1=0}^{\alpha_1} \cdots \sum_{k_p=0}^{\alpha_p} 
\prod_{i_1} \binom{\alpha_{i_1}}{k_{i_1}} y_{i_1}^{k_{i_1}} (-m_{i_1})^{\alpha_{i_1} - k_{i_1}} \\
&& \times
\sum_{l_1=0}^{\beta_1} \cdots \sum_{l_p=0}^{\beta_p} 
\prod_{i_2} \binom{\beta_{i_2}}{l_{i_2}} y_{i_2}^{l_{i_2}} (-m_{i_2})^{\beta_{i_2} - l_{i_2}},
\end{eqnarray*}
where $\alpha_1 + \cdots + \alpha_p = \beta_1 + \cdots + \beta_p = j$, and $\sum_{i_1,i_2} r_{i_1,i_2} = j$, a sufficient condition is that
$$ \int \text{det}(\Sigma)^{-\frac{1}{2}} \left( \prod_{i_1,i_2} (\Sigma^{-1})_{i_1,i_2}^{r_{i_1,i_2}} \right)  \prod_{i_1} \binom{\alpha_{i_1}}{k_{i_1}} (-m_{i_1})^{\alpha_{i_1} - k_{i_1}} \prod_{i_2} \binom{\beta_{i_2}}{l_{i_2}} (-m_{i_2})^{\beta_{i_2} - l_{i_2}} (dF - dF^{*}) = 0 ,$$
for $j=0,1,\ldots,J-1$.
According to Lemma A.1 in Ghosal and Van Der Vaart
(2001), the discrete distribution $F^{*}$ can be chosen with at most $N \lesssim J^{p^2+2p}$ support points. The remaining proof follows similar lines as the proof of Lemma D.1 of \cite{Xie2020}.

\end{proof}

Lemmas D.2, D.3, and D.4 from \cite{Xie2020} can be readily adapted to the present context. By combining these with Lemma \ref{lemma_mixture_approx_1}, we obtain the following lemma, which is analogous to Lemma 4 of \cite{Xie2020}.

\begin{lemma}
\label{lemma_mixture_approx_2}
    Assume conditions A0-A3 and B1-B5 hold. For some constant $\eta > 0$ and for all sufficiently small $\epsilon > 0$, there exists a discrete distribution 
\[
F^* = \sum_{k=1}^{N} w^*_k \delta_{(m^*_k, \Sigma^*_k)}
\]
supported on a subset of $\{(m, \Sigma) \in \mathbb{R}^p \times \mathcal{S} : \|m\|_{\infty} \leq 2a\}$ with 
\[
a = b^{- \frac{1}{2}} \left( \log \frac{1}{\epsilon} \right)^{\frac{1}{2}},
\]
such that 
\[
\|m^*_k - m^*_{k'}\|_{\infty} \geq 2\epsilon, \quad ||\Sigma^*_k - \Sigma^*_{k'}||_{op} \geq 2\epsilon \text{ whenever } k \neq k', \quad j = 1, \dots, p,
\]
and
\[
N \lesssim \left( \log \frac{1}{\epsilon} \right)^{p^2+2p}.
\]

Moreover, we have the inclusion
\[
\left\{ f_F : F = \sum_{k=1}^{N} w_k \delta_{(m_k, \Sigma_k)}, \quad (m_k, \Sigma_k) \in E_k, \quad \sum_{k=1}^{N} |w_k - w^*_k| < \epsilon \right\} \subset B \left( f_0, \eta \epsilon^{\frac{1}{2}} \left( \log \frac{1}{\epsilon} \right)^{\frac{p+4}{4}} \right),
\]
where
\[
E_k = \left\{ (m, \Sigma) \in \mathbb{R}^p \times \mathcal{S} : \|m - m^*_k\|_{\infty} < \frac{\epsilon}{2}, \quad || \Sigma - \Sigma^*_k||_{op} < \frac{\epsilon}{2} \right\}.
\]

\end{lemma}

The following lemma establishes a lower bound on the prior density for the covariance matrices.

\begin{lemma}
\label{lemma_wishart}
    Under condition B4, the prior density of covariance matrices $p_{\Sigma}(\Sigma)$ is uniformly bounded from below by some positive constant for all $\Sigma \in {\cal S}$.
\end{lemma}
\begin{proof}
The density function of the inverse-Wishart distribution is given by
\[
p_{IW}(\Sigma;\Psi,\nu) = C_{\Psi, \nu} \text{det}(\Sigma)^{-\frac{\nu + p + 1}{2}} e^{-\frac{1}{2} \operatorname{tr}(\Psi \Sigma^{-1})},
\]
where $C_{\Psi,\nu}$ is the normalizing constant. Since \( \Sigma \) has eigenvalues in \( [\underline{\sigma}^2, \overline{\sigma}^2] \), its determinant satisfies:
\[
\underline{\sigma}^{2p} \leq \text{det}(\Sigma) \leq \overline{\sigma}^{2p}.
\]
It follows that
$$ \text{det}(\Sigma)^{-\frac{\nu+p+1}{2}} \ge \overline{\sigma}^{-2p \frac{\nu+p+1}{2}} .$$
Since
\[
\sum_{i=1}^{p} \lambda_i(\Psi) \lambda_{\min}(\Sigma^{-1}) \leq \operatorname{tr}(\Psi \Sigma^{-1}) \leq \sum_{i=1}^{p} \lambda_i(\Psi) \lambda_{\max}(\Sigma^{-1}),
\]
where $\lambda_{\min}(\Sigma^{-1}) = \min_{j}\{\lambda_j(\Sigma^{-1})\}$, $\lambda_{\max}(\Sigma^{-1}) = \max_{j}\{\lambda_j(\Sigma^{-1})\}$,
consequently $e^{-\frac{1}{2} \operatorname{tr}(\Psi \Sigma^{-1})}$ is also bounded below, and the result follows.

\end{proof}




\end{appendices}

\section*{Declarations}

\bmhead{Funding}
Weipeng Huang has been supported by the Doctoral Research Initiation Program of Shenzhen Institute of Information Technology (Grant SZIIT2024KJ001).


\bibliography{sn-bibliography}

\end{document}